\documentclass[
 amsmath,amssymb,superscript,
 aps,pra,
onecolumn,
]{revtex4-2}

\usepackage{tcolorbox}
\usepackage{xcolor}
\usepackage{graphicx} 
\usepackage{dcolumn} 
\usepackage{bm} 
\usepackage{bbm} 
\usepackage{braket} 
\usepackage{amsthm}
\usepackage[normalem]{ulem}

\usepackage[ruled]{algorithm2e}

\newtheorem{theorem}{Theorem}

\newtheorem{lemma}{Lemma}
\newtheorem{proposition}{Proposition}
\newtheorem{definition}{Definition}

\newtheorem*{theorem-non}{Theorem}
\newtheorem*{lemma-non}{Lemma}
\newtheorem*{corollary-non}{Corollary}
\newtheorem*{proposition-non}{Proposition}

\definecolor{NatureTeal}{HTML}{0F6B68}
\definecolor{NatureBlue}{HTML}{1F4E79}
\definecolor{NatureRed}{HTML}{A23B3B}

\usepackage[colorlinks=true,
            citecolor=NatureTeal,
            linkcolor=NatureRed,
            urlcolor=NatureRed]{hyperref}

\usepackage{appendix}

\newcommand{\bx}{\bm{x}}

\newcommand{\calM}{\mathcal{M}}

\newcommand{\bxi}{\bm{x}^{(i)}}

\newcommand{\yi}{y^{(i)}}

\newcommand{\bomega}{\bm{\omega}}

\newcommand{\Prob}{\text{Pr}}

\newcommand{\llangle}{\langle\langle}
\newcommand{\rrangle}{\rangle\rangle}

\DeclareMathOperator{\CNOT}{CNOT}
 \DeclareMathOperator{\bUnitary}{\mathfrak{U}}
 
\DeclareMathOperator{\DKL}{D_{\text{KL}}}

\DeclareMathOperator{\Tr}{Tr}

\DeclareMathOperator{\CI}{\text{CI}}
\DeclareMathOperator{\T}{\text{T}}
 \DeclareMathOperator{\bRZ}{\mathfrak{R}_Z}

\DeclareMathOperator{\RZ}{RZ}

\DeclareMathOperator{\Unif}{Unif}

\DeclareMathOperator{\poly}{poly}

\usepackage{pifont}

\usepackage{xspace}

\newcommand{\DSE}{\ensuremath{\mathsf{DSE}}\xspace}
\newcommand{\OSE}{\ensuremath{\mathsf{OSE}}\xspace}

\begin{document}

\title{Learnable yet not simulable: a quantum resource theory of learning models}

\author{Xinbiao Wang$^{1}$}

\author{Yuxuan Du$^{1}$}
\thanks{Correspondence to:  duyuxuan123@gmail.com,\\dacheng.tao@ntu.edu.sg}

\author{Dacheng Tao$^{1}$}
\thanks{Correspondence to:  duyuxuan123@gmail.com,\\dacheng.tao@ntu.edu.sg}

\affiliation{$^1$College of Computing and Data Science, Nanyang Technological University, Singapore 639798, Singapore}

\begin{abstract}
Quantum resource theory has sharpened our understanding of the intrinsic complexity of quantum systems, particularly their classical simulability. However, it remains unclear which quantum resource governs the classical learnability of quantum circuits, especially beyond the regime of efficient classical simulation. Here we close this knowledge gap by studying the expectation-value functions of families of tunable quantum circuits, with many applications in digital quantum simulation, quantum metrology, and quantum-system characterization. Specifically, we introduce a new resource measure, the dynamical stabilizer entropy (\DSE), which quantifies how broadly an expectation-value function is distributed across its frequency modes. By relating \DSE to operator stabilizer entropy, we establish a computational phase diagram that compares classical simulators with quantum-data-assisted classical surrogates. We first determine the \DSE-dependent learnability boundary of this diagram by deriving bounds on the sample complexity and runtime of classical surrogates, and by developing a \DSE-guided surrogate. We then complete the diagram by proving, under standard complexity-theoretic assumptions, the existence of circuit families that can be efficiently learned by this surrogate but cannot be efficiently emulated from their circuit descriptions alone. Numerical experiments on random and structured circuits with up to 80 qubits support the predicted \DSE-dependent computational landscape. These results establish a quantitative resource-theoretic framework for delineating the boundary between classical simulation and learning, motivate resource measures linking quantum resources to learnability, and guide the design of learning-based algorithms for scalable quantum systems beyond the reach of direct classical simulation.
\end{abstract}

\maketitle

\section{Introduction}

Quantum resources offer quantitative criteria for probing the intrinsic complexity of quantum systems, from characterization and algorithmic advantage to classical simulability \cite{chitambar2019quantum}. Over the past decades, a broad range of such resources has been identified, with entanglement and magic serving as paradigmatic examples \cite{horodecki2009quantum,bravyi2005universal,veitch2014resource,leone2025entanglement}. Together, these resources delineate whether quantum systems can be efficiently characterized through tasks such as tomography \cite{leone2024learning,chia2024efficient, grewal2025efficient,mele2025efficient,mele2025learning} and fidelity estimation \cite{Flammia201111dirct,leone2023nonstabilizerness,hinsche2025efficient}, how quantum algorithms acquire provable advantages \cite{takagi2019operational,takagi2019general,ahnefeld2022coherence,ahnefeld2026coherence}, and whether concrete quantum problems remain accessible to classical methods \cite{seddon2021quantifying,christandl2024resource,gu2025magic,zhang2025classical}. When the relevant resources are constrained, the corresponding systems can admit efficient classical simulation, as exemplified by tensor-network methods for limited entanglement \cite{markov2008simulating,villalonga2019flexible,cirac2021matrix,pan2022simulation,berezutskii2025tensor} and near-Clifford or Pauli-based methods for restricted magic and operator structure \cite{bravyi2016improved,bravyi2016trading,bravyi2019simulation,rudolph2023classical,beguvsic2024fast}. This resource-based perspective has become especially visible in recent experiments on quantum processors with tens to hundreds of qubits \cite{kim2023evidence,google2025observation,fischer2026dynamical,tindall2026dynamics,haghshenas2026digital}, where claims of quantum advantages are tested against classical simulators tailored to the structure of the problem being solved.

Beyond these established roles, quantum resources are beginning to reshape quantum learning \cite{biamonte2017quantum,dunjko2018machine}, a regime in which quantum systems and learning models become mutually enabling. Early studies have developed along two complementary directions. One asks how quantum resources control the expressivity, trainability, and generalization of quantum learning models \cite{marrero2021entanglement,Holmes_2022,bu2022statistical,bu2023effects,Il_rio_Correr_2024,wang2024separable,bu2024complexity}. The other examines how resources enter the learning problem beyond the model itself, through entangled training data \cite{schatzki2021entangled,Sharma_2022,wang2024transition}, quantum access models or resourceful target distributions \cite{Huang_2021,buadescu2021improved,aharonov2022quantum,chen2022exponential,huang2022quantum, hinsche2023one,fanizza2022learning,zhao2025entanglement}. However, unlike classical simulability, classical learnability still lacks a resource-theoretic characterization. Recent classical surrogates sharpen this gap \cite{huang2021provably,schreiber2023classical,Sweke2025potential}. After consuming a limited number of quantum samples during training, they can later predict the behaviour of related quantum systems entirely on the classical side. Theoretical results show that such surrogates are provably efficient for circuit families with large non-Clifford content and substantial entangling structure, even when these same features make instance-wise classical simulation costly \cite{du2025efficient,liao2025demonstration}. This fundamental difference raises a central question:
\textit{What makes a quantum system learnable when it is not classically simulable?}

\begin{figure*}[t]  
\begin{center}
\centerline{\includegraphics[width=0.96\textwidth]{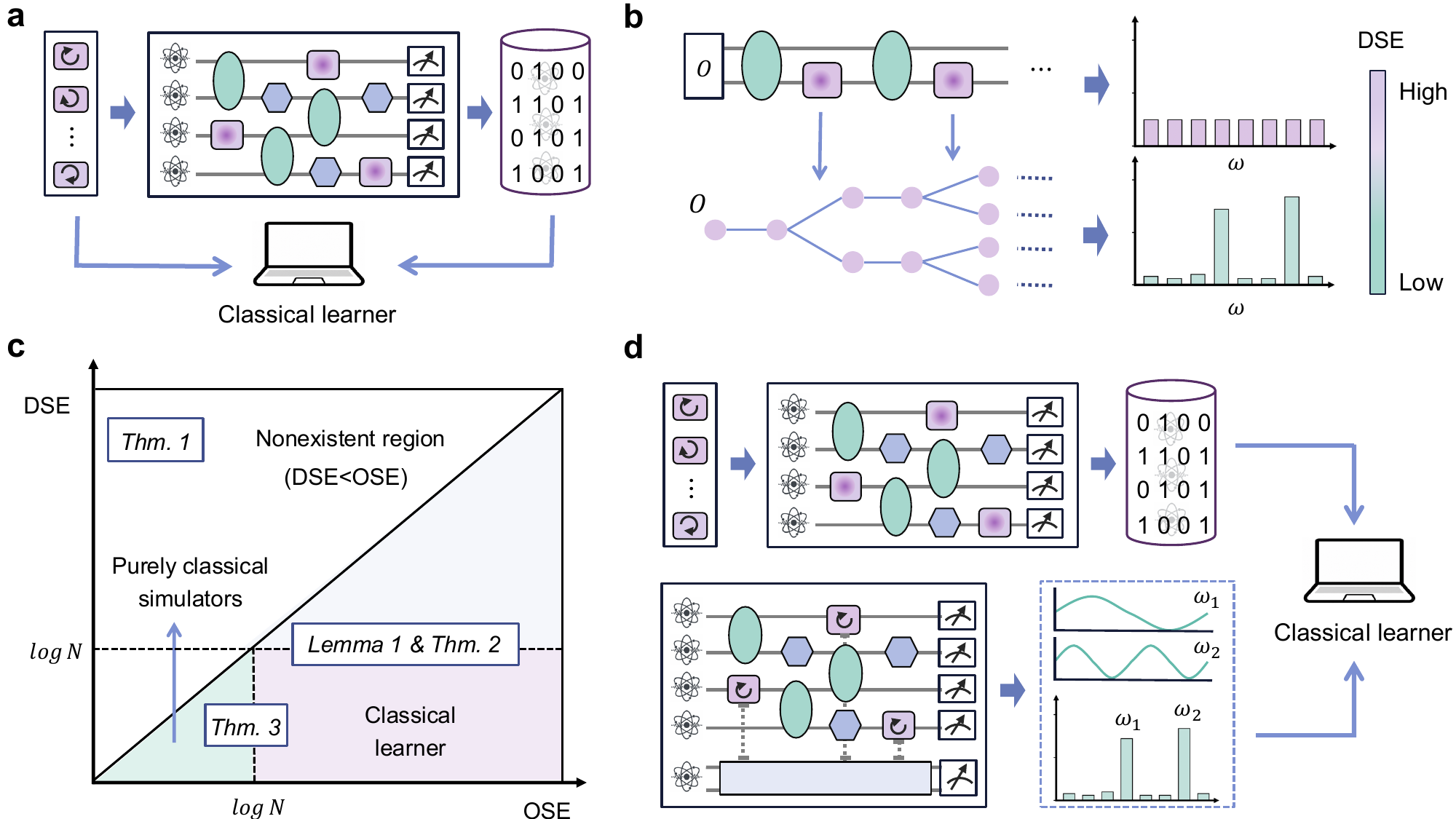}}
\caption{\small{\textbf{Schematic overview of the classical simulability and classical learnability of quantum circuits under a quantum resource framework.} \textbf{a.} The construction of classical surrogates consists of two steps. First, the classical learner collects the training dataset with labeled data from quantum devices. Then, the classical learner trains a classical machine learning model using these training data. \textbf{b.} Interpreting the operational meaning of \DSE for learnability. \DSE is defined from the trigonometric expansion of the target function, whose basis size grows exponentially with the number of tunable rotation gates. Small \DSE implies that $f(\bx)$ is concentrated on a small set of dominant trigonometric components, enabling efficient classical learning with a compact feature space. In contrast, large \DSE corresponds to a broad distribution over exponentially many components, thereby obstructing efficient classical learnability. \textbf{c.} Computational phase diagram comparing conventional classical algorithms and learning-based classical surrogates for predicting expectation values of quantum circuits. The computational capability of each method is evaluated on a circuit family with bounded resources in the worst-case setting. Conventional classical algorithms are efficient only in the low-magic regime, whereas classical surrogates remain efficient throughout the low-\DSE regime, even in a high-magic regime. The upper triangular region is excluded because \DSE is upper-bounded by a magic-resource measure in the worst case, as shown in Theorem~\ref{lem:ose_upper_bound}. \textbf{d.} The construction of the \DSE-guided classical surrogate consists of three steps. First, the classical learner collects the training dataset with
labeled data from quantum computers (as depicted in the upper panel). Second, a quantum subroutine is developed to identify the important trigonometric features contributing to \DSE (as depicted in the lower panel). Finally, the classical learner leverages the identified features on the collected training dataset to implement the classical surrogate.}} 
\label{fig:scheme}
\end{center}
\vskip -0.2in
\end{figure*}

In this study, we answer this question by establishing a computational phase diagram that delineates the capabilities of classical surrogates and classical simulators within a unified quantum resource framework, as illustrated in Fig.~\ref{fig:scheme}. Concretely, we consider the task of estimating expectation-value functions generated by $N$-qubit quantum circuits containing $d$ tunable Pauli rotation gates. Such functions underpin a wide range of applications, ranging from digital quantum simulation~\cite{feynman2018simulating,zhang2022digital} and quantum metrology~\cite{degen2017quantum,koczor2020variational} to variational quantum algorithms~\cite{cerezo2021variational_VQA} and quantum-system certification~\cite{eisert2020quantum,Flammia201111dirct}. To characterize the learnability of these functions, we introduce a new resource measure $\mathcal{M}$, termed \textit{dynamical stabilizer entropy} (\DSE), which quantifies how broadly an expectation-value function is distributed across the frequency modes generated by the tunable gates.

We first use \DSE to establish the information-theoretic boundary of classical learnability in the computational phase diagram. By deriving \DSE-dependent upper and lower bounds on the sample complexity, we show that the learnability of the explored circuit families $\mathcal{F}_{\mathcal{U}_{\lambda}}$ with bounded \DSE $\mathcal{M}_{\lambda}$ depends exponentially on $\mathcal{M}_{\lambda}$. When the bounded \DSE $\mathcal{M}_{\lambda}$ scales at most logarithmically with the system size $N$, polynomially many training samples suffice to learn all expectation-value functions in $\mathcal{F}_{\mathcal{U}_{\lambda}}$. By contrast, when $\mathcal{M}_{\lambda}$ grows superlogarithmically, there exist circuit families in $\mathcal{F}_{\mathcal{U}_{\lambda}}$ whose expectation-value functions cannot be learned sample-efficiently by any classical learner. These results provide a resource-theoretic characterization of the classically learnable regime, revealing both the capabilities and the fundamental limitations of classical surrogates.

We then characterize the remaining part of the phase diagram by comparing classical learning with classical simulation. We relate \DSE to operator stabilizer entropy (\OSE)~\cite{dowling2025magic}, an instance-wise measure of quantum magic, and show that \DSE is upper-bounded by \OSE in the worst-case setting. This relation implies that the classically simulable regime is contained within the classically learnable regime. Moreover, we develop a \DSE-guided classical surrogate $h_{\mathsf{q}}$ that uses a quantum subroutine to identify dominant frequency modes and construct a compact predictive model. Our runtime analysis shows that $h_{\mathsf{q}}$ remains computationally efficient for learning $\mathcal{F}_{\mathcal{U}_{\lambda}}$ throughout the low-\DSE regime, thereby turning the information-theoretic learnability guarantee into an implementable learning protocol. Under standard complexity-theoretic assumptions, we further show that classical algorithms without access to quantum-generated training data cannot attain comparable computational guarantees for simulating $\mathcal{F}_{\mathcal{U}_{\lambda}}$ in this regime. This establishes a rigorous computational separation between classical simulators and classical surrogates, demonstrating that learning can remain efficient beyond the reach of direct classical simulation. Together, these results complete the computational phase diagram in Fig.~\ref{fig:scheme} and provide a resource-theoretic map for determining when description-only classical simulation or quantum-data-assisted emulation is preferable.

We conduct systematic numerical experiments to validate these theoretical findings up to $80$-qubit. The numerical results reproduce the theoretically predicted computational landscape in Fig.~\ref{fig:scheme}c, with \DSE serving as a dividing line between classically learnable and non-learnable regimes, and \OSE characterizing the boundary between classically simulable and learnable regimes.

\section{Main Result}

\subsection{Problem setup}
We study mean-value estimation of an observable on states prepared by a family of parametrized quantum circuits. As illustrated in Fig.~\ref{fig:scheme}(a), for an observable $O=\sum_{l=1}^qO_l$ with $ \sum_l\|O_l\|_{\infty}\le 1$, we consider the function class
\begin{equation}\label{eq:target_function_set}
   \mathcal{F}_{\mathcal{U}} = \Big\{f(\bx)= \Tr\big(\rho_0, O(\bx;U)\big)~\big|~ U\in\mathcal{U} \Big\},
\end{equation}
where $\rho_0$ is an arbitrary $N$-qubit state, $\bx\in [-\pi,\pi]^d$, and $O(\bx;U):=U(\bx)^{\dagger}OU(\bx)$ is the Heisenberg-evolved observable. Here, we focus on a restricted circuit family $\mathcal{U}\subset \mathsf{Arc}(N,d)$, where the set $\mathsf{Arc}(N,d)$ collects all $N$-qubit circuits $U(\bx) = \prod_{j=1}^{d}(\RZ(\bx_j)V_j)$ of $G$ gates, comprising $d$ tunable rotation gates $\RZ$ interleaved with fixed blocks $V_j$ from a universal gate set~\cite{nielsen2010quantum} (e.g., $\{H, S, T, \CNOT\}$), in an arbitrary layout. Each $U$ defines one mean-value function $f(\bx)$, so $\mathcal{F}_{\mathcal{U}}$ collects these functions across all $U\in \mathcal{U}$. Such estimation tasks underlie a broad range of applications, including digital quantum simulation~\cite{feynman2018simulating,zhang2022digital}, quantum metrology~\cite{degen2017quantum,koczor2020variational}, variational quantum algorithms~\cite{cerezo2021variational_VQA}, and quantum-system certification~\cite{eisert2020quantum,Flammia201111dirct}.

The objective of mean-value estimation in Eq.~(\ref{eq:target_function_set}) is to use a classical model $h(\bx)$ to approximate $f(\bx)$ for varying $\bx$. Following standard practice, the quality of $h(\bx)$ is measured by the average prediction error, or expected risk~\cite{gao2018efficient,huang2022provably,lerch2024efficient,fontana2025classical,du2025efficient,angrisani2025classically,liao2025demonstration}, under an input distribution $\mathbb{D}$, i.e.,
\begin{equation}\label{eq:prediction_error}
	\mathsf{R}(h) = \mathbb{E}_{\bx\sim \mathbb{D}} \left| h(\bx) - f(\bx)  \right|^2.
\end{equation}
This metric applies equally to classical simulators and to classical learning surrogates (see Supplementary Information (SI)~\ref{append:sec:preliminary}). The family $\mathcal{F}_{\mathcal{U}}$ is said to be \textit{classically efficient} if, for any $f\in \mathcal{F}_{\mathcal{U}}$ and a fixed target accuracy $\epsilon$, such an $h(\bx)$ can be constructed in time polynomial in the qubit count $N$ and the input dimension $d$ while achieving $\mathsf{R}(h) \leq \epsilon$ with high probability. When $\mathbb{D}$ is a point mass, $\mathcal{F}_{\mathcal{U}}$ collapses to a single circuit instance and is efficient for classical simulators when  $\mathcal{U}$ consists of circuits containing low instance-wise quantum resources such as magic~\cite{bravyi2016improved,leone2022stabilizer} and entanglement~\cite{cirac2021matrix,markov2008simulating}. However, for a general $\mathbb{D}$, which quantum resource governs the computational efficiency of $\mathcal{F}_{\mathcal{U}}$ for a restricted class $\mathcal{U}$ remains largely open, leaving the computational separation among classical simulators and classical surrogates unclear.

\subsection{Dynamical stabilizer entropy}
To address this question, we propose a new resource measure, dubbed \textit{dynamical stabilizer entropy} (\DSE). To define \DSE formally, we first expand each target function $f(\bx)\in \mathcal{F}_{\mathcal{U}}$ for any $\mathcal{U}\subset \mathsf{Arc}(N,d)$ in the trigonometric basis via its Pauli-transfer-matrix representation~\cite{greenbaum2015introduction}, namely
\begin{equation}\label{eq:f_x_tri_expansion}
	f(\bx) = \sum_{\bomega\in \Lambda} \Phi_{\bomega}(\bx)\Tr(\rho_0 Q_{\bomega}),
\end{equation}
where the trigonometric basis at the mode $\bomega$ is $\Phi_{\bomega}(\bx)=\prod_{j=1}^d [\mathbbm{1}_{\bomega_j=0} + \cos(\bx_j), \mathbbm{1}_{\bomega_j=1}+ \sin(\bx_j), \mathbbm{1}_{\bomega_j=-1}]$, the operator $Q_{\bomega}$ is a linear combination of Pauli strings, and $\Lambda = \{-1, 0, 1\}^{d}$ is the frequency set  with $|\Lambda|=3^d$~\cite{du2025efficient}. As illustrated in Fig.~\ref{fig:scheme}b, the number of modes that contribute to $f(\bx)$, i.e., those $\bomega$ with $\Tr(\rho_0 Q_{\bomega})\neq 0$, is governed by the circuit layout $U$ in $\mathsf{Arc}(N,d)$, and can grow exponentially with $d$.

\DSE essentially quantifies how broadly $f(\bx)$ spreads over $\Lambda$. In this measure, the fixed gates $V_j$ in Eq.~\eqref{eq:target_function_set} are free operations, whereas the tunable rotations $\RZ(\bx_j)$ are resource-generating. Formally, the $\alpha$-order \DSE is
\begin{equation}\label{eq:SSE}
    \mathcal{M}^{(\alpha)}[O(\bx;U)] = \max_{V\in \mathbb{S}\mathbb{U}(2^N)} \mathcal{S}^{(\alpha)}[V^{\dagger}O(\bx;U) V],
\end{equation}
where $\mathcal{S}^{(\alpha)}[\cdot]$ is the order-$\alpha$ R\'enyi entropy of the mode distribution induced by the operator $V^{\dagger}O(\bx;U) V$. Following Eq.~(\ref{eq:f_x_tri_expansion}), its expansion is $\sum_{\bomega\in\Lambda}\Phi_{\bomega}(\bx)Q'_{\bomega}$. As a result, the induced distribution takes the form as $\mathrm{p}(\bomega)=\tilde{\mathrm{p}}(\bomega)/\sum_{\bomega\in\Lambda}\tilde{\mathrm{p}}(\bomega)$ with $\tilde{\mathrm{p}}(\bomega)=\mathbb{E}_{\bx\sim\mathbb{D}}\Phi_{\bomega}(\bx)^2\Tr(\rho_0 Q'_{\bomega})^2$, recording the average squared contribution of each mode, and $\mathcal{S}^{(\alpha)}[V^{\dagger}O(\bx;U) V]=(1-\alpha)^{-1}\log\big(\sum_{\bomega\in\Lambda}\mathrm{p}^{\alpha}(\bomega)\big)$. For brevity, we write $\mathcal{M}^{(\alpha)}[O(\bx;U)]$ as shorthand, leaving the average over $\mathbb{D}$ and the maximization over $V$ implicit when they are clear from context.

The following theorem establishes that \DSE is a valid quantum resource monotone~\cite{chitambar2019quantum} and bounds it against the operator stabilizer entropy (\OSE)~\cite{dowling2025magic}, $\mathfrak{M}_{\mathrm{ose}}^{(\alpha)}[O(\bx; U)]$, an instance-wise magic monotone for the nonstabilizerness of $O(\bx; U)$. The proofs are deferred to SI.~\ref{append:sec:properties-SSE}.
\begin{theorem}[Informal]
    \label{lem:ose_upper_bound}
    Following Eqs.~\eqref{eq:target_function_set}-\eqref{eq:SSE}, \DSE is non-negative and invariant under fixed gates. Given the class $\mathsf{Arc}(N,d)$, we have $ \max_{U\in \mathsf{Arc}(N,d)}\mathcal{M}^{(\alpha)}[O(\bx;U)]\lesssim d$,  where the symbol $\lesssim$ hides constant factors. When $d\le N$, we have $d \lesssim \max_{\bx,U\in \mathsf{Arc}(N,d)}\mathfrak{M}_{\mathrm{ose}}^{(\alpha)}[O(\bx;U)]$.
\end{theorem}
These two inequalities determine the infeasible region of the computational phase diagram in Fig.~\ref{fig:scheme}c. For a fixed number of rotation gates $d$, the worst-case \DSE scales at most as $\mathcal{O}(d)$, whereas the worst-case \OSE scales at least as $\Omega(d)$. In other words, all worst-case operating points lie on or below the diagonal $\DSE=\OSE$, leading the upper-left triangle in Fig.~\ref{fig:scheme}c infeasible. The computational separation among classical simulators and classical surrogates therefore appears in the lower-right region.

Beyond its role as a monotone, \DSE lifts the resource-theoretic characterization provided by \OSE from a single instance to an entire parametrized circuit family $\mathcal{F}_{\mathcal{U}}$. This lifting connects \DSE to related resource measures, including state stabilizer R\'enyi entropy~\cite{leone2022stabilizer}, unitary stabilizer nullity~\cite{jiang2023lower}, and circuit complexity~\cite{bu2024complexity}. It also yields a family-level notion of compressibility. Low entanglement enables tensor-network truncation~\cite{cirac2021matrix}, low \OSE enables Pauli-path truncation~\cite{rudolph2023classical}, and low \DSE enables efficient construction of Pauli-path simulators and classical surrogates across $\mathcal{F}_{\mathcal{U}}$. Similarly, dynamical properties studied through \OSE, such as the growth, locality, and scrambling of operator magic~\cite{dowling2025magic}, can be reformulated for $\mathcal{F}_{\mathcal{U}}$ as questions about how tunable circuits $U(\bx)$ generate and distribute operator magic across parameter space. In this sense, \DSE provides a family-wide diagnostic of the classical accessibility of parametrized quantum dynamics. Further discussions are provided in SI.~\ref{append:sec:properties-SSE}.

\subsection{Classical learnability of \texorpdfstring{$\mathcal{F}_{\mathcal{U}}$}{mathcalF-mathcalU} in terms of \DSE}
We now use \DSE to delineate the classical learnability of $\mathcal{F}_{\mathcal{U}_{\lambda}}$ for the circuit family $\mathcal{U}_{\lambda}$ with bounded \DSE. In particular, we consider the circuit family
\begin{equation}\label{eq:dse_circuit}
    \mathcal{U}_{\lambda}=\{U\in \mathsf{Arc}(N,d):\mathcal{M}^{(1)}[O(\bx;U)]\le \mathcal{M}_{\lambda}\}
\end{equation}
Here, $\lambda$ labels the circuit class, while $\mathcal{M}_{\lambda}$ denotes its \DSE bound. The learnability of $\mathcal{F}_{\mathcal{U}_{\lambda}}$ follows by combining a sample-complexity characterization with an explicit runtime analysis. The former gives an information-theoretic criterion for sample-efficient learning, while the latter verifies that the same \DSE threshold controls the cost of an implementable surrogate. Together, they identify where classical learning is both statistically and computationally efficient.

To set the stage, we briefly recall the classical learning protocols. Classical learners (or classical surrogates) are model-agnostic and receive the training set $\mathcal{T}=\{(\bx^{(i)},y^{(i)})\}_{i=1}^{n}$, where $\bx^{(i)}\sim\mathbb{D}$ on $[-\pi,\pi]^d$ and each label $y^{(i)}$ is estimated from $m$ incoherent measurements. After training, a classical learner runs the trained predictor $h(\bx)$ on classical hardware at the inference stage without any access to a quantum device.

The lemma below characterizes the sample complexity of classical learning $\mathcal{F}_{\mathcal{U}_{\lambda}}$ in terms of the first-order \DSE, with the formal statement and proof deferred to SI.~\ref{append:sec:information_theoretic_bound}-\ref{append:sec:dse_lower_bound}.
\begin{lemma}[Informal]\label{mt:thm:general_ML_complexity}
    Following Eqs.~\eqref{eq:target_function_set}-\eqref{eq:dse_circuit}. Let $\nu_{f} =\mathbb{E}_{\bx \sim [-\pi,\pi]^d} f(\bx)^2$ and $\epsilon< \nu_{f}/4$, and $\mathcal{U}_{\lambda}\subset \mathsf{Arc}(N,d)$ be the circuit class with bounded \DSE, i.e., $\mathcal{M}^{(1)}[O(\bx;U)]\le \mathcal{M}_{\lambda}$ for any $U\in \mathcal{U}_{\lambda}$. Given $\mathcal{T}$, the training data size
    \begin{equation}\label{mt:eq:ml_sample_complexity}
       \Omega\Big( \frac{2^{ \mathcal{M}_{\lambda}} }{ m \nu_{f} }\Big) \le n \le 
       \tilde{\mathcal{O}} \Big( \frac{2^{12 \nu_{f} \mathcal{M}_{\lambda}/\epsilon} \cdot d}{\epsilon}  \Big)
    \end{equation}
    is sufficient and necessary to achieve $\epsilon$-prediction error in Eq.~\eqref{eq:prediction_error} for a classical surrogate $h(\bx)$ with high probability. Here $h(\bx)$ has a model dimension $\tilde{\mathcal{O}}(2^{12\nu_{f} \mathcal{M}_{\lambda}/\epsilon})$.
\end{lemma}

Lemma~\ref{mt:thm:general_ML_complexity} captures the classically learnable regime of $\mathcal{F}_{\mathcal{U}_{\lambda}}$ for the circuit subset $\mathcal{U}_{\lambda}\subset \mathsf{Arc}(N,d)$ along the horizontal boundary in Fig.~\ref{fig:scheme}c. The exponential dependence of both the upper and lower bounds in Eq.~\eqref{mt:eq:ml_sample_complexity} on \DSE establishes \DSE as an exact dividing line for classical learnability. When $\mathcal{M}_{\lambda}=\max_U\mathcal{M}^{(1)}[O(\bx;U)]=\mathcal{O}(\log N)$, polynomially many samples and a polynomial-size model suffice for classical surrogates to reach $\mathsf{R}(h) \leq \epsilon$. Beyond this regime, the sample complexity grows exponentially in $N$. Together with $\mathcal{M}^{(1)}[O(\bx;U)]\lesssim d$ from Theorem~\ref{lem:ose_upper_bound}, this criterion places all families $\mathcal{F}_{\mathcal{U}}$ with $d=\mathcal{O}(\log N)$ tunable rotations on the classically learnable side and refines prior sample complexity bounds based on the number of tunable gates~\cite{leone2024learning,chia2024efficient,grewal2025efficient,bu2022statistical,bu2023effects,molteni2024exponential}. Thus, learnability is governed not by tunability alone, but by how broadly the target function spreads across frequency modes.

\noindent \textbf{Remark.} The assumption $\epsilon \le \nu_{f}/4$ in Lemma~\ref{mt:thm:general_ML_complexity} stems from the lower bound and ensures that the two bounds in Eq.~\eqref{mt:eq:ml_sample_complexity} share the same dependence on $\nu_{f}$.

\smallskip

To establish the computational side of the separation, we next construct a \DSE-guided surrogate $h_{\mathsf q}$ and analyze its runtime. Its implementation has three stages: dataset construction, model construction, and inference. The dataset construction and inference stages follow standard classical surrogates, with $\mathcal{T}$ obtained from circuit queries and $h_{\mathsf q}$ deployed on classical hardware. The main new ingredient is model construction, where a quantum subroutine identifies informative modes $\Lambda_{\mathsf q}\subset\Lambda$ in Eq.~\eqref{eq:f_x_tri_expansion}.  

The quantum subroutine identifies modes in $\Lambda_{\mathsf q}$ with non-negligible contributions to the target function $f(\bx)$. Here, we assume coherent access to the unknown circuit $U(\bx)$, $U(\bx)^{\dagger}$, and coherent operation of the tunable rotation gate during circuit execution. We emphasize that these coherent operations are only required in the training stage and still leave the classical surrogate model agnostic, as the layout of the fixed gates in $U(\bx)$ remains unknown to the surrogate. As shown in Fig.~\ref{fig:scheme}c, given this access and a $q$-sparse Pauli representation $O=\sum_{l=1}^q \bm{a}_l P_l$, the
subroutine prepares the state $\ket{\psi_{\mathsf q}}
    = \sum_{\bomega\in\Lambda} \bm{s}_{\bomega}\sqrt{\mathrm{p}(\bomega)}\ket{\bomega}$ with $\bm{s}_{\bomega}$ being a sign factor, using a $(N+d+\lceil \log q\rceil)$-qubit circuit with $\Theta(qN+G)$ gates. Measuring this state samples modes according to $\mathrm{p}(\bomega)$ in Eq.~\eqref{eq:SSE}. Since $\mathrm{p}(\bomega)$ depends on the structure-specific coefficient $\Tr(\rho_0 Q_{\bomega})$, the samples are biased toward modes that are relevant to $f(\bx)$. Repeating this procedure independently $m_f$ times yields the mode set $\Lambda_{\mathsf q}$ used by the surrogate (see SI.~\ref{append:sec:direct_mode_sampler} for details). Once $\Lambda_{\mathsf q}$ is constructed, the surrogate predicts a new input $\bx$ as
\begin{equation}\label{eq:kernel_ML}
    h_{\mathsf q}(\bx)
    =
    \frac{1}{n}\sum_{i=1}^n
    \kappa_{\mathsf q}(\bx,\bx^{(i)}) y^{(i)},
\end{equation}
where $\kappa_{\mathsf q}(\bx,\bx^{(i)}) = \sum_{\bomega\in\Lambda_{\mathsf q}}2^{\|\bomega\|_0} \Phi_{\bomega}(\bx)\Phi_{\bomega}(\bx^{(i)})$ is the partial trigonometric-monomial kernel.

The resulting efficiency guarantee is stated below, with formal statement and proof deferred to SI.~\ref{append:sec:surrogate_prediction}. 
 \begin{theorem}[Informal]\label{mt:thm:prediction_error_bound}  
 Following Eqs.~\eqref{eq:target_function_set}--\eqref{eq:kernel_ML}, it suffices to use $m_f=\widetilde{\mathcal{O}}\left( 2^{2\nu_f\mathcal{M}_{\lambda}/\epsilon} \log(2/\delta)\right)$ independent frequency draws. The quantum subroutine then returns the set $\Lambda_{\mathsf q}$ of distinct observed modes, with $d_{\mathsf q}:=|\Lambda_{\mathsf q}|\le m_f$, in expected time $\mathcal{O}\left( (qN+G)m_f\|\bm a\|_1/\sqrt{\nu_f} \right)$. A training-set size $n=4d_{\mathsf q}/(\epsilon\delta)$ then suffices to ensure $\mathsf{R}(h_{\mathsf q})\le\epsilon$ with probability at least $1-\delta$.
 \end{theorem} 
 
Theorem~\ref{mt:thm:prediction_error_bound} and Lemma~\ref{mt:thm:general_ML_complexity} together delineate the learnable regime for classical surrogates in Fig.~\ref{fig:scheme}c, where \DSE acts as the efficiency threshold. For the circuit class $\mathcal{U}_{\lambda}$ with a low-\DSE bound where $\mathcal{M}_{\lambda}=\mathcal{O}(\log N)$, the \DSE-guided surrogate $h_{\mathsf{q}}$ has polynomial sample complexity and runtime. Once \DSE grows beyond logarithmic size, the achieved sample-complexity lower bound in Lemma~\ref{mt:thm:general_ML_complexity} and runtime scaling in Theorem~\ref{mt:thm:prediction_error_bound} imply a cost exponential in $\mathcal{M}_{\lambda}$, making classical learning of $\mathcal{F}_{\mathcal{U}_{\lambda}}$ intractable.

\noindent \textbf{Remark.}  Details of the quantum subroutine are provided in Methods. For $\nu_f\leq \mathcal{O}(1/\poly(N))$, the classical learner can determine a trivial constant surrogate to achieve $\mathsf{R}(h_{\mathsf q})\leq \epsilon$ in polynomial runtime.

\begin{figure*}[t]  
\begin{center}
\centerline{\includegraphics[width=0.98\textwidth]{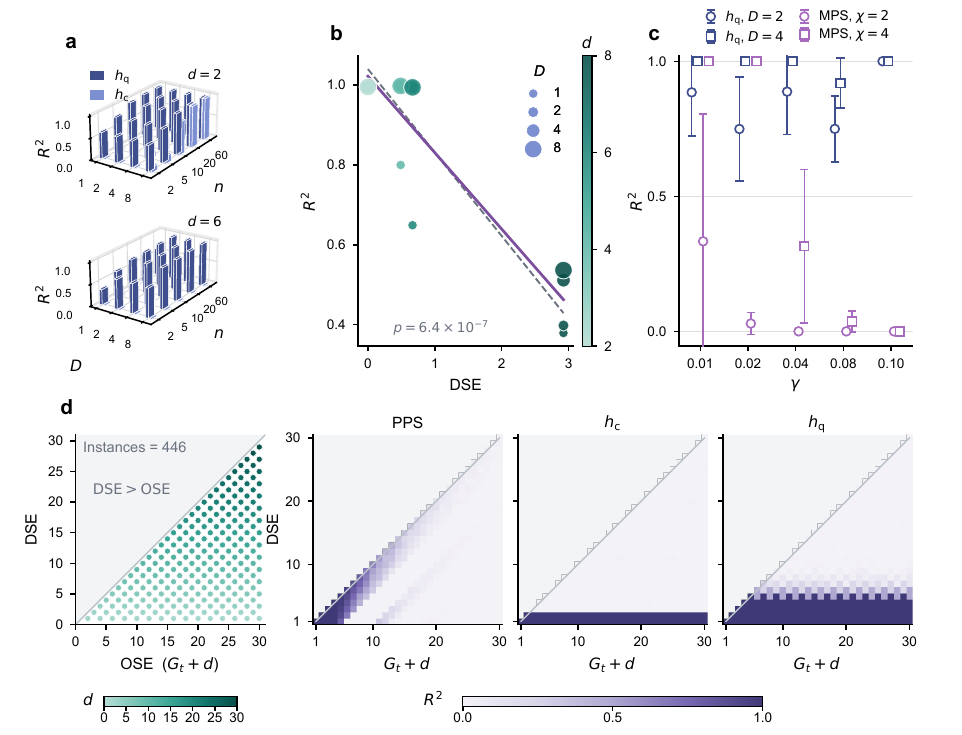}}
\caption{\small{\textbf{Numerical validation of \DSE-guided classical learnability and its separation from classical simulation.}
\textbf{a.} Prediction accuracy $R^2$ of the \DSE-guided surrogate $h_{\mathsf q}$ and the Hamming-weight surrogate $h_{\mathsf c}$ for random $N=6$ qubit circuits with depth $66$, $G_t=6$, and CNOT rate $\gamma=0.1$. The top and bottom panels correspond to $d=2$ and $d=6$ tunable rotations, respectively. The training-set size is varied over $n\in\{2,5,10,20,60\}$, and the retained feature dimension over $D\in\{1,2,4,8\}$.
\textbf{b.} Prediction accuracy of $h_{\mathsf q}$ versus the first-order mode entropy $\mathcal{S}^{(1)}[O(\bx;U)]$, used as an instance-dependent proxy for \DSE, for circuits with $d\in{2,4,6,8}$ tunable gates (as indicated by point colors) and retained feature dimensions $D\in{1,2,4,8}$ (as highlighted by point size), using a training set of size $n=20$. The solid line is a linear fit ($r=-0.916$), and the dashed line denotes a reference anticorrelation with $r=-1$; the two-sided Pearson correlation test gives $p=6.4\times10^{-7}$.
\textbf{c.} Comparison between $h_{\mathsf q}$ and a matrix-product-state (MPS) simulator as the CNOT rate is varied, with $d=G_t=6$ and $n=100$. Circles and squares denote truncation parameters $D=2$ and $D=4$, respectively. Points show the mean over random circuit layouts, and error bars denote the standard deviation.
\textbf{d.} Resource relation and prediction accuracy for a structured $N=31$ qubit circuit family with varying $d$ (indicated by color) and $G_t$, subject to $d+G_t\leq30$. Left, \DSE versus \OSE for $446$ circuit instances. The diagonal marks $\DSE=\OSE$, and the shaded upper-left region is forbidden by $\DSE>\OSE$. Right, $R^2$ heat maps for the Pauli-path simulator (PPS), $h_{\mathsf c}$, and $h_{\mathsf q}$ as functions of \DSE and $G_t+d$, using $n=400$ and truncation budget $D=16$. Values of $R^2$ below zero are displayed as zero.}}
\label{fig:exp_results_1}
\end{center}
\vskip -0.2in
\end{figure*}

\subsection{Separation between classical simulators and learners}
We last compare classical surrogates with classical simulators. According to Eq.~\eqref{eq:f_x_tri_expansion}, efficient classical simulation always entails efficient learnability, since a simulator can generate training labels for a classical model (see SI.~\ref{append:subsec:prelim-sim-learn} for details). The nontrivial question is whether the converse fails, namely whether some mean-value tasks are learnable from quantum-generated data but not simulable from their circuit descriptions alone.

The theorem below answers this question affirmatively, where the full statement and proof are deferred to SI.~\ref{append:sec:classical-hard}.
\begin{theorem}[Informal]
    \label{coro:hardness_classical_simulat}
    Consider a randomized classical algorithm $\mathcal{C}$ that does not learn from quantum-generated data. Suppose that, for any polynomial-depth circuit ensemble $\{U(\bx):\bx\in[-\pi,\pi]^d\}$ with $\mathcal{M}^{(1)}[O(\bx;U)]\le \mathcal{O}(\log N)$, $\mathcal{C}$ outputs an efficiently evaluable hypothesis $\hat{f}$
    satisfying $\mathsf{R}(\hat{f})\le \epsilon$. Then arbitrary  $\mathsf{BQP}$ computations can be simulated by a randomized polynomial-time classical algorithm, i.e., $\mathsf{BQP}\subseteq\mathsf{BPP}$.
\end{theorem}

Theorem~\ref{coro:hardness_classical_simulat} reveals that matching the low-\DSE surrogate without quantum-generated data would collapse standard quantum and classical complexity classes. The achieved result completes the phase diagram in Fig.~\ref{fig:scheme}c, i.e., in the low-\DSE regime, classical learners can outperform description-only classical simulators. Consequently, circuit families $\mathcal{F}_{\mathcal{U}}$ may still be learnable even when their individual instances are highly entangled and magic-rich, placing them beyond the reach of classical simulators, e.g., tensor-network simulators~\cite{cirac2021matrix} and near-Clifford or Pauli-propagation methods~\cite{bravyi2016improved,leone2022stabilizer}. This separates learnability from simulability and provides a resource-theoretic perspective on quantum learning theory, complementing prior approaches based on quantum data, classical shadows, and surrogate models~\cite{huang2022quantum,schreiber2023classical,jerbi2023shadows}.

The proposed \DSE-guided surrogate $h_{\mathsf q}$, together with Theorem~\ref{coro:hardness_classical_simulat}, gives a practical criterion for when a classical surrogate suffices after training. At low \DSE, finite quantum measurements can be converted into a reusable classical surrogate, enabling predictions at new settings even when individual instances remain beyond description-only simulation. This replaces repeated quantum queries with a one-time training cost in tasks that evaluate a parametrized family many times, such as digital quantum simulation across Hamiltonian parameters and evolution times~\cite{feynman2018simulating,zhang2022digital,kim2023evidence}, variational quantum algorithms~\cite{cerezo2021variational_VQA,peruzzo2014variational}, and quantum certification routines~\cite{eisert2020quantum,Flammia201111dirct}. In this regard, the achieved separation clarifies a form of quantum-data-assisted utility, where the processor generates hard-to-simulate training data, and the classical surrogate handles subsequent predictions.

\section{Numerical results}
We conduct numerical simulations on quantum circuits with up to $80$ qubits to assess the efficiency of the proposed \DSE-guided classical surrogate $h_{\mathsf q}$ and its separation in prediction performance from conventional classical surrogates and classical simulators, as established theoretically and illustrated in Fig.~\ref{fig:scheme}c. Implementation details and additional numerical results, including applications of \DSE-guided classical surrogates to the classical optimization of variational quantum algorithms (VQAs), are provided in SI.~\ref{append:sec:more-numerical-simulations}.

We first evaluate the prediction performance of $h_{\mathsf q}$ on random quantum circuits $U_{\mathrm r}$ with a fixed system size $N=6$ and circuit depth $L=66$, while varying \DSE. Each circuit contains a prescribed number $G_T$ of $T$ gates and CNOT gates inserted with probability $\gamma$, with a random gate layout. Unless otherwise stated, we set $G_T=6$ and $\gamma=0.1$. The observable is fixed as $O=Z_1\cdots Z_N$. We vary the number of tunable rotation gates as $d\in\{2,4,6,8\}$, thereby generating circuits with different values of \DSE. To construct $h_{\mathsf q}$, we consider four different feature dimensions,
$D:=|\Lambda_{\mathsf q}|\in\{1,2,4,8\},$
where $D$ also controls the computational cost of classical inference. For all methods, prediction performance is quantified by the coefficient of determination $R^2$, evaluated on a test set of $200$ samples.

 Fig.~\ref{fig:exp_results_1}a compares the \DSE-guided surrogate $h_{\mathsf q}$ with the Hamming-weight surrogate $h_{\mathsf c}$ of Refs.~\cite{du2025efficient,liao2025demonstration} for varying training-set size $n\in \{2, 5, 10, 20, 60\}$ and $D\in \{1,2,4,8\}$. For both settings of $d=\{2,8\}$, increasing the feature dimension $d$ and $n$ could constantly improve the prediction accuracy of $h_{\mathsf{q}}$. In particular, when $d=2$, $h_{\mathsf{q}}$ constructed with $d=1$ and $n=10$ suffices to achieve the perfect prediction with $R^2\approx 1$, whereas $h_{\mathsf c}$ requires $(D,n)=(8,20)$ to reach a comparable accuracy. The difference becomes sharper at $d=6$. With only $D=2$ features, $h_{\mathsf q}$ reaches $R^2=0.96$ at $n=5$, while $h_{\mathsf c}$ remains at $R^2=0$ throughout the displayed range. 
 
 We next directly test the operational role of \DSE. Fig.~\ref{fig:exp_results_1}b collects the results for varying  $d\in\{2,4,6,8\}$ and $D\in\{1,2,4,8\}$ at fixed $n=20$. Here, we employ $\mathcal{S}^{(1)}$ as a tractable proxy of \DSE. Across the resulting $16$ configurations, the prediction accuracy is strongly anticorrelated with this entropy, with Pearson coefficient $r=-0.916$ and $p=6.4\times10^{-7}$. In particular, the low-entropy instances attain $R^2\approx 1$, whereas the largest observed entropy, $\mathcal{S}^{(1)}\approx 3$, reduces the accuracy to $R^2=0.38$ even when more features are retained.
 These results accord with Theorem~\ref{mt:thm:prediction_error_bound} and highlight that the relevant features identified by the quantum subroutine provide a substantially more compact representation than a truncation based only on Hamming weight.

To compare the inference efficiency of learning with conventional simulation, we first vary the rate of inserting CNOT gates $\gamma\in \{0.01,0.02,0.04,0.08,0.1\}$ while fixing $d=6$ and $n=100$. Fig.~\ref{fig:exp_results_1}c reports averaged prediction performance of $h_{\mathsf q}$ with $D=\{2,4\}$ and a matrix-product-state (MPS) simulator at matched truncation parameters $\chi=\{2,4\}$ to ensure the inference cost of $h_{\mathsf{q}}$ less than the cost of MPS. At $D$($\chi$)$=4$, the MPS accuracy rapidly decreases from $R^2=1$ for $\gamma \leq0.02$ to $0$ at $0.1$, respectively. In contrast, $h_{\mathsf q}$ maintains $R^2\geq0.92$ over the same range. The same qualitative separation is visible at $D$($\chi$)$=2$, where $h_{\mathsf q}$ retains $R^2\geq0.75$ while the MPS result approaches zero. These results highlight the prediction advantage of $h_{\mathsf{q}}$ over the tensor-network simulation methods for highly entangled quantum circuits.

Fig.~\ref{fig:exp_results_1}d extends this comparison to structured quantum circuits of size $N=31$ with varied \DSE and \OSE. To this end, we consider the quantum circuit of form $U(\bx)=\prod_{j=1}^d \RZ_j(\bx_j) \prod_{j=d+1}^{d+G_t} T_k \prod_{j=2}^{d+G_t}\CNOT_{1,j} H_1$ with the observable $O=X_1\cdots X_{d+G_t}$, which allows us to vary the worst-case \DSE and \OSE by varying $d$ and $G_t$.
The left panel evaluates the relation between $\OSE$ and $\DSE$ and shows that all displayed instances lie in the allowed region $\DSE \leq\OSE$, in agreement with Theorem~\ref{lem:ose_upper_bound}. The remaining panels compare a Pauli-path simulator (PPS), $h_{\mathsf c}$, and $h_{\mathsf q}$ using $n=400$ and a common truncation budget $D=16$. PPS is accurate mainly when the \OSE-related quantity $G_t+d$ is small, whereas $h_{\mathsf c}$ is largely confined to the lowest-\DSE edge. By contrast, $h_{\mathsf q}$ retains a substantially broader high-accuracy region, including circuits with a large \OSE. These results indicate the computational boundary of classical simulators (as highlighted by the vertical line at $G_t+d=5$) and classical surrogates (as highlighted by the horizontal line at $\mathcal{M}^{(1)}=5$), therefore providing a numerical counterpart of the established separation depicted in Fig.~\ref{fig:scheme}c.  

\begin{figure}[t]
\centering
\includegraphics[width=0.45\textwidth]{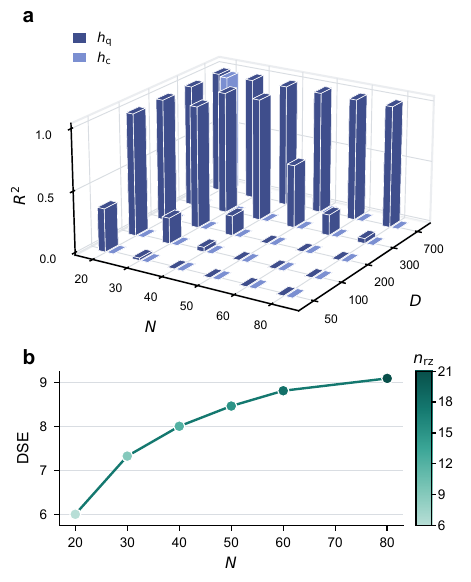}
\caption{\small{\textbf{Scaling behavior of the \DSE-guided surrogate.} \textbf{a.} Prediction accuracy $R^2$ of $h_{\mathsf q}$ and $h_{\mathsf c}$ for varying numbers of qubits, $N\in\{20,30,40,50,60,80\}$, and feature dimensions, $D\in\{50,100,200,300,700\}$, paired with training-set sizes $n\in\{100,200,300,500,700\}$, respectively. \textbf{b.} \DSE for the six system sizes considered in panel \textbf{a}, where $N\in\{20,30,40,50,60,80\}$ is paired with $d\in\{6,9,12,15,18,21\}$ tunable rotation gates (indicated by color).}}
\label{fig:scalability}
\end{figure}

Finally, we examine the scalability behavior of $h_{\mathsf q}$ at large scale. Specifically, we consider the quantum circuit 
$U(\bx)=\prod_{i=1}^{N-1}\mathrm{CZ}_{i,i+1}\prod_{i=1}^{d}T_i\RZ_i(x_i)\prod_{i=1}^{N}H_i$
with the observable $O=\frac{1}{d-k+1}\sum_{i=1}^{d-k+1}\prod_{j=i}^{i+k-1}O_j$, where $O_j=Z_{j-1}X_jZ_{j+1}$ and $k=5$. We consider six combinations of qubit size and the number of rotation gates, $(N,d)=\{(20,6), (30,9), (40,12), (50, 15), (60, 18), (80,21)\}$. As shown in Fig.~\ref{fig:scalability}b, the \DSE proxy $\mathcal{S}^{(1)}$ grows from approximately $6$ to $9$ across these instances.
Fig.~\ref{fig:scalability}a shows the prediction performance of $h_{\mathsf{q}}$ and $h_{\mathsf{c}}$ with varying feature dimension and training set size $(D,n)=\{(50,100),(100,200),(200,300),(300,500),(700,700)\}$. As the system size $N$ grows from $20$ to $80$, the required training examples for $h_{\mathsf{q}}$ to achieve $R^2\approx 1$ increase from $200$ to $700$. Importantly, this moderate increase in sample complexity accompanies the growth of $\mathcal{S}^{(1)}$. These results highlight that the prediction performance of $h_{\mathsf{q}}$ is governed primarily by \DSE rather than directly by the system size, thereby demonstrating its scalability.

\section{Discussions}
We have characterized the classical learnability of quantum circuits from a resource-theoretic standpoint. Introducing the dynamical stabilizer entropy (\DSE) as a resource measure, we proved that \DSE governs both the sample and computational efficiency of classical surrogates, and we built an explicit \DSE-guided surrogate that meets these guarantees. The diagonal boundary follows from Theorem~\ref{lem:ose_upper_bound} and marks the infeasible region where the worst-case \DSE would exceed the corresponding instance-wise magic. The horizontal boundary, established by Lemma~\ref{mt:thm:general_ML_complexity} and Theorem~\ref{mt:thm:prediction_error_bound}, separates low-\DSE families that are learnable from quantum-generated data from high-\DSE families where classical learning becomes intractable. Theorem~\ref{coro:hardness_classical_simulat} then distinguishes the lower-right region from purely classical simulation.

We further showed that the efficiently learnable regime can extend beyond classical simulability. Together, these results map out a resource-theoretic landscape linking classical surrogates and conventional simulators, sharpening the distinction between classical learnability and simulability and guiding the use of classical methods \cite{Sweke2025potential,schreiber2023classical,du2025efficient,du2025artificial,liao2025demonstration,rudolph2026pauli,shao2026characterizing} to predict quantum-circuit behavior. 

By bridging the classical learnability of quantum circuits with quantum resource theory, our work opens several directions. From a physics perspective, \DSE provides a parameter-space diagnostic of operator spreading that complements real-space measures of scrambling, such as out-of-time-order correlators~\cite{nahum2018operator,leone2022stabilizer}. Establishing its connections to these dynamical quantities, as well as to the magic and entanglement generated during circuit evolution~\cite{nahum2017entanglement,bejan2024dynamical,niroula2024phase,turkeshi2025magic}, could lead to a unified resource-theoretic picture of when quantum dynamics is learnable, simulable, or neither. 

Another important direction is to determine how noise reshapes \DSE and the performance of \DSE-guided surrogates. Characterizing how realistic noise redistributes or suppresses the frequency modes of circuit expectation functions, and how these changes propagate into learning guarantees, would extend our framework to noisy quantum devices and inform the design of robust surrogates. Moving beyond the present independent-input setting is equally important. Correlated parameters and general input distributions, as encountered in Hamiltonian simulation \cite{feynman2018simulating,zhang2022digital,kim2023evidence} and geometric quantum machine learning \cite{ragone2022representation,perrier2024quantum,wiersema2025geometric}, call for distribution-aware resource measures based on appropriate Fourier expansions of circuit expectation functions~\cite{schuld2021effect}, together with efficient procedures for extracting the relevant Fourier coefficients~\cite{barthe2025quantum} and constructing the corresponding surrogates.

Finally, extending the resource-theoretic analysis to deep-learning-based surrogates~\cite{zhu2022flexible,wu2024variational,du2025artificial} could reveal whether analogous resources govern the learnability of nonlinear model classes. Progress along these directions would deepen the connection between quantum resource theory and classical learnability and broaden the applicability of classical surrogates for predicting quantum-circuit behaviour.

\newpage

\newpage

\newpage
\bibliographystyle{unsrt}
\bibliography{apssamp}

\newpage
 
\clearpage

\onecolumngrid

\appendix 

\tableofcontents
\renewcommand{\appendixname}{SI}
 \renewcommand\thefigure{\thesection.\arabic{figure}}   
 
\bigskip

\section{Preliminary}\label{append:sec:preliminary}

\subsection{Notation}\label{append:sec:notation}
We summarize the notations used throughout this manuscript. Let $N$ denote the number of qubits and $n$ the number of training samples. For a positive integer $m$, we write $[m]=\{1,\ldots,m\}$
. Bold lowercase letters, such as $\bx
$ and $\bomega$, denote vectors. We use $\mathbb{D}$ to denote the probability distribution of the input $\bx\in \mathbb{R}^d$ and focus on the uniform distribution $\mathrm{Unif}[-\pi,\pi]^d$. The sans-serif letters $\mathsf{F}$ and $\mathsf{Q}$ label the quantum registers. We write $\mathcal{P}_N:=\{\mathbb{I},X,Y,Z\}^{\otimes N}$ for the $N$-qubit Pauli basis and $\tilde{\mathcal{P}}_N:=\frac{1}{\sqrt{2^N}}\{\mathbb{I},X,Y,Z\}^{\otimes N}$ for the normalized $N$-qubit Pauli basis. For a matrix $A$, we let $\|A\|_{\rm HS}$ denote its Hilbert-Schmidt norm, defined by $\|A\|_{\rm HS}=\sqrt{\sum_{i,j}|A_{i,j}|^2}$. For an $N$-qubit operator $A$, the notations $|A\rrangle$ and $|A))$ denote its Pauli-transfer-matrix and vectorization representations, respectively, i.e., $|A\rrangle:=[\Tr(AP)]_{P\in\tilde{\mathcal{P}}_N}$ and $|A)):=[\Tr(A\ket{i}\bra{j})]_{i,j\in[2^N]}$. We use $\lesssim$ to hide universal constant factors and $\tilde{\mathcal{O}}(\cdot)$ to hide logarithmic factors. Unless otherwise stated, logarithms are base $2$. We denote $\bm{1}_d$ as the $d$-dimensional all-ones vector $[1,1,\cdots,1]$.

\subsection{Classical simulability and learnability of quantum circuits}\label{append:subsec:prelim-sim-learn}

The classical simulability and learnability of quantum circuits may refer to different computational targets. One line of work concerns reproducing the \emph{output state} of the circuit, either by sampling from its measurement distribution or by outputting a classical description of the state itself. In this work, we focus on a different and operationally more relevant target, namely estimating the \emph{mean value} of a specified observable, since this is the quantity accessed in most applications of parametrized quantum circuits, such as digital quantum simulation, variational quantum algorithms, and quantum-system certification.

Recall from Eq.~\eqref{eq:target_function_set} that the object of interest is the class of mean-value functions
\begin{equation}\label{append:eq:target_recall}
    \mathcal{F}_{\mathcal{U}} = \Big\{f(\bx)= \Tr\big(\rho_0\, O(\bx;U)\big)~\big|~ U\in \mathcal{U} \subset \mathrm{Arc}(N,d) \Big\},
\end{equation}
where each circuit layout $U\in\mathcal{U}$ induces one mean-value function $f(\bx)$ with $\bx\in[-\pi,\pi]^d$, and $\mathcal{F}_{\mathcal{U}}$ collects these functions across all $U\in\mathcal{U}$. Classically predicting a function $f\in\mathcal{F}$ can be attempted by two distinct computational models: description-only classical simulators and learning-based classical surrogates. In this subsection, we formalize each model together with the notion of efficiency attached to it, providing the definitions on which the subsequent analysis relies. Throughout, the quality of a predictor $h(\bx)$ is measured by the expected risk under the input distribution $\mathbb{D}$,
\begin{equation}\label{append:eq:risk}
    \mathsf{R}(h) = \mathbb{E}_{\bx\sim\mathbb{D}}\,\big|h(\bx)-f(\bx)\big|^2 .
\end{equation}
In both models, efficiency is required to hold \emph{uniformly over the whole class}, i.e., for every target function $f\in\mathcal{F}_{\mathcal{U}}$.

\medskip
\noindent\textbf{Classical simulability (description-only simulators)}. A \emph{description-only classical simulator} is a classical algorithm that receives the gate description of the circuit $U(\bx)$, i.e., the fixed blocks $\{V_j\}$, the layout, the input state $\rho_0$, and the observable $O$, but is given \emph{no} data generated by a quantum device. From this description alone, it evaluates an estimate of the mean-value function.
\begin{definition}[Efficient classical simulability]\label{append:def:simulability}
    The function class $\mathcal{F}_{\mathcal{U}}$ generated from the circuit family $U\in \mathcal{U}$ in Eq.~\eqref{append:eq:target_recall} is \emph{efficiently classically simulable} if there exists a randomized classical algorithm $\mathcal{C}$ such that, for \emph{every} $f\in\mathcal{F}_{\mathcal{U}}$, given the gate description of $U(\bx)$ and the input $\bx$, the algorithm $\mathcal{C}$ outputs an estimate $\hat{f}(\bx)$ of $f(\bx)$ with $\mathsf{R}(\hat{f})\le \epsilon$ with high probability, in time $\mathrm{poly}(N,d)$.
\end{definition}
This is the model instantiated by the algorithm $\mathcal{C}$ in Theorem~\ref{coro:hardness_classical_simulat}, where efficiency is required for the $q$-sparse Pauli observable $O$ up to constant additive error. Its capability is dictated by instance-wise quantum resources. In particular, leading simulators based on tensor networks~\cite{cirac2021matrix,markov2008simulating}, stabilizer rank~\cite{bravyi2016improved,bravyi2019simulation}, and Pauli propagation~\cite{rudolph2023classical,beguvsic2024fast} run efficiently only when the entanglement or magic of the individual instance $O(\bx;U)$ is suitably bounded. Crucially, a description-only simulator has no access to the outcomes of physical measurements on the target system.

\medskip
\noindent\textbf{Classical learnability (learning-based surrogates)}. A \emph{learning-based classical surrogate} is trained on a finite dataset generated by a quantum device and, once trained, is deployed entirely on classical hardware. For a target function $f\in\mathcal{F}$, the training set is
\begin{equation}\label{append:eq:train-set}
    \mathcal{T} = \big\{(\bx^{(i)},y^{(i)})\big\}_{i=1}^{n},\qquad \bx^{(i)}\sim\mathbb{D},
\end{equation}
where each label $y^{(i)}$ is an empirical estimate of $f(\bx^{(i)})$ obtained from $m$ incoherent measurements of the evolved observable $O(\bx^{(i)};U)$ on a quantum computer. From $\mathcal{T}$, the learner constructs a predictor $h(\bx)$ that is subsequently evaluated on classical hardware without further quantum queries.
\begin{definition}[Efficient classical learnability]\label{append:def:learnability}
    The function class $\mathcal{F}$ in Eq.~\eqref{append:eq:target_recall} is \emph{efficiently classically learnable} if there exists a classical learning algorithm such that, for \emph{every} $f\in\mathcal{F}$, given a training set $\mathcal{T}$ of size $n=\mathrm{poly}(N,d)$ generated from $f$, the algorithm outputs a predictor $h(\bx)$ with $\mathsf{R}(h)\le\epsilon$ with high probability, where both the training and the per-input evaluation of $h$ run in time $\mathrm{poly}(N,d)$.
\end{definition}
The defining distinction from Definition~\ref{append:def:simulability} is the \emph{source of information}, namely, a classical learner consumes quantum-generated labels $\{y^{(i)}\}$ without knowing the circuit layout, whereas a classical simulator works from the circuit description only.


\subsection{Quantum resource theory}\label{append:subsec:prelim-resource-theory}

In this subsection, we recall the elements of quantum resource theory that underpin our analysis, with an emphasis on the resource of \emph{magic} or nonstabilizerness. We first outline the general structure of a resource theory, then review the stabilizer formalism that fixes the free objects and free operations for magic, and finally collect the magic monotones that are relevant to this work, culminating in the operator stabilizer entropy invoked in the proof of Theorem~\ref{lem:ose_upper_bound}.

\medskip
\noindent\textbf{General structure of a resource theory}. A quantum resource theory~\cite{chitambar2019quantum} formalizes a physical setting in which certain states or operations are freely available while others are costly and constitute the \emph{resource}. It is specified by three ingredients: (i) a set of \emph{free objects} $\mathcal{O}_{\mathrm{free}}$ (e.g., states, unitaries, or observables) that can be prepared or implemented at no cost; (ii) a set of \emph{free operations} $\mathbb{F}$ that map free objects to free objects; and (iii) a \emph{resource measure} $\mathcal{R}(\cdot)$, quantifying the resource content of an object. A function $\mathcal{R}$ qualifies as a \emph{resource monotone} if it satisfies two minimal requirements:
\begin{itemize}
    \item[(R1)] \emph{Non-negativity and vanishing on free objects}: $\mathcal{R}(\cdot)\ge 0$ and $\mathcal{R}(\sigma)=0$ for every $\sigma\in\mathcal{O}_{\mathrm{free}}$;
    \item[(R2)] \emph{Monotonicity under free operations}: $\mathcal{R}(V(\sigma))\le \mathcal{R}(\sigma)$ for every free operation $V\in\mathbb{F}$ and every object $\sigma$.
\end{itemize}
Entanglement~\cite{horodecki2009quantum}, coherence~\cite{baumgratz2013quantifying}, and magic~\cite{veitch2014resource,howard2017application} are archetypal resource theories obtained by different choices of $(\mathcal{O}_{\mathrm{free}},\mathbb{F})$. However, these resources are defined for individual quantum objects.

\medskip
\noindent\textbf{Stabilizer formalism and magic}. The resource theory of magic~\cite{veitch2014resource,howard2017application} identifies the classically tractable sector of quantum computation with the \emph{stabilizer} formalism~\cite{gottesman1997stabilizer}, and treats any departure from it as the resource enabling universal, potentially hard-to-simulate, quantum computation. We briefly recall the relevant notions.

The $N$-qubit Pauli group is $\mathcal{P}_N=\{\mathbb{I},X,Y,Z\}^{\otimes N}$ up to phases, and the Clifford group $\mathcal{C}_N$ is its normalizer, i.e., the set of unitaries $V$ such that $V P V^{\dagger}\in\mathcal{P}_N$ for all $P\in\mathcal{P}_N$. The \emph{stabilizer states} are those obtained from $\ket{0}^{\otimes N}$ by Clifford unitaries, and the \emph{stabilizer operations} comprise Clifford gates, Pauli measurements, and classical control. By the Gottesman-Knill theorem~\cite{gottesman1998heisenberg}, any circuit built solely from stabilizer operations acting on stabilizer inputs can be simulated classically in polynomial time. Consequently, in the resource theory of magic, one designates
\begin{equation}\label{append:eq:magic-free}
    \mathcal{O}_{\mathrm{free}} = \{\text{stabilizer states/operators}\},\qquad \mathbb{F} = \{\text{stabilizer (Clifford) operations}\},
\end{equation}
so that magic or the nonstabilizerness of a state or operator is precisely the resource that must be injected (e.g., via non-Clifford gates such as $T$) to reach universal quantum computation. Magic is a central quantity because it governs the runtime of leading classical simulators of quantum circuits, including stabilizer-rank~\cite{bravyi2016improved,bravyi2019simulation} and Pauli-propagation~\cite{rudolph2023classical,beguvsic2024fast} methods, whose cost scales with the amount of magic rather than the circuit size alone.

\medskip
\noindent\textbf{Operator stabilizer entropy}. Many advanced resource measures have been proposed for quantifying the non-stabilizers of a quantum object. When the object is the Heisenberg-evolved operator $O(\bx;U)=U(\bx)^{\dagger}OU(\bx)$, the relevant instance-wise monotone is the \emph{operator stabilizer entropy} (\OSE) introduced in Ref.~\cite{dowling2025magic}, which lifts the stabilizer Rényi entropy \cite{leone2022stabilizer} from states to operators via the Pauli decomposition. Expanding a traceless-normalized $N$-qubit operator $A$ in the Pauli basis as $A=\sum_{P\in\mathcal{P}_N}c_P P$ with $c_P=\Tr(AP)/2^N$, one defines the operator Pauli-weight distribution
\begin{equation}\label{append:eq:ose-dist}
    \Xi_P(A) = \frac{c_P^2}{\sum_{P'\in\mathcal{P}_N}c_{P'}^2},
\end{equation}
and the order-$\alpha$ \OSE\ as its Rényi entropy,
\begin{equation}\label{append:eq:ose}
    \mathfrak{M}_{\mathrm{ose}}^{(\alpha)}[A] = \frac{1}{1-\alpha}\log\Big(\sum_{P\in\mathcal{P}_N}\Xi_P(A)^{\alpha}\Big).
\end{equation}
Intuitively, $\mathfrak{M}_{\mathrm{ose}}^{(\alpha)}[A]$ measures how broadly the operator $A$ is spread x the Pauli basis. In particular, a Pauli string or, more generally, a stabilizer-supported operator has a sharply concentrated spectrum and hence small \OSE, whereas a magic-rich operator distributes its weight over many Paulis and attains large \OSE. The \OSE\ inherits the defining properties of a magic monotone from the operator viewpoint: it is non-negative, invariant under Clifford conjugation $A\mapsto V^{\dagger}AV$ with $V\in\mathcal{C}_N$, and---crucially for our proof---\emph{additive under tensor products},
\begin{equation}\label{append:eq:ose-additive}
    \mathfrak{M}_{\mathrm{ose}}^{(\alpha)}[A\otimes B] = \mathfrak{M}_{\mathrm{ose}}^{(\alpha)}[A] + \mathfrak{M}_{\mathrm{ose}}^{(\alpha)}[B].
\end{equation}
This additivity is what allows the worst-case \OSE\ of a $d$-rotation circuit to be lower-bounded by a sum of $d$ single-qubit contributions in SI.~\ref{append:subsec:SSE-OSE}. 


 \subsection{Pauli transfer matrix and the trigonometric expansion of quantum circuits}\label{append:subsec:trigo-monomial-exp-QC}

\noindent \textbf{Pauli transfer matrix.} We first review the Pauli--Liouville representation of quantum states and observables. Let $P_l\in\frac{1}{\sqrt{2^N}}\{\mathbb I,X,Y,Z\}^{\otimes N}$ denote the $l$-th normalized Pauli operator, so that $\llangle P_l|P_k\rrangle=\Tr(P_lP_k)=\delta_{lk}$. In this normalized Pauli basis, an $N$-qubit observable $O$ is represented by the $4^N$-dimensional vector
\begin{equation}
    |O\rrangle=\bigl[\Tr(OP_1),\cdots,\Tr(OP_{4^N})\bigr]^{\top},
\end{equation}
and a quantum state $\rho$ by $|\rho\rrangle=[\Tr(\rho P_1),\cdots,\Tr(\rho P_{4^N})]^{\top}$.

A unitary can likewise be represented in this basis. For a circuit $U(\bx)$, its Pauli transfer matrix (PTM)~\cite{hantzko2025fast}, denoted $\bUnitary(\bx)$, is
\begin{equation}\label{eqn:append:PTM}
	[\bUnitary(\bx)]_{jk}=\llangle P_j|\bUnitary(\bx)|P_k\rrangle=\Tr\bigl(P_jU(\bx)P_kU(\bx)^{\dagger}\bigr).
\end{equation}
For instance, the single-qubit rotation gate $\RZ(\bx_j)$ has the PTM
\begin{equation}
    \mathsf{R_Z}=\begin{pmatrix}
		1 & 0 & 0 & 0\\
		0 & \cos(\bx_j) & -\sin(\bx_j) & 0\\
		0 & \sin(\bx_j) & \cos(\bx_j) & 0\\
		0 & 0 & 0 & 1
	\end{pmatrix}
    =\mathsf{D}_0+\cos(\bx_j)\,\mathsf{D}_1+\sin(\bx_j)\,\mathsf{D}_{-1},
\end{equation}
where
\begin{equation}\label{append:eq:ptm_basis_noiseless}
	\mathsf{D}_0=\begin{pmatrix}1&0&0&0\\0&0&0&0\\0&0&0&0\\0&0&0&1\end{pmatrix},\quad
    \mathsf{D}_1=\begin{pmatrix}0&0&0&0\\0&1&0&0\\0&0&1&0\\0&0&0&0\end{pmatrix},\quad
    \mathsf{D}_{-1}=\begin{pmatrix}0&0&0&0\\0&0&-1&0\\0&1&0&0\\0&0&0&0\end{pmatrix}.
\end{equation}

\medskip

\noindent \textbf{Trigonometric expansion of $\RZ+\CI$ quantum circuits.} The family of $\RZ+\CI$ circuits has been widely studied owing to its universality. An $N$-qubit $\RZ+\CI$ circuit takes the form
\begin{equation}\label{eqn:append:circuit-ideal}
	U(\bx)=\prod_{j=1}^{d}\RZ(\bx_j)V_j,
\end{equation}
where each $V_j$ is a fixed Clifford circuit composed of gates from $\{H,S,\CNOT\}$.

We now derive the trigonometric expansion of such circuits in the Heisenberg picture. For a normalized Pauli observable $P\in\frac{1}{\sqrt{2^N}}\{\mathbb I,X,Y,Z\}^{\otimes N}$, the evolved observable is
\begin{equation}\label{append:eq:ptm_state_noiseless}
	P(\bx;U)=U(\bx)^{\dagger}PU(\bx)=\sum_{\bomega\in\Lambda}\Phi_{\bomega}(\bx)\,P_{\bomega},
\end{equation}
where $\Omega=\{0,1,-1\}^d$ is the frequency set and $\Phi_{\bomega}(\bx)$ is the trigonometric monomial
\begin{equation}\label{append:eq:tri_featrue}
	\Phi_{\bomega}(\bx)=\prod_{j=1}^d\begin{cases}
		1 & \textnormal{if}~\bomega_j=0,\\
		\cos(\bx_j) & \textnormal{if}~\bomega_j=1,\\
		\sin(\bx_j) & \textnormal{if}~\bomega_j=-1.
	\end{cases}
\end{equation}
Here $P_{\bomega}$ is the (possibly signed) normalized Pauli operator in $\frac{1}{\sqrt{2^N}}\{\mathbb I,X,Y,Z\}^{\otimes N}$ obtained by replacing the $j$-th $\RZ$ gate with $\mathsf{D}_0$, $\mathsf{D}_1$, or $\mathsf{D}_{-1}$ from Eq.~\eqref{append:eq:ptm_basis_noiseless} according to whether $\bomega_j=0,1,$ or $-1$. Concretely, let $P_{\bomega_{1:j-1}}$ be the operator obtained after the first $j-1$ layers; applying the $j$-th $\RZ$ gate updates its PTM representation as
\begin{align}
    |P_{\bomega_{1:j}}\rangle\rangle=\mathbbm{1}_{\bomega_j=0}\,\mathsf{D}_0|P_{\bomega_{1:j-1}}\rangle\rangle
    +\mathbbm{1}_{\bomega_j=1}\cos(\bx_j)\,\mathsf{D}_1|P_{\bomega_{1:j-1}}\rangle\rangle
    +\mathbbm{1}_{\bomega_j=-1}\sin(\bx_j)\,\mathsf{D}_{-1}|P_{\bomega_{1:j-1}}\rangle\rangle,
\end{align}
with $\bomega_{1:j}=(\bomega_1,\cdots,\bomega_j)$. This construction extends to a general observable $O=\sum_l \bm{a}_l P_l$, for which
\begin{equation}
    O(\bx;U)=U(\bx)^{\dagger}OU(\bx)=\sum_{\bomega\in \{0,\pm 1\}}\Phi_{\bomega}(\bx)\Bigl(\sum_l \bm{a}_l P_{l,\bomega}\Bigr),
\end{equation}
where $P_{l,\bomega}$ is the Pauli operator obtained by evolving the term $P_l$ along the branch $\bomega$.

Evaluating $O(\bx;U)$ on a state $\rho$, the expectation value inherits the trigonometric expansion
\begin{equation}\label{eqn:append:expectation-idea}
	f(\bx)\equiv\Tr\bigl[\rho\,O(\bx;U)\bigr]=\sum_{\bomega\in \{0,\pm 1\}}\Phi_{\bomega}(\bx)\,\alpha_{\bomega},
\end{equation}
where $\alpha_{\bomega}=\sum_l \bm{a}_l\Tr(\rho P_{l,\bomega})$.
This representation underlies a broad class of classical emulators for quantum circuits, including Pauli-propagation simulators and learning-based surrogates. Such methods approximate $f(\bx)$ by estimating the coefficients $\Tr(\rho P_{l,\bomega})$ over a polynomial-size subset of frequencies $\{\bomega\}\subset \{0,1,-1\}^{d}$.

\medskip

\noindent \textbf{Trigonometric expansion of $\RZ+\CI+\T$ quantum circuits.} Despite the universality of $\RZ+\CI$ circuits, this model attributes all non-stabilizerness to the arbitrary-angle $\RZ$ gates. From the viewpoint of circuit synthesis and fault-tolerant implementation, however, an arbitrary-angle $\RZ(\bx)$ must generally be decomposed into Clifford and $\T$ gates, with the $\T$ gates carrying the essential discrete non-Clifford cost. This motivates the $\RZ+\CI+\T$ setting, which incorporates both continuous and discrete non-Clifford resources and thus provides a more realistic framework for studying the tradeoff between non-stabilizerness and circuit complexity.

In this setting, for the circuit $U(\bx)=\prod_{j=1}^d\RZ(\bx_j)V_j$ of Eq.~\eqref{eqn:append:circuit-ideal}, each fixed block $V_j$ now comprises Clifford and $\T$ gates. For a Pauli observable $P$, the evolved observable expands as
\begin{equation}\label{apppend:eq:sPx_withT}
    P(\bx;U)=U(\bx)^{\dagger}PU(\bx)=\sum_{\bomega\in \{0,\pm 1\}^d}\Phi_{\bomega}(\bx)\,Q_{\bomega}',
\end{equation}
where $Q_{\bomega}'=\sum_{k=1}^{s_{\bomega}}c_{\bomega,k}P_{\bomega,k}$ is a linear combination of Pauli terms. The number of terms is bounded by the $\T$-count $G_T$, namely $s_{\bomega}\le 2^{G_T}$. When $G_T=0$, this reduces to the $\RZ+\CI$ case, in which each $Q_{\bomega}'$ collapses to a single Pauli operator $P_{\bomega}$ as given in Eq.~\eqref{append:eq:ptm_state_noiseless}. Likewise, for a general observable $O=\sum_l \bm{a}_l P_l$,
\begin{equation}\label{apppend:eq:Px_withT}
    O(\bx;U)=U(\bx)^{\dagger}OU(\bx)=\sum_{\bomega\in \{0,\pm 1\}}\Phi_{\bomega}(\bx)\Bigl(\sum_l \bm{a}_l Q_{l,\bomega}'\Bigr):=\sum_{\bomega\in \{0,\pm 1\}}\Phi_{\bomega}(\bx)\,Q_{\bomega},
\end{equation}
where $Q_{\bomega}:=\sum_l \bm{a}_l Q_{l,\bomega}'$ and $Q_{l,\bomega}'$ is the linear combination of Paulis obtained by evolving the term $P_l$ along the branch $\bomega$ as in Eq.~\eqref{apppend:eq:sPx_withT}. The trigonometric expansion of the expectation value on a state $\rho$ is then given by
\begin{equation}\label{append:eq:tri_exp_T}
    f(\bx)=\sum_{\bomega\in \{0,\pm 1\}}\Phi_{\bomega}(\bx)\,\Tr(\rho_0 Q_{\bomega}).
\end{equation}

\subsection{Literature review}  

In this section, we review prior works relevant to this study, which can be broadly classified into four categories: classical simulability of quantum circuits, learnability of quantum circuits, learning-based classical surrogates, and resource-theoretic characterizations of quantum advantage. In the following, we review each of these research directions separately and clarify how our work differs from them.

\medskip

\noindent \textbf{Classical simulability of quantum circuits.}
Classical simulability of quantum circuits refers to the ability of a classical algorithm to reproduce the output statistics of a quantum circuit, or to estimate its relevant observables, using computational resources that scale polynomially with the circuit size. For generic quantum circuits, such a task is believed to be classically intractable in the worst case, as efficient classical simulation would contradict standard complexity-theoretic assumptions. Nevertheless, many advanced classical algorithms have been developed to efficiently simulate quantum circuits with particular structures. For instance, tensor-network-based methods can efficiently simulate quantum systems with limited entanglement or small contraction complexity~\cite{shi2006classical,markov2008simulating,singh2010tensor,biamonte2017tensor,hauschild2018efficient,pang2020efficient,causer2023optimal,patra2024efficient}. Another important class of methods exploits stabilizer structure and limited magic~\cite{gottesman1998heisenberg,aaronson2004improved,nest2008classical,bravyi2016improved,lerch2024efficient,beguvsic2025simulating}, including stabilizer-rank, quasiprobability, Pauli-path simulation, and Clifford perturbation techniques. More general group-theoretic approaches have also been proposed for simulating quantum systems whose dynamics are restricted by symmetries or low-dimensional Lie-algebraic structures~\cite{somma2005quantum,somma2006efficient,galitski2011quantum,goh2023lie,anschuetz2023efficient}.

These structures can often be understood through the lens of quantum resource theory \cite{chitambar2019quantum}, as restrictions on the quantum resources that underlie computational advantage. Various resource metrics \cite{horodecki2009quantum,bravyi2016trading,howard2017application,wang2019quantifying,leone2022stabilizer,dowling2025magic} have been proposed to quantify these resources and to characterize the classical simulability of quantum circuits quantitatively. For example, stabilizer-based simulators typically have runtimes that scale exponentially with the number of $\T$ gates, or more generally with a magic-related quantity, whereas tensor-network-based methods are governed by entanglement-related quantities such as bond dimension or contraction complexity. In contrast, the present work studies the learnability of quantum circuits, where the relevant object is not a single circuit instance but the output function generated by an entire parametrized circuit family. This distinction motivates the introduction of \DSE as a family-level resource measure for classical learnability.

\medskip

\noindent \textbf{Learnability of quantum circuits.}
Learnability of quantum circuits concerns whether one can infer a predictive model from a feasible number of input--output samples such that it accurately generalizes to previously unseen inputs. It is a central problem in quantum learning theory~\cite{anshu2024survey,banchi2023statistical}. A substantial body of work has investigated this question by characterizing the sample and computational complexity of learning quantum circuits. In particular, Refs.~\cite{sharma2022reformulation,wang2024transition,wang2024separable} establish sample-complexity lower bounds for learning general quantum circuits from a no-free-lunch perspective, highlighting the potential role of quantum data resources in overcoming these limitations. Ref.~\cite{zhao2023learning} further studies bounded-gate quantum circuits and characterizes both their sample and computational complexity. It shows that, although a number of samples linear in the gate count is both necessary and sufficient for learning such circuits, the corresponding computational cost can scale exponentially with the number of gates. 
Efficient learning algorithms have also been developed for circuit families with additional restrictions on the quantum circuits or input states. Ref.~\cite{huang2023learning} introduces an efficient classical-shadow-based method for predicting observable expectation values of arbitrary quantum circuits when the input states are drawn from a locally scrambled distribution. Ref.~\cite{huang2024learning} develops a polynomial-time classical algorithm for learning the description of an unknown shallow quantum circuit. In addition, Ref.~\cite{jerbi2023power} investigates the power and limitations of learning unitary dynamics with quantum neural networks under both coherent and incoherent measurement settings.

Our work differs from these studies in three main respects. First, we consider a different learning setting. Whereas previous works typically take quantum states as training inputs, we consider a fixed input state and treat the tunable rotation angles as classical inputs. Second, while prior works primarily study quantum learners whose learned models require access to quantum computational resources at the inference stage, we focus on classical learners that yield classically efficient representations and can be evaluated entirely classically after training. Third, we establish matching sample- and computational-complexity bounds in terms of the proposed resource measure \DSE. These bounds do not rely on restrictions on the circuit depth or on the choice of input state.

\medskip

\noindent \textbf{Learning-based classical surrogates.} Learning-based classical surrogates, also referred to as classical learners, classical agents, or machine-learning models, are classical models trained to predict properties of quantum systems from data. Their defining feature is that the learned model admits an efficient classical representation and can be evaluated entirely classically at the inference stage. Unlike conventional classical simulators, which predict quantum-circuit outputs directly from a circuit description, learning-based surrogates can exploit quantum-generated training data during the training phase. This additional access to data can enable accurate prediction even in regimes where direct classical simulation is challenging. Existing approaches can be divided into heuristic deep learning-based surrogates and traditional machine-learning surrogates. In particular, deep learning-based surrogates involve training deep neural networks to fit measurement data generated by quantum circuits. Owing to the flexibility and expressive power of neural architectures, such methods have been applied to predict a range of quantum properties from different types of data~\cite{wu2023quantum,qin2024experimental,qian2024multimodal,tang2024towards,wang2022predicting,zhao2025rethink}. Although these approaches have achieved promising empirical performance, they are largely heuristic and generally lack rigorous guarantees on their sample or computational efficiency.

The machine learning-based surrogates instead construct explicit classical learning models with carefully designed feature maps, which enable analytical efficiency guarantees. Existing works mainly exploit various basis expansions of observable expectation-value functions of quantum circuits to design the feature map, including Fourier expansions~\cite{schuld2021effect} and trigonometric expansions~\cite{du2025efficient}. These two representations lead to distinct surrogate constructions. Fourier-based approaches formulate the target function as a linear model over Fourier features \cite{fontana2022efficient,schreiber2023classical,nemkov2023fourier,landman2022classically,gan2024concept} and employ techniques such as random Fourier features \cite{sweke2023potential} to reduce the effective model dimension. Trigonometric-based approaches construct a linear model over trigonometric features and retain only basis functions with low Hamming weight \cite{du2025efficient, liao2025demonstration}. In both cases, the construction of classical surrogates involves estimating the coefficients of the selected classical feature map by solving a linear regression problem on the quantum-generated training data. Under their respective structural assumptions, these surrogates can achieve efficient learning when the circuit contains a small number of tunable rotation gates or when the target expectation-value function has a small gradient norm.

Despite these advances, existing surrogate constructions do not explicitly exploit the circuit structure to identify the features that are most relevant to the target prediction task. Consequently, their feature-selection procedures can be suboptimal. In contrast, the \DSE-guided classical surrogate proposed in this work uses \DSE to characterize the effective feature complexity for the prediction task and achieves nearly optimal sample and computational complexity. Moreover, the low-\DSE regime that guarantees the efficiency of our surrogate contains the regimes considered in prior work, including quantum circuits with a small number of rotation gates or small gradient norms, as the derived complexity is tight in terms of the \DSE. This improvement is enabled by our developed quantum subroutine that identifies the dominant feature maps during training, allowing the surrogate to focus its classical representation on the features that contribute most substantially to the target function.

\medskip


\noindent \textbf{Combination of quantum learning theory and quantum resource theory.}
Quantum resource theory provides a systematic framework for identifying and quantifying physical resources that can underlie quantum advantages~\cite{chitambar2019quantum}. Beyond its central role in characterizing classical simulability, resource-theoretic ideas have recently been incorporated into quantum learning theory from several complementary perspectives.

One line of work investigates how resources contained in quantum learning models affect their expressivity, trainability, and generalization performance. For example, entanglement, magic, and coherence have been related to barren plateaus, statistical complexity, and the expressive power of parametrized quantum circuits~\cite{marrero2021entanglement,Holmes_2022,bu2022statistical,bu2023effects,Il_rio_Correr_2024,bu2024complexity}. In particular, Refs.~\cite{bu2022statistical,bu2023effects} connect resource-theoretic quantities with learning-theoretic complexity measures, including Rademacher and Gaussian complexities, for quantum-circuit hypothesis classes. These results show that restricting the available quantum resources can limit the statistical complexity and generalization capability of quantum learning models. A second line of work studies how quantum resources in the training data or access model can yield learning advantages. This includes settings involving entangled data, quantum memory, coherent measurements, and quantum-generated datasets~\cite{Huang_2021,huang2022quantum,chen2022exponential,wang2024transition,wang2024separable,zhao2025entanglement}. A further related direction characterizes the learnability of restricted classes of quantum states or circuits in terms of their resource content, such as stabilizer states doped with a small number of non-Clifford gates and other low-magic families~\cite{leone2024learning,chia2024efficient,grewal2025efficient,mele2025efficient}. In these settings, the sample or computational complexity of learning can scale exponentially with the number of injected $\T$ gates.

These studies primarily examine how quantum resources affect quantum learners, quantum data, or the learnability of an individual quantum state or circuit. In contrast, our work considers the task of predicting the output of an entire circuit family with tunable rotation gates, using a classical learner trained on quantum-generated data. The relevant object is therefore a circuit family, rather than a fixed circuit instance or quantum state. Standard resource measures defined for individual states or circuit instances do not directly characterize the learning complexity of this function family. To address this gap, we introduce \DSE as a family-level resource measure for parametrized quantum circuits. We show that \DSE characterizes the classical learnability of quantum circuits by establishing matching sample- and computational-complexity bounds for learning-based classical surrogates. Our result therefore provides a direct connection between quantum resource theory and the classical learnability of quantum circuits, extending the role of quantum resources beyond conventional classical simulability and quantum learning advantage.

\section{Definitions and properties of dynamical stabilizer entropy (\texorpdfstring{\DSE}{DSE})}\label{append:sec:properties-SSE}
 
This section contains the details of the quantum resource framework underlying \DSE\ and the proof of Theorem~\ref{lem:ose_upper_bound}. Specifically, in SI.~\ref{append:subsec:resource-framework}, we present a formulation of the quantum resource framework tailored to a tunable circuit ensemble, in which the fixed gates play the role of free operations and the tunable rotation gates are resource-generating, and we give the formal definition of \DSE\ as the associated resource measure. Then, in SI.~\ref{append:subsec:SSE-properties}, we prove that \DSE\ satisfies the properties expected of a resource monotone within this framework and is upper-bounded by the number of tunable rotation gates. Last, in SI.~\ref{append:subsec:SSE-OSE}, we prove that in the worst case the operator stabilizer entropy (\OSE)~\cite{dowling2025magic} is lower bounded by \DSE, thereby determining the infeasible region of the computational phase diagram in Fig.~\ref{fig:scheme}b.

\subsection{Quantum resource framework for a tunable circuit ensemble}\label{append:subsec:resource-framework}
 
Here we provide a formulation of the quantum resource framework that supports \DSE. Following the convention of resource theories~\cite{chitambar2019quantum}, such a framework is specified by three ingredients: the set of \textit{objects} of interest, a set of \textit{free operations} together with the induced \textit{free objects} that carry no resource, and a \textit{resource measure} that quantifies the resource content and does not increase under the free operations. We instantiate each ingredient below and, at the end of this subsection, give the formal definition of \DSE\ as the associated resource measure. A conceptual point worth emphasizing at the outset is that the objects and free objects are defined purely at the level of the \emph{evolution operator} $O(\bx;U)$ and the circuit structure, independently of the mean-value function $f(\bx)=\Tr(\rho_0 O(\bx;U))$. The function $f(\bx)$ and its trigonometric expansion enter only later, as the vehicle through which the resource content of an object is \emph{measured}. They play no role in specifying which objects are free. 
 
\medskip
\noindent\textbf{Objects}. The objects of the framework are the Heisenberg-evolved observables generated by a tunable circuit ensemble. Recall the setup in the main text: we consider $N$-qubit circuits $U(\bx)=\prod_{j=1}^{d}(\RZ(\bx_j)V_j)\in\mathrm{Arc}(N,d)$, where $\{V_j\}$ are fixed blocks drawn from a universal gate set (e.g., $\{H,S,T,\CNOT\}$) and $\{\RZ(\bx_j)\}_{j=1}^{d}$ are the $d$ tunable single-qubit $Z$-rotations, arranged in an arbitrary layout. Given an observable $O=\sum_{l}O_l$ with $\sum_{l}\|O_l\|_{\infty}\le 1$, an object is the parameterized evolution operator
\begin{equation}\label{append:eq:object}
   O(\bx;U) = U(\bx)^{\dagger}\,O\,U(\bx),\qquad \bx\in [-\pi,\pi]^d.
\end{equation}
We denote the set of all such objects induced by layouts in $\mathrm{Arc}(N,d)$ by $\mathcal{O}(N,d)$. We emphasize that an object is the operator (family) in Eq.~(\ref{append:eq:object}) itself, while the reference state $\rho_0$ and the input distribution $\mathbb{D}$ are auxiliary data used only when measuring its resource content, and do not enter the definition of the object. Throughout this work, unless otherwise stated, we take $\rho_0 =\ket{0}\bra{0}^{\otimes N}$ and $\mathbb{D} = \Unif[-\pi,\pi]^d$. Thus, the structural definitions of the objects and free operations are independent of $(\rho_0,\mathbb{D})$, whereas the resource value is evaluated with respect to this fixed pair. The bound $\sum_{l}\|O_l\|_{\infty} \le 1$ is the sole normalization convention imposed on $O$ throughout the manuscript. 
 
\medskip
\noindent\textbf{Free operations and free objects}. The free operations are conjugations of an object by fixed, or equivalently, parameter-independent unitaries. Concretely, for any fixed gates $V\in\mathbb{S}\mathbb{U}(2^N)$, the induced free operation acts on an object $O(\bx;U)$ by 
\begin{equation}\label{append:eq:free-op}
    \mathcal{V}: \ O(\bx;U)\ \longmapsto\ V^{\dagger}\,O(\bx;U)\,V.
\end{equation}
Because $V$ carries no dependence on $\bx$, it does not introduce any tunable structure. The collection $\mathbb{F}=\{\mathcal{V}:V\in\mathbb{S}\mathbb{U}(2^N)\}$ is closed under composition, since $\mathcal{V}_1\circ\mathcal{V}_2$ is the conjugation by the fixed unitary $V_2V_1$, so $\mathbb{F}$ forms a well-defined set of free operations. By contrast, the tunable rotations $\RZ(\bx_j)$ are \emph{not} free, as they endow the object with genuine parameter dependence. Accordingly, the \emph{free objects} are defined purely at the structural level as the evolution operators built from \emph{no} tunable rotation gate, i.e., 
\begin{equation}\label{append:eq:free-object}
    \mathcal{O}_{\mathrm{free}} = \big\{ V^{\dagger}\,O\,V ~:~ V\in\mathbb{S}\mathbb{U}(2^N) \big\},
\end{equation}
which are exactly the objects generated with $d=0$. A free object carries no $\bx$-dependence and hence no tunable resource. The free operations in Eq.~(\ref{append:eq:free-op}) map free objects to free objects, since conjugating a $\bx$-independent operator by a fixed unitary again yields a $\bx$-independent operator, so the free set $\mathcal{O}_{\mathrm{free}}$ is closed under $\mathbb{F}$.

\medskip
\noindent\textbf{Resource measure}. The resource measure of the framework is the dynamical stabilizer entropy (\DSE). Fixing an input state $\rho_0$ and an input distribution $\mathbb{D}$ over $\bx\in[-\pi,\pi]^d$, we associate with an object $O(\bx;U)$ the mean-value function $f(\bx)=\Tr(\rho_0\,O(\bx;U))$ and expand it in the Pauli-transfer-matrix (PTM) representation~\cite{greenbaum2015introduction}. Recall that the PTM of a single rotation gate admits the decomposition
\begin{equation}\label{append:eq:PTM-RZ}
    \bRZ(\bx_j) = D_0 + \cos(\bx_j)\,D_1 + \sin(\bx_j)\,D_{-1},
\end{equation}
where $\{D_0,D_1,D_{-1}\}$ are the fixed component matrices induced by the $Z$-rotation~\cite{du2025efficient}. Propagating the observable through the circuit and collecting the trigonometric factors yields
\begin{equation}
    f(\bx) = \sum_{\bomega\in \Lambda} \Phi_{\bomega}(\bx)\,\Tr(\rho_0 Q_{\bomega}),
    \qquad
    \Phi_{\bomega}(\bx)=\prod_{j=1}^{d}\big[\mathbbm{1}_{\bomega_j=0}+\cos(\bx_j)\,\mathbbm{1}_{\bomega_j=1}+\sin(\bx_j)\,\mathbbm{1}_{\bomega_j=-1}\big],
\end{equation}
where $\Lambda=\{-1,0,1\}^{d}$ is the frequency set with $|\Lambda|=3^d$, and each operator $Q_{\bomega}$ is a linear combination of Pauli strings determined by the fixed blocks $\{V_j\}$ and the layout of $U$. The number of modes contributing to $f(\bx)$, namely those $\bomega$ with $\Tr(\rho_0 Q_{\bomega})\neq 0$, is governed by the circuit layout and can grow exponentially with $d$. Based on this expansion, the $\alpha$-order \DSE\ is defined as
\begin{align}\label{append:eq:SSE}
    \mathcal{M}^{(\alpha)}[O(\bx;U)] = \max_{V\in\mathbb{S}\mathbb{U}(2^N)} \mathcal{S}^{(\alpha)}\big[V^{\dagger}O(\bx;U)V\big],
    \qquad\text{with}\qquad
    \mathcal{S}^{(\alpha)}\big[O(\bx;U)\big] = \frac{1}{1-\alpha}\log\Big(\sum_{\bomega\in\Lambda}\mathrm{p}^{\alpha}(\bomega)\Big),
\end{align}
where the maximization runs over all fixed (parameter-independent) unitaries $V\in\mathbb{S}\mathbb{U}(2^N)$, which are defined as the free operations in the resource framework of Eq.~(\ref{append:eq:free-op}). Writing the transformed operator as $V^{\dagger}O(\bx;U)V=\sum_{\bomega\in\Lambda}\Phi_{\bomega}(\bx)Q'_{\bomega}$, the induced mode distribution takes the form
\begin{equation}\label{append:eq:mode-dist}
    \mathrm{p}(\bomega) = \frac{\tilde{\mathrm{p}}(\bomega)}{\sum_{\bomega'\in\Lambda}\tilde{\mathrm{p}}(\bomega')},
    \qquad\text{with}\qquad
    \tilde{\mathrm{p}}(\bomega) = \mathbb{E}_{\bx\sim\mathbb{D}}\,\Phi_{\bomega}(\bx)^2\,\Tr(\rho_0 Q'_{\bomega})^2,
\end{equation}
which records the average squared contribution of each mode $\bomega$. Intuitively, \DSE\ characterizes how broadly the mean-value function $f(\bx)$ is distributed over the frequencies $\bomega\in\Lambda$ generated by the tunable rotation gates. Notably, for the uniform distribution $\mathbb{D}=\Unif[-\pi,\pi]^d$, we have $\mathbb{E}_{\bx\sim \mathbb{D}} f(\bx)^2=\sum_{\bomega\in \Lambda}\tilde{\mathrm{p}}(\bomega)$. In the case of $\mathbb{E}_{\bx\sim \mathbb{D}} f(\bx)^2=0$, we define $ \mathcal{S}^{(\alpha)}\big[O(\bx;U)\big]=0$.

\medskip
\noindent \underline{Remark.} 
We term the proposed measure the \emph{dynamical stabilizer entropy}, since the non-free tunable gates $\RZ(\bx_j)$ generate nonstabilizerness in the resource theory of magic, and the nonstabilizerness they generate varies with $\bx_j$, i.e., they generate \textit{dynamical nonstabilizerness}. Beyond nonstabilizerness, our resource framework focuses on the capability of generating dynamical resources of parameterized quantum circuits, which are used to study the compressibility of representing these circuits classically. This also explains why fixed non-Clifford gates such as $T$, 
which generate only static nonstabilizerness, are regarded as free operations when acting through the object-level conjugation in Eq.~\eqref{append:eq:free-op}. In the subsequent section, we will establish the connection between $\DSE$ and $\OSE$, a magic monotone for quantifying nonstabilizerness defined in Eq.~\eqref{append:eq:ose}.



\subsection{Properties of \texorpdfstring{\DSE}{DSE}---Proof of Theorem~\ref{lem:ose_upper_bound}}\label{append:subsec:SSE-properties}

Building on the resource framework and the definition of \DSE\ in Eqs.~(\ref{append:eq:SSE})-(\ref{append:eq:mode-dist}), we now show that \DSE\ satisfies the basic properties required of a resource monotone within this framework, and that it is upper bounded by the number of tunable rotation gates.

\begin{theorem-non}[Formal statement of Theorem~\ref{lem:ose_upper_bound}]\label{append:lem:ose_upper_bound}
    Following the notation in Eqs.~(\ref{eq:target_function_set})-(\ref{eq:SSE}), and assuming $\mathbb{E}_{\bx\sim\mathbb{D}}[\Phi_{\bomega}(\bx)^2]>0$ for every $\bomega\in\Lambda$ and a non-scalar reference state $\rho_0\neq\mathbb{I}/2^N$, \DSE\ satisfies the following properties for any $\alpha>0$.
    \begin{itemize}
        \item[\textnormal{(i)}] \textnormal{Non-negativity and vanishing on free objects}: $\mathcal{M}^{(\alpha)}[O(\bx;U)]\ge 0$, and $d=0$ implies $\mathcal{M}^{(\alpha)}[O(\bx;U)]=0$.
        \item[\textnormal{(ii)}] \textnormal{Invariance under fixed gates}: $\mathcal{M}^{(\alpha)}[V^{\dagger}O(\bx;U)V]=\mathcal{M}^{(\alpha)}[O(\bx;U)]$ for any fixed gate $V\in\mathbb{S}\mathbb{U}(2^N)$.
        \item[\textnormal{(iii)}] \textnormal{Upper bound by rotation-gate count}: for any $d$ and $N$,
        \begin{equation}
            \mathcal{M}^{(\alpha)}[O(\bx;U)]\le d\cdot\log(3).
        \end{equation}
    \end{itemize}
    Moreover, when $d\le N$, the worst-case \DSE\ is upper bounded by the worst-case \OSE up to a constant factor, i.e., $d\lesssim \max_{\bx,U\in\mathrm{Arc}(N,d)}\mathfrak{M}_{\mathrm{ose}}^{(\alpha)}[O(\bx;U)]$, where the maximization of $\bx$ is over $[-\pi,\pi]^d$.
\end{theorem-non}

\begin{proof}[Proof of Theorem~\ref{lem:ose_upper_bound}]
This proof is composed of two parts. The first part proves the three properties (i)-(iii) of \DSE, which is presented in this subsection. The second part proves the worst-case lower bound of \OSE, which is deferred to SI.~\ref{append:subsec:SSE-OSE}. In the following, we separately prove these properties.

\smallskip
\noindent\underline{(i) Non-negativity}. By Eq.~(\ref{append:eq:mode-dist}), $\{\mathrm{p}(\bomega)\}_{\bomega\in\Lambda}$ is a normalized probability distribution over $\Lambda$, since $\tilde{\mathrm{p}}(\bomega)\ge 0$ for every $\bomega$ and the normalization is by definition. As the order-$\alpha$ R\'enyi entropy of any probability distribution is non-negative, we have $\mathcal{S}^{(\alpha)}[V^{\dagger}O(\bx;U)V]\ge 0$ for every fixed $V$, and hence 
\begin{equation}
    \mathcal{M}^{(\alpha)}[O(\bx;U)] = \max_{V\in\mathbb{S}\mathbb{U}(2^N)}\mathcal{S}^{(\alpha)}\big[V^{\dagger}O(\bx;U)V\big]\ge 0.
\end{equation}
For the case of $d=0$, the circuit contains no tunable rotation gate, so the frequency set collapses to the single mode $\Lambda=\{\bm{0}\}$ with $\Phi_{\bm{0}}(\bx)\equiv 1$. Accordingly, the mode distribution degenerates to a trivial point distribution with $\mathrm{p}(\bm{0})=1$, which gives for any $\alpha\ge 0$
\begin{equation}
    \mathcal{S}^{(\alpha)}\big[V^{\dagger}OV\big] = \frac{1}{1-\alpha}\log\big(\mathrm{p}^{\alpha}(\bm{0})\big) = \frac{1}{1-\alpha}\log(1) = 0.
\end{equation}
Since the above holds for every fixed $V$, the maximization yields $\mathcal{M}^{(\alpha)}[O(\bx;U)]=0$.

\smallskip
\noindent\underline{(ii) Invariance under fixed gates}. The invariance of \DSE\ under free operations can be verified directly by checking its definition. In particular, for any fixed gate $W\in\mathbb{S}\mathbb{U}(2^N)$, we have 
\begin{align}\label{append:eq:invariance}
    \mathcal{M}^{(\alpha)}\big[W^{\dagger}O(\bx;U)W\big]
    & = \max_{V\in\mathbb{S}\mathbb{U}(2^N)}\mathcal{S}^{(\alpha)}\big[V^{\dagger}\big(W^{\dagger}O(\bx;U)W\big)V\big] \nonumber\\
    & = \max_{V\in\mathbb{S}\mathbb{U}(2^N)}\mathcal{S}^{(\alpha)}\big[(WV)^{\dagger}O(\bx;U)(WV)\big] \nonumber\\
    & = \max_{V'\in\mathbb{S}\mathbb{U}(2^N)}\mathcal{S}^{(\alpha)}\big[V'^{\dagger}O(\bx;U)V'\big] \nonumber\\
    & = \mathcal{M}^{(\alpha)}[O(\bx;U)],
\end{align}
where the third equality follows by denoting $V'=WV$ and using the fact that, since $W$ is fixed, the map $V\mapsto WV$ is a bijection of $\mathbb{S}\mathbb{U}(2^N)$ onto itself, so the maximization over $V$ is equivalent to the maximization over $V'$. We remark that the crucial property enabling this invariance is that the fixed gate $W$ is parameter-independent: it does not alter the frequency support $\Lambda$ or the basis functions $\Phi_{\bomega}(\bx)$ in Eq.~(\ref{append:eq:tri_exp_T}), but only redistributes the operator-valued coefficients $Q'_{\bomega}$, which are re-optimized by the outer maximization. This establishes invariance under fixed gates, and in particular shows that \DSE\ is non-increasing under the free operations in Eq.~(\ref{append:eq:free-op}), as required of a monotone.

\smallskip
\noindent\underline{(iii) Upper bound by rotation-gate count}. Fix an arbitrary $V\in\mathbb{S}\mathbb{U}(2^N)$. The mode distribution $\mathrm{p}(\bomega)$ is supported on $\Lambda=\{-1,0,1\}^{d}$, whose cardinality is $|\Lambda|=3^d$. As it is a R\'enyi entropy over a support of size $|\Lambda|$, $\mathcal{S}^{(\alpha)}[\,\cdot\,]$ is maximized by the uniform distribution, i.e.,
\begin{equation}
    \mathcal{S}^{(\alpha)}\big[V^{\dagger}O(\bx;U)V\big]\le \log|\Lambda| = \log(3^d) = d\log(3).
\end{equation}
Since this bound is independent of $V$, taking the maximum over $V\in\mathbb{S}\mathbb{U}(2^N)$ yields 
\begin{equation}
    \mathcal{M}^{(\alpha)}[O(\bx;U)]\le d\log(3),
\end{equation}
which is the claimed upper bound. Consequently, $\max_{U\in\mathrm{Arc}(N,d)}\mathcal{M}^{(\alpha)}[O(\bx;U)]\lesssim d$, where the symbol $\lesssim$ hides the constant $\log(3)$. The remaining worst-case lower bound of \OSE\ is proved in SI.~\ref{append:subsec:SSE-OSE}.
\end{proof}

\subsection{Worst-case lower bound of \texorpdfstring{\OSE}{OSE}}\label{append:subsec:SSE-OSE}

We now complete the proof of Theorem~\ref{append:lem:ose_upper_bound} by establishing that, in the regime $d\le N$, the worst-case \OSE\ grows at least linearly in $d$. 
For elucidation, we first recall the operator stabilizer entropy (\OSE) proposed in Ref.~\cite{dowling2025magic}, which is an instance-wise magic monotone quantifying the nonstabilizerness of a Heisenberg-evolved operator through the spread of its Pauli spectrum. Expanding the evolved observable in the Pauli basis $\mathcal{P}_N=\{\mathbb{I},X,Y,Z\}^{\otimes N}$ as $O(\bx;U)=\sum_{P\in\mathcal{P}_N}c_P(\bx)P$, the \OSE\ is the order-$\alpha$ R\'enyi entropy of the normalized Pauli-weight distribution $\Xi_P(\bx)=c_P(\bx)^2/\sum_{P'\in\mathcal{P}_N}c_{P'}(\bx)^2$. A key structural property of \OSE, established in Ref.~\cite{dowling2025magic}, is its additivity under tensor products, i.e., $\mathfrak{M}_{\mathrm{ose}}^{(\alpha)}[A\otimes B]=\mathfrak{M}_{\mathrm{ose}}^{(\alpha)}[A]+\mathfrak{M}_{\mathrm{ose}}^{(\alpha)}[B]$.

To lower bound the worst case, it suffices to exhibit one circuit layout $U(\bx)\in\mathrm{Arc}(N,d)$ with specific $\bx\in \mathbb{R}^d$ whose evolved observable carries \OSE\ of order $d$. Since $d\le N$, we may place the $d$ tunable rotations on $d$ distinct qubits. Consider the layout in which the observable $O=X_1\otimes \cdots \otimes X_d \otimes \mathbb{I}^{\otimes (N-d)}$ is supported on these $d$ rotated qubits, and each tunable rotation $\RZ(\bx_j)$ is conjugated by fixed blocks such that it independently generates a two-branch Pauli spread on its qubit. A representative instance is 
\begin{equation}\label{append:eq:single-qubit-branch}
    \RZ(\bx_j)^{\dagger}X_j\RZ(\bx_j) = \cos(\bx_j)X_j - \sin(\bx_j)Y_j,
\end{equation}
which produces two Pauli terms per rotated qubit. Taking $\bx=\frac{\pi}{4}\cdot \bm{1}_d$, each single-qubit factor has the Pauli-weight distribution $(1/2, 1/2)$ and hence has order-$\alpha$ \OSE equal to one. Invoking the additivity of \OSE under tensor products, the total \OSE is therefore exactly $d$, and in particular
\begin{equation}
    \max_{\bx,U\in\mathrm{Arc}(N,d)}\mathfrak{M}_{\mathrm{ose}}^{(\alpha)}[O(\bx;U)]\ \ge\ d\ =\ \Omega(d),
\end{equation}
which is exactly the claimed bound $d\lesssim \max_{\bx,U\in\mathrm{Arc}(N,d)}\mathfrak{M}_{\mathrm{ose}}^{(\alpha)}[O(\bx;U)]$.

Taken together, combining property (iii) in Theorem~\ref{append:lem:ose_upper_bound} with the above lower bound yields the two-sided worst-case relation
\begin{equation}\label{append:eq:two-sided}
    \max_{U\in\mathrm{Arc}(N,d)}\mathcal{M}^{(\alpha)}[O(\bx;U)]\ \lesssim d \le \max_{\bx,U\in\mathrm{Arc}(N,d)}\mathfrak{M}_{\mathrm{ose}}^{(\alpha)}[O(\bx;U)]\qquad (d\le N).
\end{equation}
In other words, for a fixed number of rotation gates $d$, the worst-case \DSE\ scales at most as $\mathcal{O}(d)$, whereas the worst-case \OSE\ scales at least as $\Omega(d)$. Consequently, the worst-case resources satisfy $\DSE  \lesssim \OSE$ up to universal constant factors. This completes the proof of Theorem~\ref{append:lem:ose_upper_bound}.

\section{Information-theoretic bound of sample complexity for classical surrogates---Proof of Lemma~\ref{mt:thm:general_ML_complexity}}\label{append:sec:information_theoretic_bound}
In this section, we provide the formal statement of Lemma~\ref{mt:thm:general_ML_complexity} and its proof.

\begin{theorem-non}
    [Formal statement of Lemma~\ref{mt:thm:general_ML_complexity}] Following notations in Eqs.~\eqref{eq:target_function_set}-\eqref{eq:SSE}. Let $\nu_{f} =\mathbb{E}_{\bx \sim [-\pi,\pi]^d} f(\bx)^2$ and $\epsilon\le \nu_f/4$, and $\mathcal{U}_{\lambda}=\{ U\in \mathsf{Arc}(N,d):\mathcal{M}^{(1)}[O(\bx;U)]\le \mathcal{M}_{\lambda}\}$ be the circuit family with bounded \DSE $\mathcal{M}_{\lambda}$, where $1\le\mathcal{M}_{\lambda}\le d$. Let $\mathcal{T}=\{(\bxi,y^{(i)})\}_{i=1}^n$ be the training dataset containing $n$ training examples $\bxi$ sampled from the uniform distribution over $[-\pi,\pi]^d$ and $y^{(i)}$ be the statistical estimation of $y^{(i)}$ from $m$ incoherent measurements. Then, the training data size
    \begin{equation}\label{eq:ml_sample_complexity}
       \Omega\left( \frac{2^{ \mathcal{M}_{\lambda}} \cdot (1-2\nu_f^{-1}\epsilon)^2 }{ 24m\nu_f }\right) \le n \le 
       \tilde{\mathcal{O}} \Big( \frac{2^{12 \nu_{f} \mathcal{M}_{\lambda}/\epsilon} \cdot d}{\epsilon}  \Big)
    \end{equation}
    is sufficient and necessary to achieve $\epsilon$-prediction error $\mathsf{R}(h):=\mathbb{E}_{\bx \sim [-\pi,\pi]^d} |h(\bx) -f_O(\bx)|^2 \le \epsilon$ for a classical surrogate $h(\bx)$ to learn any $f\in \mathcal{F}_{\mathcal{U}_{\lambda}}$ with high probability. Here $h(\bx)$ has a model dimension $\tilde{\mathcal{O}}(2^{12\nu_{f} \mathcal{M}_{\lambda}/\epsilon}) $, with $\tilde{\mathcal{O}}$ hiding logarithmic factors.
\end{theorem-non}

The proof of Lemma~\ref{mt:thm:general_ML_complexity} can be broken down into two parts: the \DSE-dependent upper bound of sample complexity and the \DSE-dependent lower
bound of sample complexity. The corresponding bounds are established in the following two theorems, whose proofs are deferred to SI.~\ref{append:sec:upper_bound} and SI.~\ref{append:sec:dse_lower_bound}, respectively.

\begin{theorem}[\DSE-dependent upper bound of sample complexity]\label{append:thm:general_upper_bound}
    Consider a classical learner that queries an $N$-qubit quantum circuit and collects a training dataset $\mathcal{T}=\{(\bxi, y^{(i)})\}_{i=1}^n$ of $n$ examples in order to predict the observable expectation value $f_{O}(\bx)$ for unseen inputs $\bx\in[-\pi,\pi]^d$. Then there exists a classical surrogate ${h}_{\mathcal{T}}(\bx)$ using a training dataset of size
    \begin{equation}\label{append:eq:n_upper_bound_thm}
        n \le \tilde{\mathcal{O}} \left( \frac{2^{12\nu_f \mathcal{M}_{\lambda}/\epsilon} \cdot d}{\epsilon}  \right),
    \end{equation}
    such that $\mathsf{R}({h}_{\mathcal{T}}) \le \epsilon$,
    where $\tilde{\mathcal{O}}$ hides logarithmic factors. Here, ${h}_{\mathcal{T}}$ has the model dimension $2^{12\nu_f \mathcal{M}_{\lambda}/\epsilon}$
\end{theorem}

\begin{theorem}[\DSE-dependent lower bound of sample complexity]\label{append:thm:lower_bound_complexity} 
    Consider a classical surrogate $h_{\mathcal{T}}(\bx)$ that learns from training data $\mathcal{T}=\{(\bxi,\yi)\}_{i=1}^n$, where $\bxi$ are uniformly sampled from $[-\pi,\pi]^d$ and  $\yi$ is the statistical estimation of the target expectation value $f(\bxi)$ from $m$ measurement outcomes, namely $\mathbb{E}y_i=f(\bxi)$, suppose that for all $f\in \mathcal{F}_{\mathcal{U}_{\lambda}}$ the classical surrogate $h(\bx)$ can achieve $
        \mathbb{E}_{\bx\sim [-\pi,\pi]^d} |h(\bx)-f(\bx)|^2 \le \epsilon$
    with high probability. For any $1\le \mathcal{M}_{\lambda} \le \min\{d,N/2\}$, the size of the training dataset must obey
    \begin{equation}
        n \ge  \Omega\left( \frac{2^{ \mathcal{M}_{\lambda}} \cdot (1-2\nu_f^{-1}\epsilon)^2 }{ 24m\nu_f }\right),
    \end{equation}
    where the prediction error yields $\epsilon \le \nu_f/4$.
    This lower bound holds even when the learner has access to a sampler on the probability distribution $\{\mathrm{p}(\bomega)\}_{\bomega \in \{0,\pm 1\}^d}$.
    \end{theorem}

    \noindent \underline{Remark.} The choice $\epsilon\le \nu_f/4$ is a constant-relative-accuracy assumption commonly used in deriving sample-complexity lower bounds in quantum learning theory \cite{huang2023learning,wang2024transition}. Since $\nu_f=\mathbb{E}_{\bx} f(\bx)^2$ measures the average squared magnitude of the target function, this choice requires the learner to achieve a constant relative accuracy with respect to $\nu_f$, and is meaningful for ruling out the trivial classical learner whose prediction error is comparable to that of a random guess. 

    \begin{proof}
        [Proof of Lemma~\ref{mt:thm:general_ML_complexity}] The proof of Lemma~\ref{mt:thm:general_ML_complexity} can be directly obtained by employing the results of Theorems~\ref{append:thm:general_upper_bound} and \ref{append:thm:lower_bound_complexity}.
    \end{proof}

\section{\DSE-dependent upper bound of sample complexity---Proof of Theorem~\ref{append:thm:general_upper_bound}}\label{append:sec:upper_bound}

In this section, we analyze the upper bound on the sample complexity in terms of the \DSE\ when a classical surrogate is used to predict the output of quantum circuits containing $d$ tunable rotation gates. Following the convention of Ref.~\cite{du2025efficient}, our analysis adopts a standard route in statistical learning theory. In particular, one first specifies a hypothesis space in which the learner searches for a classical surrogate, and then bounds the number of training examples needed for the empirical risk minimizer over this space to achieve a small prediction error, where the model complexity of the hypothesis space dictates the required sample size. The key technical contribution of our analysis lies in the construction of this hypothesis space, which is tailored to the trigonometric expansion of the target function in Eq.~\eqref{append:eq:tri_exp_T} and thereby makes the \DSE\ enter the sample-complexity bound.

The organization of this section is as follows. In SI.~\ref{append:subsec:ub_formulation}, we first formulate the problem of learning quantum circuits with classical surrogates by constructing the classical hypothesis space from a partial trigonometric expansion. In SI.~\ref{append:subsec:ub_theorem}, we present the proof of the \DSE-dependent upper bound of the sample complexity for this learning problem. In SI.~\ref{append:subsec:ub_lemma}, we provide the proof of the technical lemma used in the above analysis.

\subsection{Learning problem formulation}\label{append:subsec:ub_formulation}

We begin by recalling the form of the concept class of the target expectation-value function 
\begin{equation}\label{append:eq:concept_class}
    \mathcal{F}_{\mathcal{U}_{\lambda}}=\left\{f_{O}(\bx)=\Tr[\rho_0 U(\bx)^{\dagger}OU(\bx)] \bigg| \bx\sim[-\pi,\pi]^d, U\in \mathcal{U}_{\lambda} \subset \mathsf{Arc}(N,d) \right\},
\end{equation}
where $\rho_0$ is the fixed initial state, and $O=\sum_{l=1}^q O_l$ refers to a fixed bounded observable satisfying $\sum_{l=1}^q \|O_l\|_{\infty}\le 1$. 
Here, the circuit family $\mathcal{U}_{\lambda}$ is given by
$\mathcal{U}_{\lambda}=\{U\in \mathsf{Arc}(N,d):\mathcal{M}^{(1)}[O(\bx;U)]\le \mathcal{M}_{\lambda}\}$ consisting of quantum circuits with bounded \DSE, where $1\le \mathcal{M}_{\lambda} \le d$. 
For any $f\in \mathcal{F}_{\mathcal{U}_{\lambda}}$, it can be written as the trigonometric basis expansion, i.e.,
\begin{equation}\label{append:eq:ub_target}
    f_{O}(\bx)=\Tr[\rho_0 O(\bx;U)]=\sum_{\bomega\in \Lambda} \Phi_{\bomega}(\bx)\Tr(\rho_0 Q_{\bomega}),
\end{equation}
where $\Phi_{\bomega}(\bx)$ denotes the trigonometric basis over the frequency set $\Lambda=\{0,\pm 1\}^d$ defined in Eq.~\eqref{append:eq:tri_featrue}, and $Q_{\bomega}$ is an unknown linear combination of Pauli strings, as described in Eq.~\eqref{apppend:eq:Px_withT}. In our learning setting, the classical learner only has access to the classical inputs and the classical measurement outcomes from the unknown quantum circuit $U(\bx)$, from which it must infer the target function. The circuit layout and depth of $U(\bx)\in \mathcal{U}$ may be arbitrary with bounded \DSE $\mathcal{M}^{(1)}[O(\bx;U)]$, provided that the circuit contains $d$ Pauli rotation gates. We note that the expectation-value function associated with $U(\bx)$ admits the expansion in Eq.~\eqref{append:eq:ub_target} for arbitrary circuit layouts and depths, although different circuits generally lead to different operators $Q_{\bomega}$, as explained in SI.~\ref{append:subsec:trigo-monomial-exp-QC}.

The expansion in Eq.~\eqref{append:eq:ub_target} involves $|\Lambda|=3^d$ frequencies, so that learning $f_{O}(\bx)$ by fitting all of them as done by prior studies \cite{schreiber2023classical} would be intractable. Nevertheless, the contributions of $\Phi_{\bomega}(\bx)\Tr(\rho_0Q_{\bomega})$ for different frequencies $\bomega\in\Lambda$ to $f_{O}(\bx)$ may be highly non-uniform for a small $\mathcal{M}^{(1)}[O(\bx;\mathcal{U})]$, in which case discarding those with negligible contributions barely alters the target function. In this regard, we depart from previous constructions of the hypothesis space that are based directly on the full function $f_{O}(\bx)$, and instead adopt a \emph{partial} trigonometric expansion in which only the frequencies whose induced probability exceeds a threshold $\tau$ are retained, namely
\begin{equation}\label{append:eq:hypothesis_space}
    \mathcal{F}_{\tau}=\left\{f_{\tau}(\bx)= 
    g\left(\sum_{\bomega\in \Lambda_{\tau}} \Phi_{\bomega}(\bx)\Tr(\rho_0 Q_{\bomega})\right) \bigg| \bx \in [-\pi,\pi]^d, ~ \Lambda_{\tau}=\{\bomega:\mathrm{p}(\bomega)\ge \tau\} \subset \{0,\pm 1\}^d \right\},
\end{equation}
where $g(z)$ refers to the clip function $g(z)=\min\{1,\max\{-1,z\}\}$ which is used to control the bound of the truncation function $f_{\tau}(\bx)$ such that $|f_{\tau}(\bx)|\le \max_{\bx}|f(\bx)|\le \|O\|_{\infty}\le 1$, the retained set $\Lambda_{\tau}$ collects the frequencies whose induced probability exceeds a threshold $\tau$, and
\begin{equation}\label{append:eq:ub_prob}
    \mathrm{p}(\bomega)=\frac{2^{-\|\bomega\|_0}\Tr(\rho_0 Q_{\bomega})^2}{\nu_f},
    \qquad \text{with}\quad
    \nu_f=\sum_{\bomega\in \{0,\pm 1\}^d}2^{-\|\bomega\|_0}\Tr(\rho_0 Q_{\bomega})^2,
\end{equation}
is the probability distribution over $\bomega$ induced by $f_{O}(\bx)$, as defined in Eq.~\eqref{append:eq:mode-dist}. From the learner's perspective, the quantum computer can generate many circuits with varying layouts and depths, which in turn induce different subsets $\Lambda_{\tau}\subset \{0,\pm 1\}^d$ containing frequencies with large probabilities $\mathrm{p}(\bomega)$. We emphasize that the threshold $\tau$ is a free parameter of the construction rather than an assumption imposed on the circuit, as it recovers the full expansion in the limit $\tau\to 0$. In this regard, analyzing how small $\tau$ must be chosen to attain a prescribed prediction error, and how many frequencies are consequently retained, is precisely what allows us to identify the quantity governing the classical learnability for the circuit family $\mathcal{U}$.

Equipped with the hypothesis space $\mathcal{F}_{\tau}$, we now specify how the learner selects a hypothesis from it. For a specified quantum circuit with an unknown but fixed layout, we denote the corresponding target concept by $f^*_O(\bx)$, and its truncated version by $f^*_{\tau}(\bx)$. The classical learner uses the collected training dataset $\mathcal{T}=\{(\bxi,y^{(i)})\}_{i=1}^n$ to infer the truncated target function $f^*_{\tau}(\bx)$. Since $\mathcal{F}_{\tau}$ is a continuous function class, we construct a finite $\eta$-packing net $\mathcal{P}_{\eta}(\mathcal{F}_{\tau},\varrho)$ of $\mathcal{F}_{\tau}$, as detailed in Definition~\ref{append:def:packing}. Here, $\rho$ refers to the distance metric of $\mathcal{F}_{\tau}$ induced from the average prediction, given as 
\begin{equation}\label{append:eq:metric_distance}
    \varrho(f,g):=\left(\mathbb{E}_{\bx\sim\mathbb{D}} \left|f(\bx)-g(\bx)\right|^2 \right)^{\frac{1}{2}}
\end{equation}
In this regard, the learner performs the empirical risk minimization over a discretization set $\mathcal{P}_{\eta}(\mathcal{F}_{\tau},\varrho)$ of $\mathcal{F}_{\tau}$, i.e.,
\begin{equation}\label{append:eq:emp_minimize}
    h_{\mathcal{T}} = \arg\min_{f \in \mathcal{P}_{\eta}(\mathcal{F}_{\tau},\varrho)} \frac{1}{n} \sum_{i=1}^n \left| f(\bxi) - y^{(i)} \right|^2 ,
\end{equation}
where the detailed construction of $\mathcal{P}_{\eta}(\mathcal{F}_{\tau},\varrho)$ and the setting of $\eta$ is formalized in SI.~\ref{append:subsec:ub_theorem}. We will show that when choosing appropriate $\tau$ such that the discrepancy between \emph{truncated} target $f^*_{\tau}(\bx)$ and the \emph{full} target $f^*_{O}(\bx)$ is controlled, i.e.,
\begin{equation}
   \mathbb{E}_{\bx\in [-\pi,\pi]^d} \left|f_O^*(\bx)-f_{\tau}^*(\bx)\right|^2 \le \frac{\epsilon}{4},
\end{equation}
then the learned classical surrogate $h_{\mathcal{T}}$ is able to achieve a small prediction error $\mathsf{R}(h_{\mathcal{T}})\le \epsilon$ for a sufficiently large $n$. In this regard, \DSE\ enters as the quantity that dictates how small $\tau$ must be, and hence how many training examples $n$ are required.

\subsection{Sample-complexity upper bound}\label{append:subsec:ub_theorem}

Based on the formulation in SI.~\ref{append:subsec:ub_formulation}, we provide the proof of Theorem~\ref{append:thm:general_upper_bound}. Following the convention of Ref.~\cite{du2025efficient}, the proof adopts a standard technical approach in statistical learning theory and quantum learning theory. Specifically, the discretized hypothesis space $\mathcal{P}_{\eta}(\mathcal{F}_{\tau},\varrho)$ in Eq.~\eqref{append:eq:emp_minimize} is formalized as a maximal packing net, whose cardinality, i.e., the packing number, characterizes the complexity of the hypothesis space $\mathcal{F}_{\tau}$. The formal definition is as follows.

\begin{definition}[Maximal $\epsilon$-packing net and covering net]\label{append:def:packing}
    Let $(\mathcal{F},\varrho)$ be a metric space and let $\eta >0$, where $\varrho$ is a distance metric defined on the function class $\mathcal{F}$. An $\epsilon$-packing net $\mathcal{P}_{\epsilon}(\mathcal{F},\varrho)$ is a discrete subset of $\mathcal{F}$ such that, for any $f_1, f_2\in \mathcal{P}_{\epsilon}(\mathcal{F},\varrho)$, their distance satisfies $\varrho(f_1, f_2)\ge \epsilon$. It is maximal if no further element in $\mathcal{F}$ can be added while preserving this separation. The packing number is the cardinality of $\mathcal{P}_{\epsilon}(\mathcal{F},\varrho)$. Every maximal $\epsilon$-packing is also an $\epsilon$-covering net, namely for every $f\in\mathcal{F}$, there exists $g\in \mathcal{P}_{\epsilon}(\mathcal{F},\varrho)$ such that $\varrho(f,g)\le \epsilon$.
\end{definition}



Given the packing number, the prediction error $\mathsf{R}(h_{\mathcal{T}})$ is controlled by the following proposition, which was established in prior work for the same learning setting but with a different hypothesis space.

\begin{proposition}[Adapted from Proposition 1, Ref.~\cite{huang2021information}]\label{append:prop:n_upper_bound_packing}
    Suppose that the observable $O=\sum_{i}O_i$ satisfies $\sum_i\|O_i\|_{\infty} \le B=1$, and let $\mathcal{P}_{\eta}(\mathcal{F}_{\tau},\varrho)$ be the maximal $\eta$-packing net of $\mathcal{F}_{\tau}$. Let $h_{\mathcal{T}}$ be an element of $\mathcal{P}_{\eta}(\mathcal{F}_{\tau},\varrho)$ that minimizes the empirical training error in Eq.~\eqref{append:eq:emp_minimize}. If for any $f_O^*\in \mathcal{F}$, there exists $f_{\tau}\in\mathcal{P}_{\eta}(\mathcal{F}_{\tau},\varrho)$ such that $\varrho(f_O^*,f_{\tau}) \le 2\eta$. Then, for $\delta\in (0,1)$ and the size of the training dataset satisfying
    \begin{equation}
        n \ge \frac{38B^2 \log(4|\mathcal{P}_{\eta}(\mathcal{F}_{\tau},\varrho)|/\delta)}{\eta^2},
    \end{equation}
    with probability at least $1-\delta$, we have
    \begin{equation}
        \mathbb{E}_{\bx \sim [-\pi,\pi]^d} \left|h_{\mathcal{T}}(\bx) -f_{O}^*(\bx) \right|^2 \le 12 \eta^2 .
    \end{equation}
\end{proposition}

\noindent\underline{Remark}. 
Proposition~\ref{append:prop:n_upper_bound_packing} adapts Ref.~\cite[Proposition~1]{huang2021information}, which is stated in terms of a $2\eta$-packing of $\mathcal{F}$, to an $\eta$-packing of the truncated function class $\mathcal{F}_{\tau}$. The proof of the original result uses the packing set only as a finite hypothesis class containing an element that approximates the target function $f_O^*$ within distance $2\eta$. Therefore, the same argument applies to $\mathcal{P}_{\eta}(\mathcal{F}_{\tau},\varrho)$, provided that, for every $f_O^*\in\mathcal{F}$, there exists an $f_{\tau}\in\mathcal{P}_{\eta}(\mathcal{F}_{\tau},\varrho)$ satisfying
$\varrho(f_O^*,f_{\tau})\le 2\eta$.
Moreover, it only focuses on the case of $m=1$, but the results still hold for the case of $m>1$. This is because in the extreme case with $m\to \infty$, the required training data size is reduced to $n\ge 38\log(2\mathcal{P}_{\eta}(\mathcal{F}_{\tau}, \varrho)/\delta)/\epsilon$ \cite[Equation (C66)]{huang2021information}
In this regard, the setting of $m \ge 1$ only trivially influences the sample complexity bound (at most logarithmically). For this reason, we omit the relevant analysis in our proof.


\smallskip

To satisfy the approximation condition $\varrho(f_O^*,f_{\tau})\le 2\eta$ required by Proposition~\ref{append:prop:n_upper_bound_packing}, we apply the triangle inequality
\begin{equation}
    \varrho(f_O^*,f_{\tau}) \le \varrho(f_O^*,f_{\tau}^*) + \varrho(f_{\tau}^*,f_{\tau}).
\end{equation}
It therefore suffices to show that there exists a function $f_{\tau}^*\in \mathcal{F}_{\tau}$ such that $\varrho(f_O^*,f_{\tau}^*)\le \eta$ and  $\varrho(f_{\tau}^*,f_{\tau}) \le \eta$. By the definition of the maximal $\eta$-packing $\mathcal{P}_{\eta}(\mathcal{F}_{\tau},\varrho)$, we have that for any $f_{\tau}^*$ there exists $f_{\tau}\in\mathcal{P}_{\eta}(\mathcal{F}_{\tau},\varrho)$ satisfying $\varrho(f_{\tau}^*,f_{\tau}) \le \eta$, as $\mathcal{P}_{\eta}(\mathcal{F}_{\tau},\varrho)$ is also a $\eta$-covering net of $\mathcal{F}_{\tau}$. Meanwhile, $\varrho(f_O^*,f_{\tau}^*)$ refers to the truncation error, which is controlled by choosing an appropriate threshold $\tau$, as stated in the following lemma, whose proof is deferred to SI.~\ref{append:subsec:ub_lemma}.

\begin{lemma}\label{lem:est_err_sse_bound_1_order}
    Let $\mathrm{p}(\bomega)=2^{-\|\bomega\|_0}\Tr(\rho_0 Q_{\bomega})^2/\nu_f$ be the probability distribution induced by the circuit expectation $f_O(\bx)$ over the frequencies $\bomega\in \{0,\pm 1\}^d$. Consider the set of retained frequencies $\Lambda_{\tau}=\{\bomega\in \{0,\pm 1\}^d:\mathrm{p}(\bomega)\ge \tau \}$ with $\tau \in (0,1)$, and let $\mathcal{M}^{(1)}[O(\bx;U)]$ be the first-order \DSE\ defined in Eq.~\eqref{append:eq:SSE}. Setting
    \begin{equation}\label{append:eq:tau_choice}
        \tau=2^{-\nu_f \mathcal{M}^{(1)}[O(\bx;U)]/\eta^2},
    \end{equation}
    we have
    \begin{equation}
        \varrho(f_O^*,f_{\tau}^*):=\sqrt{\mathbb{E}_{\bx\in [-\pi,\pi]^d} \left|f_O^*(\bx)-f_{\tau}^*(\bx)\right|^2} \le \eta.
    \end{equation}
\end{lemma}

Supported by Lemma~\ref{lem:est_err_sse_bound_1_order} and Proposition~\ref{append:prop:n_upper_bound_packing}, the proof of Theorem~\ref{append:thm:general_upper_bound} reduces to upper bounding the packing number $|\mathcal{P}_{\eta}(\mathcal{F}_{\tau},\varrho)|$. Unlike the analysis in Ref.~\cite{du2025efficient}, which bounds the packing number by enumerating circuits consisting of a specified number of Clifford gates and $\RZ$ gates, we consider the truncated hypothesis space $\mathcal{F}_{\tau}$ in Eq.~\eqref{append:eq:hypothesis_space}. The packing number of this space is determined by two quantities: the feature dimension $|\Lambda_{\tau}|$ associated with the feature map $[\Phi_{\bomega}(\bx)]_{\bomega \in \Lambda_{\tau}}$, and the number of possible subsets $\Lambda_{\tau}\subset \Lambda$ related to various circuit layouts for the threshold $\tau$ specified in Lemma~\ref{lem:est_err_sse_bound_1_order}.

\begin{proof}[Proof of Theorem~\ref{append:thm:general_upper_bound}]
    As discussed above, this proof is composed of two parts. We first upper bound the packing number $|\mathcal{P}_{\eta}(\mathcal{F}_{\tau},\varrho)|$ with setting $\eta^2=\epsilon/12$, and then apply Proposition~\ref{append:prop:n_upper_bound_packing} to derive the sample-complexity upper bound required to achieve the $\epsilon$-prediction error $\mathsf{R}(h)\le \epsilon$.

    \smallskip
    \noindent\underline{Upper bound on the packing number}. The packing number $|\mathcal{P}_{\eta}(\mathcal{F}_{\tau},\varrho)|$ is determined by two components: the total number of possible subsets $\Lambda_{\tau}$, and, for each fixed $\Lambda_{\tau}$, the packing number of the corresponding class of functions $f_{\tau}$. In particular, for every admissible support $\Lambda_{\tau} \subset \{0,\pm 1\}^d$, choose a maximal $\eta$-packing $\mathcal{P}_{\eta}(\mathcal{F}(\Lambda_{\tau}),\varrho)$, where the hypothesis space $\mathcal{F}(\Lambda_{\tau})$ refers to an $|\Lambda_{\tau}|$-dimensional linear function class over the fixed feature map $[\Phi_{\bomega}(\bx)]_{\bomega\in\Lambda_{\tau}}$, namely
    \begin{equation}\label{append:eq:union_packing}
        \mathcal{F}(\Lambda_{\tau})=
        \left\{ f_{\tau}'(\bx)=\sum_{\bomega\in \Lambda_{\tau}} \Phi_{\bomega}(\bx)\Tr(\rho_0 Q_{\bomega}) \bigg| [\Tr(\rho_0 Q_{\bomega})]_{\bomega\in \Lambda_{\tau}}:\sum_{\bomega\in \Lambda^{\tau}}2^{-\|\bomega\|_0}\Tr(\rho_0 Q_{\bomega})^2\le 1\right\}.
    \end{equation}
    By Definition~\ref{append:def:packing}, the maximal $\eta$-packing of $\mathcal{F}_{\tau}$ yields,
    \begin{equation}
        |\mathcal{P}_{\eta}(\mathcal{F}_{\tau},\varrho)| \le \left|  \bigcup_{\Lambda_{\tau} \subset \{0,\pm 1\}^d} \mathcal{P}_{\eta}(\mathcal{F}(\Lambda_{\tau}),\varrho) \right|.
    \end{equation}

    To this end, we first count the possible subsets $\Lambda_{\tau}\in \{0,\pm 1\}^d$. We note that, for a fixed threshold $\tau$, the cardinality of $\Lambda_{\tau}$, or equivalently the feature dimension of $f_{\tau}$, satisfies
    \begin{equation}\label{append:eq:Lambda_tau_size}
        |\Lambda_{\tau}|
        = \sum_{\bomega\in\Lambda_{\tau}}1
        \le \frac{1}{\tau} \sum_{\bomega\in\Lambda_{\tau}} \mathrm{p}(\bomega)
        \le \left\lceil {\tau}^{-1} \right \rceil,
    \end{equation}
    where the first inequality follows from the definition of $\Lambda_{\tau}$, namely $\mathrm{p}(\bomega)\ge \tau$ for every $\bomega\in\Lambda_{\tau}$, and the second inequality follows from $\sum_{\bomega\in\Lambda_{\tau}}\mathrm{p}(\bomega)\le 1$. Since $\Lambda=\{0,\pm1\}^d$ has cardinality $3^d$, the total number of possible subsets $\Lambda_{\tau}$ with cardinality $|\Lambda_{\tau}|$ is precisely the number of ways of choosing $|\Lambda_{\tau}|$ elements from the set $\Lambda$, namely
    \begin{equation}\label{append:eq:all_Omega_tau}
       \binom{3^d}{|\Lambda_{\tau}|}  \le \binom{3^d}{\left\lceil  {\tau}^{-1} \right \rceil}
        \le \left( \frac{e 3^d}{ \left\lceil {\tau}^{-1} \right \rceil} \right)^{ \left\lceil {\tau}^{-1} \right \rceil},
    \end{equation}
    where the first equality follows the nondecreasing property of $\binom{3^d}{k}$ by restricting $\tau\ge 2/3^d$, and the second inequality follows from the standard binomial bound $\binom{a}{b}\le (ea/b)^b$.

    We next bound the packing number of $\mathcal{F}(\Lambda_{\tau})$. Define the normalized features and weighted coefficients
    \begin{equation}
        \Psi_{\bomega}(\bx):=2^{\|\bomega\|_0/2} \Phi_{\bomega}(\bx), \qquad \bm{c}_{\bomega}=2^{-\|\bomega\|_0/2}\Tr(\rho_0 Q_{\bomega})
    \end{equation}
    The orthogonality relation $\mathbb{E}_{\bx\sim [-\pi,\pi]^d} \Phi_{\bomega}(\bx) \Phi_{\bomega'}(\bx)=\delta_{\bomega,\bomega'} 2^{-\|\bomega\|_0}$ implies the orthonormality of $\{\Psi_{\bomega}(\bx)\}_{\bomega\in\Lambda_{\tau}}$ and that $f_{\tau}(\bx)=\sum_{\bomega\in\Lambda_{\tau}} \bm{c}_{\bomega} \Psi_{\bomega}(\bx)$. Consequently,
    \begin{equation}
        \|\bm{c}\|_2^2= \sum_{\bomega\in\Lambda_{\tau}} 2^{-\|\bomega\|_0}\Tr(\rho_0 Q_{\bomega})^2= \mathbb{E}_{\bx\sim [-\pi,\pi]^d} |f_{\tau}(\bx)|^2 \le \mathbb{E}_{\bx\sim [-\pi,\pi]^d} |f_{O}(\bx)|^2=\nu_f \le 1, 
    \end{equation}
    where the first inequality follows because $f_{\tau}$ is the orthogonal projection of $f_O$ onto the retained modes $\bomega\in\Lambda_{\tau}$, while the
    last inequality follows from $|f_O(\bx)| \le \|O\|_{\infty} \le \sum_{i}\|O_i\|_{\infty}\le 1$. Thus, the weighted-coefficient representation of $f_{\tau}'$ regarding the feature $\Psi_{\bomega}(\bx)$ embeds $\mathcal{F}(\Lambda_{\tau})$ isometrically into the unit Euclidean ball in $\mathbb{R}^{|\Lambda_{\tau}|}$ under the metric $\varrho$. Employing the standard volumetric estimate~\cite[Proposition 4.2.10]{vershynin2018high}, we obtain
    \begin{equation}\label{append:eq:packing_each_Omega_tau}
        |\mathcal{P}_{\eta}(\mathcal{F}(\Lambda_{\tau}),\varrho)|\le \left( \frac{2\|\bm{c}\|_2}{\eta} + 1\right)^{|\Lambda_{\tau}|}
        \le \left( \frac{2}{\eta} + 1\right)^{\left\lceil {\tau}^{-1} \right \rceil}.
    \end{equation}
    In conjunction with Eqs.~\eqref{append:eq:union_packing},~\eqref{append:eq:all_Omega_tau} and \eqref{append:eq:packing_each_Omega_tau}, the maximal $\eta$-packing number of the hypothesis space $\mathcal{F}_{\tau}$ yields
    \begin{equation}\label{append:eq:packing_total}
        |\mathcal{P}_{\epsilon}(\mathcal{F}_{\tau},\varrho)|
        \le
        \left( \frac{e 3^d}{ \left\lceil {\tau}^{-1} \right \rceil} \right)^{ \left\lceil {\tau}^{-1} \right \rceil}
        \cdot
        \left( \frac{2}{\eta} + 1\right)^{\left\lceil {\tau}^{-1} \right \rceil}.
    \end{equation}

    \smallskip
    \noindent\underline{Upper bound on the sample complexity}. According to Lemma~\ref{lem:est_err_sse_bound_1_order}, to achieve the desired prediction error it suffices to set the threshold as $\tau=2^{-\nu_f \mathcal{M}^{(1)}[O(\bx;U)]/\eta^2}$ in Eq.~\eqref{append:eq:tau_choice}. Substituting this choice into Eq.~\eqref{append:eq:packing_total}, the metric entropy of the packing net satisfies
    \begin{align}
        & \log\left(4 |\mathcal{P}_{\eta}(\mathcal{F}_{\tau},\varrho)|/\delta \right)
        \nonumber \\
        \le ~
        & \log \left( \frac{e 3^d}{ \left\lceil {\tau}^{-1} \right \rceil} \right)^{ \left\lceil {\tau}^{-1} \right \rceil}
        + \log \left( \frac{2}{\eta} + 1\right)^{\left\lceil {\tau}^{-1} \right \rceil}
        + \log\left( \frac{4}{\delta}\right)
        \nonumber \\
        \le ~
        &  \left\lceil 2^{\nu_f \mathcal{M}^{(1)}[O(\bx;U)]/\eta^2} \right \rceil
        \left[
        \log \left( \frac{e 3^d}{ \left\lceil {\tau}^{-1} \right \rceil} \right)
        + \log \left( \frac{2}{\eta} + 1\right)
        \right]
        + \log\left( \frac{4}{\delta}\right)
        \nonumber \\
        \le ~
        &  \tilde{\mathcal{O}} \left( 2^{\nu_f \mathcal{M}^{(1)}[O(\bx;U)]/\eta^2} \cdot d  \right).
    \end{align}
    Therefore, applying  Proposition~\ref{append:prop:n_upper_bound_packing} with $\eta^2 = \epsilon/12$ and $\mathcal{M}^{(1)}[O(\bx;U)]\le \mathcal{M}_{\lambda}$ for any $U\in\mathcal{U}_{\lambda}$ implies that the sample complexity satisfies
    \begin{equation}
        n = \tilde{\mathcal{O}} \left( \frac{2^{12\nu_f \mathcal{M}_{\lambda}/\epsilon} \cdot d}{\epsilon}  \right),
    \end{equation}
    In particular, with probability at least $1 - \delta$, Proposition~\ref{append:prop:n_upper_bound_packing} gives the prediction error bound  $\mathsf{R}(h_{\mathcal{T}}) \le  12\eta^2 =\epsilon$.
\end{proof}

\subsection{Proof of Lemma~\ref{lem:est_err_sse_bound_1_order}}\label{append:subsec:ub_lemma}

\begin{proof}[Proof of Lemma~\ref{lem:est_err_sse_bound_1_order}]
    The main idea of this proof is to relate the truncation error to the tail probability of the induced distribution $\mathrm{p}(\bomega)$, and then control this tail probability by applying Markov's inequality to the random variable $-\log(\mathrm{p}(\bomega))$, whose expectation is precisely the first-order \DSE. In particular, we first note that
    \begin{align}\label{append:eq:trunc_err_tail}
         \mathbb{E}_{\bx\in [-\pi,\pi]^d} \left|f_O(\bx)-f_{\tau}(\bx)\right|^2 \le & \mathbb{E}_{\bx\in [-\pi,\pi]^d} \left|\sum_{\bomega:\mathrm{p}(\bomega)<\tau} \Phi_{\bomega}(\bx)\Tr(\rho_0 Q_{\bomega}) \right|^2
         \nonumber \\
         = & \mathbb{E}_{\bx\in [-\pi,\pi]^d} \sum_{\bomega:\mathrm{p}(\bomega)<\tau} \Phi_{\bomega}(\bx)^2\Tr(\rho_0 Q_{\bomega})^2
         \nonumber \\
         = & \nu_f \sum_{\bomega:\mathrm{p}(\bomega)<\tau} \mathrm{p}(\bomega),
    \end{align}
    where the first inequality follows the property of the function $g(z)=\min\{1,\max\{-1,z\}\}$ that $|g(z)-y|\le |z-y|$ holds for any $|y|\le \|O\|_{\infty}$ and $z$, the first equality employs the orthogonality of the trigonometric basis $\Phi_{\bomega}(\bx)$, and the last equality follows from the definition of $\mathrm{p}(\bomega)$ in Eq.~\eqref{append:eq:ub_prob}.

    Now, we show how to control the tail probability using \DSE. The upper bound of the truncation error can then be derived by applying Markov's inequality to the distribution $\mathrm{p}(\bomega)$. In particular, let $X:=-\log(\mathrm{p}(\bomega))$, where $\bomega \in \{0,\pm 1\}^{d}$ follows the distribution $\mathrm{p}(\bomega)$. Its expectation is exactly the first-order entropy, i.e.,
    \begin{equation}
        \mathbb{E}[X]=\sum_{\bomega} \mathrm{p}(\bomega) \log\left(\frac{1}{\mathrm{p}(\bomega)} \right) = \mathcal{S}^{(1)}[O(\bx;U)].
    \end{equation}
    Hence, employing Markov's inequality, we obtain that for any $a>0$,
    \begin{equation}
        \Pr_{\bomega \sim \mathrm{p}}\big(-\log(\mathrm{p}(\bomega))\ge a\big) \le \frac{\mathcal{S}^{(1)}[O(\bx;U)]}{a}
        ~~\Longleftrightarrow~~
        \Pr_{\bomega \sim \mathrm{p}}\big(\mathrm{p}(\bomega)\le \tau\big) \le \frac{\mathcal{S}^{(1)}[O(\bx;U)]}{\log(1/\tau)},
    \end{equation}
    where the right-hand side follows by setting $a=\log(1/\tau)$.

    Substituting the above tail bound into Eq.~\eqref{append:eq:trunc_err_tail} and letting $\tau=2^{-\nu_f \mathcal{M}^{(1)}[O(\bx;U)]/\eta^2}$, we arrive at
    \begin{equation}
        \nu_f \sum_{\bomega:\mathrm{p}(\bomega)<\tau} \mathrm{p}(\bomega)
        = \nu_f\cdot \Pr_{\bomega \sim \mathrm{p}}\big(\mathrm{p}(\bomega)\le \tau\big)
        \le \nu_f\cdot \frac{\mathcal{S}^{(1)}[O(\bx;U)]}{\log(1/\tau)}
        \le \nu_f\cdot \frac{\mathcal{M}^{(1)}[O(\bx;U)]}{\log(1/\tau)}
        \le \eta^2,
    \end{equation}
    where the second inequality follows from the definition of $\mathcal{M}^{(1)}[O(\bx;U)]$ as the maximum of $\mathcal{S}^{(1)}$ over the fixed unitaries, and the last inequality employs $\log(1/\tau)=\nu_f \mathcal{M}^{(1)}[O(\bx;U)]/\eta^2$. This completes the proof.
\end{proof}

\section{\DSE-dependent lower bound of sample complexity---proof of Theorem~\ref{append:thm:lower_bound_complexity}}\label{append:sec:dse_lower_bound}

This section aims to derive the lower bound on the sample complexity of classically learning the function class $\mathcal{F}_{\mathcal{U}_{\lambda}}$ with a bounded \DSE $\mathcal{M}_{\lambda}$ in Theorem~\ref{append:thm:lower_bound_complexity}. In particular, we prove Theorem~\ref{append:thm:lower_bound_complexity} by means of an information-theoretic analysis. Conceptually, the core technical tool is Fano's method, which is widely used for analyzing the lower bound of sample complexity in statistical learning theory and quantum learning theory. We follow the conventions of our prior work \cite{wang2024transition} to organize the proof in the following section. In SI.~\ref{append:sec:learning_problem}, we first formulate the learning problem by constructing a discrete subset of the function class $\tilde{\mathcal{F}}=\{f_{\bm{b}}(\bx)\}_{\bm{b}}\subset \mathcal{F}_{\mathcal{U}_{\lambda}}$ that saturates the bounded \DSE $\mathcal{M}_{\lambda}$. Then in SI.~\ref{append:subsec:reduce_learning_problem}, we reduce the learning problem to a hypothesis testing problem over a packing net $\{\bm{b}_1,\cdots,\bm{b}_M\}$ of the target function class. Under this reduction, constructing a classical surrogate amounts to identifying the target index $\bm{b}^*$ from the packing net, with the estimated index denoted by $\hat{\bm{b}}$. Finally, we apply Fano's method to derive a sample-complexity lower bound by bounding the key quantities in the resulting hypothesis testing problem, namely the upper bound on the mutual information, including the upper bound of mutual information $I(\hat{\bm{b}};\bm{b}^*)$ in SI.~\ref{append:subsec:packing_number_LB} and the lower bound of packing number $M$ in SI.~\ref{append:subsec:MI_UB}.

    \subsection{Learning problem formulation}\label{append:sec:learning_problem}
    To reach Theorem~\ref{append:thm:lower_bound_complexity}, we constitute a class of quantum circuits $\tilde{\mathcal{U}}\subset \mathcal{U}_{\lambda}$ whose maximal \DSE is $\mathcal{M}_{\lambda}$ satisfying $1\le \mathcal{M}_{\lambda}\le \min\{d,N/2\}$ and analyze the lower bound of the required sample complexity for learning the function class $\mathcal{F}_{\tilde{\mathcal{U}}}$ in $\epsilon$-prediction error. For convenience, we denote $\lfloor \mathcal{M}_{\lambda} \rfloor$ as the round down to the nearest integer of $\mathcal{M}_{\lambda}$. Considering the observable $O=X_1$ and the input state $\rho_{0}=\ket{0}\bra{0}^{\otimes N}$, we construct a quantum circuit class $\tilde{\mathcal{U}}=\{U_{\bm{b}}(\bx)\}_{\bm{b}}$ indexed by a $2^{\lfloor \mathcal{M}_{\lambda} \rfloor}$-dimensional sign vector $\bm{b}\in \{-1,1\}^{2^{\lfloor \mathcal{M}_{\lambda} \rfloor}}$, such that the related function class is given by\begin{equation}\label{append:eq:bit_target_function}
        \mathcal{F}_{\tilde{\mathcal{U}}}=\left\{f_{\bm{b}}(\bx)=\Tr(U_{\bm{b}}(\bx)^{\dagger}OU_{\bm{b}}(\bx)\rho_0)=\frac{1}{2^{\lfloor \mathcal{M}_{\lambda} \rfloor}} \sum_{\bomega\in \{-1,1\}^{\lfloor \mathcal{M}_{\lambda} \rfloor}} \bm{b}_{\bomega} \cdot \Phi_{\bomega}(\bx)~\bigg|~\bm{b}\in \{-1,1\}^{2^{\lfloor \mathcal{M}_{\lambda} \rfloor}}\right\},
    \end{equation}
    where the $2^{\lfloor \mathcal{M}_{\lambda} \rfloor}$-dimensional sign vector $\bm{b}=[\bm{b}_{\bomega}]_{\bomega\in\{-1,1\}^{\lfloor \mathcal{M}_{\lambda} \rfloor}} \in \{-1,1\}^{2^{\lfloor \mathcal{M}_{\lambda} \rfloor}}$ determine the construction of fixed gates in the circuit $U_{\bm{b}}(\bx)$.  To this end, we consider the $N$-qubit quantum circuit containing of $d$ rotation gates $\RZ$ and $G-d$ fixed gates, which is given as follows
    \begin{equation}\label{append:eq:U_b_x}
        U_{\bm{b}}(\bx)=W(\bx) \cdot  V_{\bm{b}}:=\prod_{i=2}^{\lfloor \mathcal{M}_{\lambda} \rfloor}\CNOT_{1,i} \cdot \prod_{i=1}^{d} S_i^{\dagger} \RZ_i(\bx) \cdot V_{\bm{b}} , 
    \end{equation}
    where $\CNOT_{1,i}$ is the control-X gate with control and target qubit indices being $1$ and $i$, $S_i^{\dagger}$ and $\RZ_i(\bx)$ represent the Phase gate and rotation gate acting on the $i$-th qubit, $V_{\bm{b}}$ will be detailed later. We first note that applying $W(\bx)$ on the observable $O=X_1$ yields
    \begin{align}
            W(\bx)^{\dagger} X_1 W(\bx) = & \left(  \prod_{i=1}^{d} S_i^{\dagger} \RZ_i(\bx)\right)^{\dagger} \cdot \left( \prod_{i=2}^{\lfloor \mathcal{M}_{\lambda} \rfloor}\CNOT_{1,i} X_1 \prod_{i=2}^{\lfloor \mathcal{M}_{\lambda} \rfloor}\CNOT_{1,i} \right)\cdot \prod_{i=1}^{d} S_i^{\dagger} \RZ_i(\bx)
            \nonumber \\
            = &  \prod_{i=1}^{\lfloor \mathcal{M}_{\lambda} \rfloor} \left( \RZ_i^{\dagger}(\bx)S_i \cdot X_i \cdot S_i^{\dagger} \RZ_i(\bx) \right)
            \nonumber \\
            = &  \sum_{\bomega\in \{1,-1\}^{\lfloor \mathcal{M}_{\lambda} \rfloor}} \Phi_{\bomega}(\bx) Q_{\bomega},
        \end{align}
    where the equalities can be obtained by using the direct transformation of applying the elementary Clifford gate and rotation gates to the Pauli terms. Here, $Q_{\bomega}=\bigotimes_{i=1}^{\lfloor \mathcal{M}_{\lambda} \rfloor} P_{i,\bomega_{i}}$ and $P_{i,\bomega_{i}}=\mathbbm{1}_{\bomega_{i}=-1}\cdot X_{i} + \mathbbm{1}_{\bomega_{i}=1}\cdot Y_{i}$ refers to the Pauli operator acting on the $i$-th qubit. Moreover, we consider the fixed gate blocks $V_{\bm{b}}$ indexed by $\bm{b}$ such that applying $V_{\bm{b}}$ to the input state $\ket{0}^{\otimes N}$,
    \begin{equation}
        \rho_{\bm{b}}:=\Tr_{\lfloor \mathcal{M}_{\lambda} \rfloor+1:N}\left[V_{\bm{b}}\ket{0}\bra{0}^{\otimes N} V_{\bm{b}}^{\dagger} \right] = \frac{1}{2^{\lfloor \mathcal{M}_{\lambda} \rfloor}}\cdot \left( \mathbb{I}_{2^{\lfloor \mathcal{M}_{\lambda} \rfloor}} + \frac{1}{2^{\lfloor \mathcal{M}_{\lambda} \rfloor}} \sum_{\bomega\in \{-1,1\}^{\lfloor \mathcal{M}_{\lambda} \rfloor}}\bm{b}_{\bomega} Q_{\bomega} \right).
    \end{equation}
    One can directly check that $\rho_{\bm{b}}$ is a valid mixed state as a Hermitian, positive semi-definite matrix with unit trace $\Tr(\rho_{\bm{b}})=1$. More importantly, we have
    \begin{align}\label{append:eq:tr_equal_fbx}
        \Tr\left[\rho_0{U}_{\bm{b}}(\bx)^{\dagger} X_1 {U}_{\bm{b}}(\bx)\right]
        = & \frac{1}{2^{\lfloor \mathcal{M}_{\lambda} \rfloor}}\cdot \Tr\left[\left(\mathbb{I}_{2^{\lfloor \mathcal{M}_{\lambda} \rfloor}} + \frac{1}{2^{\lfloor \mathcal{M}_{\lambda} \rfloor}} \sum_{\bomega\in \{-1,1\}^{\lfloor \mathcal{M}_{\lambda} \rfloor}}\bm{b}_{\bomega} Q_{\bomega} \right) \cdot  \sum_{\bomega\in \{1,-1\}^{\lfloor \mathcal{M}_{\lambda} \rfloor}} \Phi_{\bomega}(\bx) Q_{\bomega}\right]
        \nonumber \\
        = & \frac{1}{2^{\lfloor \mathcal{M}_{\lambda} \rfloor}}\cdot \frac{1}{2^{\lfloor \mathcal{M}_{\lambda} \rfloor}}\cdot  \sum_{\bomega\in \{-1,1\}^{\lfloor \mathcal{M}_{\lambda} \rfloor}} \sum_{\bomega'\in \{-1,1\}^{\lfloor \mathcal{M}_{\lambda} \rfloor}} \bm{b}_{\bomega} \Phi_{\bomega'}(\bx) \Tr\left(Q_{\bomega} Q_{\bomega'} \right)
        \nonumber \\
        = & \frac{1}{2^{\lfloor \mathcal{M}_{\lambda} \rfloor}}\cdot  \sum_{\bomega\in \{-1,1\}^{\lfloor \mathcal{M}_{\lambda} \rfloor}}  \bm{b}_{\bomega} \Phi_{\bomega}(\bx)
        \nonumber \\
        = & f_{\bm{b}}(\bx),
    \end{align}
    where the third equality follows that $Q_{\bomega} \ne Q_{\bomega'}\in \{\mathbb{I},X,Y,Z\}^{\lfloor \mathcal{M}_{\lambda} \rfloor}$ for $\bomega\ne\bomega' $ and hence $\Tr(Q_{\bomega}Q_{\bomega'})=2^{\lfloor \mathcal{M}_{\lambda} \rfloor}\cdot \delta_{\bomega,\bomega'}$. Moreover, employing the orthogonality of trigonometric basis $\Phi_{\bomega}(\bx)$, we have
    \begin{equation}\label{append:eq:p0_scaling_LB}
        \nu_f = \mathbb{E}_{\bx\sim [-\pi,\pi]^d} f_{\bm{b}}(\bx)^2 = \frac{1}{2^{2\lfloor \mathcal{M}_{\lambda} \rfloor}} \sum_{\bomega\in \{-1,1\}^{\lfloor \mathcal{M}_{\lambda} \rfloor}}  \mathbb{E}_{\bx\sim [-\pi,\pi]^d}  \Phi_{\bomega}(\bx)^2 = \frac{1}{2^{2\lfloor \mathcal{M}_{\lambda} \rfloor}}
    \end{equation}
    By construction, the expectation value $\Tr(\rho_{0}{U}_{\bm{b}}(\bx)^{\dagger} X_1 {U}_{\bm{b}}(\bx))$ exactly reproduces the function $f_{\bm{b}}(\bx)$ defined in Eq.~\eqref{append:eq:bit_target_function}. As a result, each choice of $\bm{b}\in \{-1,1\}^{2^{\lfloor \mathcal{M}_{\lambda} \rfloor}}$ yields a target function $f_{\bm{b}}(\bx)$ with a varying $d$-dimensional input vector $\bx\in [-\pi,\pi]^d$. Being able to accurately estimate the sign vector $\bm{b}$ will be equivalent to accurately learning this function. 

    \medskip
    
    \noindent \underline{Remark.} Let $k=\lfloor\mathcal{M}_{\lambda}\rfloor$. Although an explicit gate-by-gate decomposition of $V_{\bm b}$ is unnecessary for the information-theoretic lower bound, its existence and a conservative gate-count upper bound can be made explicit. The $k$-qubit state $\rho_{\bm b}$ has rank $r_{\bm b}\le 2^k$ and therefore admits a purification on $k+\lceil\log_2 r_{\bm b}\rceil\le 2k$ qubits. The condition $\mathcal{M}_{\lambda}\le N/2$ guarantees $N\ge 2k$, so one may choose such a purification $\ket{\Psi_{\bm b}}$ and extend the map $\ket{0}^{\otimes N}\mapsto\ket{\Psi_{\bm b}}\otimes\ket{0}^{\otimes(N-2k)}$ to an $N$-qubit unitary $V_{\bm b}$. Standard universal state-preparation decompositions prepare an arbitrary state on $2k$ qubits using $\mathcal{O}(2^{2k})=\mathcal{O}(4^k)$ one- and two-qubit gates~\cite{plesch2011quantum}. Together with the $\mathcal{O}(d+k)$ gates in $W(\bx)$, this gives the conservative bound $G=\mathcal{O}\!\left(d+k+4^k\right)   =\widetilde{\mathcal{O}}\!\left(d+4^{\lfloor\mathcal{M}_{\lambda}\rfloor}\right).
    $

    Given these formulations, we can immediately derive the probability distribution $\mathrm{p}(\bomega)$ induced by the observable expectation function $f_{\bm{b}}(\bx)$ over the uniform distribution $\bx\in[-\pi,\pi]^d$ and its \DSE. The results are encapsulated in the following lemma.

    \begin{lemma}\label{append:lem:qc_ld_ose}
        Following notations defined in Eqs.~\eqref{append:eq:bit_target_function}-\eqref{append:eq:tr_equal_fbx} with the definition of
        $f_{\bm{b}}(\bx)=2^{-\lfloor \mathcal{M}_{\lambda} \rfloor}\sum_{\bomega} \bm{b}_{\bomega} \cdot \Phi_{\bomega}(\bx)$,
        where $\bm{b}=[\bm{b}_{\bomega}]_{\bomega}\in \{-1,1\}^{2^{\lfloor \mathcal{M}_{\lambda} \rfloor}}$ refers to a $2^{\lfloor \mathcal{M}_{\lambda} \rfloor}$-dimensional vector.
        The probability distribution $\mathrm{p}(\bomega)$ induced by  $f_{\bm{b}}(\bx)$ over the uniform distribution $\bx\in [-\pi,\pi]^d$ for varying $\bm{b}$ is the same, and is given by
        \begin{equation}
            \mathrm{p}(\bomega) = 
            \begin{cases}
        		 2^{-\lfloor \mathcal{M}_{\lambda} \rfloor} ~ & \textnormal{if}~ \bomega  \in \{1,-1\}^{\lfloor \mathcal{M}_{\lambda} \rfloor}  \\
        		 0 & \textnormal{if}~\bomega  \notin \{1,-1\}^{\lfloor \mathcal{M}_{\lambda} \rfloor}
        	\end{cases},
        \end{equation}
        and for any $U_{\bm{b}}(\bx)$, the related dynamical stabilizer entropy (\DSE) yields
        $$\mathcal{M}^{(1)}[O(\bx;U_{\bm{b}})] = \lfloor \mathcal{M}_{\lambda} \rfloor.$$
    \end{lemma}

    \begin{proof}
        [Proof of Lemma~\ref{append:lem:qc_ld_ose}]
        According to the definition of the induced distribution $\mathrm{p}(\bomega)$, we have for $\bomega\in\{1,-1\}^{\lfloor \mathcal{M}_{\lambda} \rfloor}$,
        \begin{equation}
            \mathrm{p}(\bomega)=\frac{\mathbb{E}_{\bx\in [-\pi,\pi]^d} \left(2^{-{\lfloor \mathcal{M}_{\lambda} \rfloor}} \Phi_{\bomega}(\bx)\right)^2}{\sum_{\bomega\in \{1,-1\}^{\lfloor \mathcal{M}_{\lambda} \rfloor}}\mathbb{E}_{\bx\in [-\pi,\pi]^d} \left(2^{-\lfloor \mathcal{M}_{\lambda} \rfloor} \Phi_{\bomega}(\bx)\right)^2 } =\frac{2^{-2\lfloor \mathcal{M}_{\lambda} \rfloor} 2^{-\lfloor \mathcal{M}_{\lambda} \rfloor} }{ \sum_{\bomega\in \{1,-1\}^{\lfloor \mathcal{M}_{\lambda} \rfloor}} 2^{-2\lfloor \mathcal{M}_{\lambda} \rfloor} 2^{-\lfloor \mathcal{M}_{\lambda} \rfloor}}=\frac{1}{2^{\lfloor \mathcal{M}_{\lambda} \rfloor}},
        \end{equation}
        where the second equality employs that $\mathbb{E}_{\bx\in [-\pi,\pi]^d} \Phi_{\bomega}(\bx)^2=2^{-\|\bomega\|_0}$, and the third equality follows that $\|\bomega\|_0=\lfloor \mathcal{M}_{\lambda} \rfloor$ for any $\bomega\in\{-1,1\}^{\lfloor \mathcal{M}_{\lambda} \rfloor}$.
        Moreover, the expansion of $f(\bx)$ over the feature $\Phi_{\bomega}(\bx)$ with the support set $\bomega \in \{1,-1\}^{\lfloor \mathcal{M}_{\lambda} \rfloor}$ immediately implicates that for $\mathrm{p}(\bomega)=0$ for $\bomega\notin \{1,-1\}^{\lfloor \mathcal{M}_{\lambda} \rfloor}$.

        For the identity choice of the fixed unitary in the maximization defining DSE, the above uniform distribution $\mathrm{p}(\bomega)$ over $2^{\lfloor \mathcal{M}_{\lambda} \rfloor}$ modes has Shannon entropy $\lfloor \mathcal{M}_{\lambda} \rfloor$. Therefore, $\mathcal{M}^{(1)}[O(\bx, U_{\bm{b}})] \ge \lfloor \mathcal{M}_{\lambda} \rfloor$. To prove the converse bound, observe from the circuit construction that the full evolved observable has the operator-valued expansion
        $$U_{\bm{b}}(\bx)^{\dagger} X_1 U_{\bm{b}}(\bx) =\sum_{\bomega\in \{-1,1\}^{\lfloor \mathcal{M}_{\lambda} \rfloor}} \Phi_{\bomega}(\bx)V_{\bm{b}}^{\dagger}Q_{\bomega}V_{\bm{b}}.$$
        Hence, for any fixed, parameter-independent unitary $V$, conjugating this observable only replaces each coefficient operator $V_{\bm{b}}^{\dagger}Q_{\bomega}V_{\bm{b}}$ by $V^{\dagger}V_{\bm{b}}^{\dagger}Q_{\bomega}V_{\bm{b}} V$ and does not introduce any additional basis function $\Phi_{\bomega}(\bx)$. Consequently, for any $V$, the induced distribution $\mathrm{p}(\bomega)$ is supported on a subset of $\{-1, 1\}^{\lfloor \mathcal{M}_{\lambda} \rfloor}$ and its Shannon entropy is at most $\log(2^{\lfloor \mathcal{M}_{\lambda} \rfloor}) = \lfloor \mathcal{M}_{\lambda} \rfloor$. Maximizing over all fixed $V$ thus gives $\mathcal{M}^{(1)}[O(\bx;U_{\bm{b}})] \le \lfloor \mathcal{M}_{\lambda} \rfloor$. Combining the two bounds proves $\mathcal{M}^{(1)}[O(\bx;U_{\bm{b}})] = \lfloor \mathcal{M}_{\lambda} \rfloor$. This completes the proof.
    \end{proof}

    \subsection{Reducing the learning problem to hypothesis testing---Proof of Theorem~\ref{append:thm:lower_bound_complexity}}\label{append:subsec:reduce_learning_problem}
    In this section, we elucidate how to reduce the learning problem of observable expectation $f_{\bm{b}}(\bx)$ introduced in the previous section to a hypothesis testing problem by exploiting Fano’s method, which is extensively used in deriving the information-theoretical lower bound of both classical and quantum learning tasks. In particular, the accessible information for this learning problem is specified by the conditions in Theorem~\ref{append:thm:lower_bound_complexity} that one can access (i) a training dataset $\{(\bxi,y_{\bm{b}}^{(i)})\}_{i=1}^n$ of size $n$ with $y^{(i)}$ is the unbiased estimation of $f_{\bm{b}}(\bxi)$ from $m$ measurement shots, and (ii) the probability distribution $\mathrm{p}(\bomega)$ induced by the observable expectation $f_{\bm{b}}(\bx)$ over the uniform distribution $\bx\sim [-\pi,\pi]^d$. As specified in Lemma~\ref{append:lem:qc_ld_ose}, for varying $\bm{b}$, we have $\mathrm{p}(\bomega)=2^{-{\lfloor \mathcal{M}_{\lambda} \rfloor}}$ for $\bomega\in \{-1,1\}^{\lfloor \mathcal{M}_{\lambda} \rfloor}$. 
    The learner aims to construct a machine learning model $h(\bx)$ to achieve the desired prediction error
    \begin{equation}\label{append:eq:target_error}
        \mathbb{E}_{\bx\in [-\pi,\pi]^d}  \left| h(\bx)-f_{\bm{b}}(\bx)\right|^2 \le  \epsilon.
    \end{equation}
    where the equality follows the definition of $f_{\bm{b}}$ and the orthogonality of the trigonometric basis $\Phi_{\bomega}(\bx)$.

    The first step is to construct a $2\sqrt{\epsilon}$-packing $
    \{f_{\bm{b}^{(j)}}\}_{j\in [M]}$ for the function class $\{f_{\bm{b}}:\bm{b}\in\{-1,1\}^{2^{\lfloor \mathcal{M}_{\lambda} \rfloor}}\}$ with $M$ being the packing number such that the function $f_{\bm{b}^{(j)}}$ in this packing net is $2\sqrt{\epsilon}$-separate, namely $\varrho(f_{\bm{b}^{(i)}},f_{\bm{b}^{(j)}})\ge 2\sqrt{\epsilon}$ for any $i\ne j\in [M]$ under specific distance metric $\varrho$, which will be specified later. This step aims to ensure that one can identify the unique sign vector $\bm{b}$ when the desired prediction error is achieved. Recall the definition of the distance metric $\varrho$ as 
    \begin{equation}
        \varrho(f,f')=\left(\mathbb{E}_{\bx}|f(\bx)-f'(\bx)|^2\right)^{\frac{1}{2}},
    \end{equation}
    and if $\varrho(f_{\bm{b}},f_{\bm{b}'})\ge 2\sqrt{\epsilon}$ for any $\bm{b}'\in \{\bm{b}^{(j)}\}_{j\in [M]}$, then employing the triangle inequality yields
    \begin{equation}\label{append:eq:h_b_prime_lb}
        \mathbb{E}_{\bx}  \left| h(\bx)-f_{\bm{b}'}(\bx)\right|^2 \ge \left(\sqrt{\mathbb{E}_{\bx} \left| f_{\bm{b}'}(\bx)-f_{\bm{b}}(\bx)\right|^2} - \sqrt{\mathbb{E}_{\bx} \left| h(\bx)-f_{\bm{b}}(\bx)\right|^2}\right)^2  \ge \left(2\sqrt{\epsilon}-\sqrt{\epsilon}\right)^2 \ge \epsilon.
    \end{equation}

    Given this $2\sqrt{\epsilon}$-packing set $
    \{f_{\bm{b}^{(j)}}\}_{j\in [M]}$ with separable concept function, learning this family of quantum circuits under the conditions of Theorem~\ref{append:thm:lower_bound_complexity} can be reduced to the following multiple hypothesis testing problem:
    \begin{itemize}
        \item[(i).] Alice randomly and uniformly samples a $2^d$-dimensional sign vector $\bm{b}^*$ from $\{\bm{b}^{(j)}\}_{j\in [M]}\subset \{-1,1\}^{2^{{\lfloor \mathcal{M}_{\lambda} \rfloor}}}$ and denote the target concept as $f_{\bm{b}^*}$;
        \item[(ii).] Alice implements the quantum circuits related to $f_{\bm{b}^*}$ to collect the training dataset $\mathcal{T}=\{(\bxi, y^{(i)})\}_{i=1}^n$ of size $n$ and construct a quantum sampler of $\mathrm{p}({\bomega})$, where $\bxi$ are randomly sampled from $[-\pi,\pi]^d$, $\bm{o}(\bxi)=[\bm{o}_1(\bxi),\cdots,\bm{o}_m(\bxi)]\in \{-1,1\}^m$ is the $m$ measurement outcomes of the observable $O=X_1$, and $y^{(i)}=\frac{1}{m}\sum_k \bm{o}_{k=1}^m(\bxi)$ is the statistical estimation of $f_{\bm{b}^*}(\bxi)$ with $\mathbb{E} y^{(i)}=f_{\bm{b}^*}(\bxi)$;
        \item[(iii).] Bob uses the training dataset $\mathcal{T}$ and the probability distribution $\mathrm{p}(\bomega)$ to construct the classical surrogate by conducting the empirical risk minimization, i.e.,
        \begin{equation}
            h_{\mathcal{T}} = \mathop{\arg\min}\limits_{h\in  
            \{f_{\bm{b}^{(j)}}|j\in [M]\} } \frac{1}{n} \sum_{i=1}^n \left( h(\bxi)-y^{(i)}\right)^2;
        \end{equation}
        \item[(iv).] Bob conducts hypothesis testing over $\{\bm{b}^{(j)}\}_{j\in[M]}$ to infer the randomly chosen $\bm{b}^*$, i.e.,
        \begin{equation}
            \bar{\bm{b}} = \mathop{\arg\min}\limits_{\bm{b}\in \{\bm{b}^{(j)}\}_{j\in[M]}} \mathbb{E}_{\bx\in [-\pi,\pi]^d} \left|h_{\mathcal{T}}(\bx)-f_{\bm{b}}(\bx) \right|^2,
        \end{equation}
        and the associated error probability is denoted by $\mathbb{P}(\bar{\bm{b}}\ne \bm{b}^*)$.
    \end{itemize}
    We note that the utilization of the probability distribution $\mathrm{p}(\bomega)$ in step (iii) for the construction of the classical surrogate $h_{\mathcal{T}}(\bx)$ is enabled by the optimization of $h_{\mathcal{T}}(\bx)$ over the hypothesis space $\{f_{\bm{b}}(\bx)\}_{\bm{b}}$ whose induced probability distribution is exactly $\mathrm{p}(\bomega)$ as shown in Lemma~\ref{append:lem:qc_ld_ose}. 
    By Eq.~\eqref{append:eq:h_b_prime_lb}, the decoder returns $\bar{\bm{b}}=\bm{b}^*$ whenever
    $\mathbb{E}_{\bx \sim [-\pi,\pi]^d}|h_{\mathcal{T}}(\bx)-f_{\bm{b}^*}(\bx)|^2 \le \epsilon$. Hence,
    \begin{equation}
        \mathbb{P}\left(\bar{\bm{b}}\ne\bm{b}^* \right) \le 1-\mathbb{P}\left(\mathbb{E}_{\bx \sim [-\pi,\pi]^d}|h_{\mathcal{T}}(\bx)-f_{\bm{b}^*}(\bx)|^2 \le \epsilon \right).
    \end{equation}
    In particular, a learning guarantee holding with probability at least $2/3$ implies $ \mathbb{P}\left(\bar{\bm{b}}\ne\bm{b}^* \right)\le 1/3$

    With this setup, Fano's inequality gives an information-theoretical lower bound on the error probability $\mathbb{P}(\bar{\bm{b}}\ne \bm{b}^*)$ for the multiple hypothesis testing problem.

    \begin{lemma}
        [Fano's inequality, Ref.~\cite{duchi2016lecture}]
        \label{append:lem:fano_inequality}
        Assume that $\bm{b}^*$ is uniform in $\{\bm{b}^{(j)}\}_{j\in [M]}$. The learning procedure can be depicted by the Markov chain $\bm{b}^* \rightarrow \mathcal{T} \rightarrow  \bar{\bm{b}}$, where $\bar{\bm{b}}$ is returned by the hypothesis testing. Then we have
        \begin{equation}
        \mathbb{P}(\bar{\bm{b}}\ne \bm{b}^*) \ge 1 - \frac{I(\bar{\bm{b}};\bm{b}^*)+\ln(2)}{\log(M)},
        \end{equation}
        where $I(\bar{\bm{b}};\bm{b}^*)$ represents the mutual information between the estimated $\bar{\bm{b}}$ and the randomly sampled variable $\bm{b}^*$, $M$ refers to the cardinality of the $2\epsilon$-packing net $\{f_{\bm{b}^{(j)}}\}_{j\in [M]}$.
    \end{lemma}

   The Fano's inequality implies that the derivation of the lower bound of the sample complexity $n$ amounts to separately deriving the upper bound of the mutual information $I(\bar{\bm{b}};\bm{b}^*)$ and the lower bound of the cardinality $M$, which are given by the following two lemmas.

   \begin{lemma}
       [Lower bound of $2\sqrt{\epsilon}$-packing cardinality $M$]
       \label{append:lem:cardinality_lb}
       Let $\{f_{\bm{b}}:\bm{b}\in \{-1,1\}^{2^{\lfloor \mathcal{M}_{\lambda} \rfloor}}\}$ be the class of target function with $f_{\bm{b}}$ defined in Eq.~\eqref{append:eq:bit_target_function}, and $\varrho(f_{\bm{b}},f_{\bm{b}'})=\sqrt{\mathbb{E}_{\bx\sim [-\pi,\pi]^d}|f_{\bm{b}}(\bx)-f_{\bm{b}'}(\bx)|^2}$ be the distance measure. Then there exists a $2\sqrt{\epsilon}$-packing $\{f_{\bm{b}^{(j)}}\}_{j\in [M]}$ in the $\varrho$-metric such that for any $\epsilon \le \nu_f/4$, the $2\sqrt{\epsilon}$-packing number yields
       \begin{equation}
           M\ge \exp\left(\frac{2^{\lfloor \mathcal{M}_{\lambda} \rfloor}\cdot (1-2\nu_f^{-1} \epsilon)^2}{8} \right) \ge \exp\left(\frac{2^{\lfloor \mathcal{M}_{\lambda} \rfloor}}{32} \right)
       \end{equation}
   \end{lemma}

   \begin{lemma}
       [Upper bound of the mutual information $I(\bar{\bm{b}};\bm{b}^*)$]
       \label{append:lem:mutual_information_up}
       Following the notations in Lemma~\ref{append:lem:fano_inequality} and Lemma~\ref{append:lem:cardinality_lb}, the mutual information between the random variables $\bar{\bm{b}}$ and $\bm{b}^*$ yields
       \begin{equation}
           I(\bar{\bm{b}};\bm{b}^*)  \leq 3nm\nu_f
       \end{equation}
   \end{lemma}

   We are now ready to present the proof of Theorem~\ref{append:thm:lower_bound_complexity}.

   \begin{proof}
       [Proof of Theorem~\ref{append:thm:lower_bound_complexity}]
        Employing the Fano's inequality in Lemma~\ref{append:lem:fano_inequality}, we have
        \begin{subequations}
        \begin{eqnarray}
        && \mathbb{P}(\bar{\bm{b}}\ne \bm{b}^*) \ge 1 - \frac{I(\bar{\bm{b}};\bm{b}^*)+\log(2)}{\log(M)}
        \nonumber \\
        \Rightarrow && I(\bm{b}^*; \bar{\bm{b}})  \geq (1 - 	\Prob[\bm{b}^* \neq \bar{\bm{b}}]) \log |M| -  \ln2 
        \nonumber  \\
        \Rightarrow &&n\cdot 3m\nu_f \geq (1 - 	\Prob[\bar{\bm{b}}\neq \bm{b}^*]) \cdot 2^{\lfloor \mathcal{M}_{\lambda} \rfloor-3} \cdot (1-2\nu_f^{-1}\epsilon)^2 - \ln2 
        \nonumber \\
        \Rightarrow && n \geq \frac{(1 - 	\Prob[\bm{b}^* \neq \bar{\bm{b}}]) \cdot 2^{\lfloor \mathcal{M}_{\lambda} \rfloor-3} \cdot (1-2\nu_f^{-1}\epsilon)^2 - \ln2 }{3m\nu_f},
         \end{eqnarray} 
         \end{subequations}
         where the third line employs the upper bound of  $I(\bm{b}^*; \bar{\bm{b}})$ and the lower bound of the cardinality $M$ achieved in Lemma~\ref{append:lem:cardinality_lb} and Lemma~\ref{append:lem:mutual_information_up}, respectively.
         
        Hence, by setting the failure probability as $	\Prob[\bm{b}^* \neq \bar{\bm{b}}]\le1/3$, we can obtain the lower bound in terms of the dynamical stabilizer entropy (\DSE)
         \begin{equation}
         	n \geq   \tilde{\Omega}\left(\frac{2^{\mathcal{M}_{\lambda}} \cdot (1-2\nu_f^{-1}\epsilon)^2 }{ 24m\nu_f }\right),
         \end{equation} 
         where the symbol $\tilde{\Omega}$ hides constant factor. This completes the proof.
   \end{proof}

    \smallskip

    \subsection{Bounding the \texorpdfstring{$2\epsilon$}{2 epsilon}-packing cardinality \texorpdfstring{$M$}{M}---Proof of Lemma~\ref{append:lem:cardinality_lb}}
    \label{append:subsec:packing_number_LB}
    The proof of Lemma~\ref{append:lem:cardinality_lb} employs the following lemma.
    \begin{lemma}
        [Adapted from Lemma 5.12 in Ref.~\cite{rigollet2015high}] 
        \label{append:lem:cardinality_lb_1}
        For any constant $\gamma \in (0,1/2)$, there exists a set $\{\bm{b}^{(1)}, \cdots,\bm{b}^{(M)}\} \subset \{-1,1\}^{K}$ of cardinality $M\ge \exp(\gamma^2 K/2)$ such that $\mathrm{d}_{\mathrm{H}}(\bm{b}^{(i)},\bm{b}^{(j)}) \ge K(1/2-\gamma)$ for any $i\ne j$, where $\mathrm{d}_{\mathrm{H}}(\bm{b}^{(i)},\bm{b}^{(j)})=\sum_{\bomega}\mathbbm{1}\{\bm{b}_{\bomega}^{(i)}\ne\bm{b}_{\bomega}^{(j)}\}$ denotes the Hamming distance between $\bm{b}^{(i)},\bm{b}^{(j)} \in  \{-1,1\}^{K}$. If let $\gamma \in (0,1/4)$, then the Hamming distance also satisfies $\mathrm{d}_{\mathrm{H}}(\bm{b}^{(i)},\bm{b}^{(j)}) \ge K(1/2-\gamma) \ge K \gamma$.
    \end{lemma}

    \begin{proof}
        [Proof of Lemma~\ref{append:lem:cardinality_lb}]
        We derive the lower bound of the cardinality of $2\epsilon$-packing $M$ by showing the existence of a subset $\{\bm{b}^{(1)}, \cdots,\bm{b}^{(M)}\} \subset \{-1,1\}^{K}$ of cardinality $M\ge \exp(K/32)$ for $K=2^{\lfloor \mathcal{M}_{\lambda} \rfloor}$ and some constant $c$ such that for any $i,j\in [M]$, $\varrho(f_{\bm{b}^{(i)}},f_{\bm{b}^{(j)}})\ge 2\varepsilon$ for every $i\ne j$. To this end, we derive an explicit expression of $\varrho(f_{\bm{b}^{(i)}},f_{\bm{b}^{(j)}})$ in terms of the Hamming distance $\mathrm{d}_{\mathrm{H}}(\bm{b}^{(i)},\bm{b}^{(j)})$ between $\bm{b}^{(i)},\bm{b}^{(j)} \in  \{-1,1\}^{K}$. Employing the definition of $\varrho(f_{\bm{b}^{(i)}},f_{\bm{b}^{(j)}})$, we have
        \begin{align}\label{append:eq:distance_d_H}
            \varrho(f_{\bm{b}^{(i)}},f_{\bm{b}^{(j)}}) = &  \left(\mathbb{E}_{\bx\in [-\pi,\pi]^d} \left| f_{\bm{b}^{(i)}}(\bx)-f_{\bm{b}^{(j)}}(\bx)\right|^2  \right)^{\frac{1}{2}}
            \nonumber \\
            = &  \left(\mathbb{E}_{\bx\in [-\pi,\pi]^d} \frac{1}{2^{2\lfloor \mathcal{M}_{\lambda} \rfloor}} \left( \sum_{\bomega\in \{1,-1\}^{\lfloor \mathcal{M}_{\lambda} \rfloor}} (\bm{b}^{(i)}_{\bomega} -\bm{b}^{(j)}_{\bomega})\cdot \Phi_{\bomega}(\bx) \right)^2 \right)^{\frac{1}{2}}
            \nonumber \\
            = & \left(\nu_f \cdot \mathbb{E}_{\bx\in [-\pi,\pi]^d}   \sum_{\bomega,\bomega'\in \{1,-1\}^{\lfloor \mathcal{M}_{\lambda} \rfloor}} (\bm{b}_{\bomega}^{(i)} -\bm{b}_{\bomega}^{(j)}) (\bm{b}_{\bomega'}^{(i)} -\bm{b}_{\bomega'}^{(j)})\cdot \Phi_{\bomega}(\bx)  \Phi_{\bomega'}(\bx) \right)^{\frac{1}{2}}
            \nonumber \\
            = &  \left( \frac{\nu_f}{2^{\lfloor \mathcal{M}_{\lambda} \rfloor}}  \sum_{\bomega \in \{1,-1\}^{\lfloor \mathcal{M}_{\lambda} \rfloor}} (\bm{b}_{\bomega}^{(i)} -\bm{b}_{\bomega}^{(j)})^2 \right)^{\frac{1}{2}}
            \nonumber \\
            = &  \left( \frac{\nu_f}{2^{\lfloor \mathcal{M}_{\lambda} \rfloor}}  \sum_{\bomega \in \{1,-1\}^{\lfloor \mathcal{M}_{\lambda} \rfloor}} 4 \cdot \mathbbm{1}\{\bm{b}_{\bomega}^{(i)} \ne \bm{b}_{\bomega}^{(j)}\}\right)^{\frac{1}{2}}
            \nonumber \\
            = &  \left(\frac{4\nu_f}{2^{\lfloor \mathcal{M}_{\lambda} \rfloor}}  \mathrm{d}_{\mathrm{H}}(\bm{b}^{(i)},\bm{b}^{(j)})\right)^{\frac{1}{2}},
    \end{align}
    where the second equality follows the explicit scaling of $\nu_f=2^{-2\lfloor \mathcal{M}_{\lambda} \rfloor}$ derived in Eq.~\eqref{append:eq:p0_scaling_LB}, the fourth equality employs the orthogonality of the trigonometric basis, namely $\mathbb{E}_{\bx} \Phi_{\bomega}(\bx)\Phi_{\bomega'}(\bx)=\delta_{\bomega,\bomega'}/2^d$ for $\bomega\in\{-1,1\}^d$. Accordingly, the condition $\varrho(f_{\bm{b}^{(i)}},f_{\bm{b}^{(j)}})\ge 2\sqrt{\epsilon}$ is equivalent to $$\mathrm{d}_{\mathrm{H}}(\bm{b}^{(i)}, \bm{b}^{(j)}) \ge \frac{2^{\lfloor \mathcal{M}_{\lambda} \rfloor} \epsilon}{\nu_f}.$$ 
    In this regard, employing Lemma~\ref{append:lem:cardinality_lb_1} with taking $\gamma=(1-2\nu_f^{-1}\epsilon)/2$ and $\epsilon \le \nu_f/4$ implies that there exists $\{\bm{b}^{(1)}, \cdots,\bm{b}^{(M)}\} \subset \{-1,1\}^{2^{\lfloor \mathcal{M}_{\lambda} \rfloor}}$ of size $M\ge \exp(2^{\lfloor \mathcal{M}_{\lambda} \rfloor}\cdot (1-2\nu_f^{-1}\epsilon)^2/8)$ satisfying this requirement. Since $\epsilon \le \nu_f /4$, this also implies $M \ge \exp(2^{\lfloor \mathcal{M}_{\lambda} \rfloor}/32)$.
    \end{proof}

    \subsection{Bounding the mutual information \texorpdfstring{$I(\bar{\bm{b}};\bm{b}^*)$}{I(bar{b};b*)}---Proof of Lemma~\ref{append:lem:mutual_information_up}}
    \label{append:subsec:MI_UB}
    The proof idea of Lemma~\ref{append:lem:mutual_information_up} follows the conventions of Ref.~\cite{wang2024transition,du2025efficient}. Before moving to elucidate the detailed proof, we make some remarks on the labels $\{y^{(i)}\}_{i=1}^n$ in the training dataset $\mathcal{T}$.

    \smallskip
    
    \noindent \underline{Remark.} For ease of analysis, in the proof of Lemma~\ref{append:lem:mutual_information_up}, we consider the shot number $m$ for obtaining the statistical labels $y^{(i)}$ is sufficiently large such that the measured results can be approximated by the normal distribution with the mean $\mu_{\bm{b}}=f_{\bm{b}}(\bx)$. This assumption is widely utilized in quantum computation and quantum information theory \cite{wang2024transition,du2025efficient}. Besides, for each $\bm{b}$, the corresponding variance of measured results is assumed to be equal with the varied $\bx$, i.e., for all $\bx$, the variance is $\nu_{\bm{b}}=\mathbb{E}_{\bx}\nu_{\bm{b}}(\bx)$. Note that the assumption can be omitted if the primary focus is solely on the \DSE $\mathcal{M}^{(1)}[O(\bx;\mathcal{U})]$, or the number of RZ gates $d$, without considering the measurement cost $m$. In this case, the upper bound of the mutual information $I(\bar{\bm{b}};\bm{b}^*)$ can be effectively derived using Holevo’s theorem.

    \begin{proof}
        [Proof of Lemma~\ref{append:lem:mutual_information_up}]
        Recall that the process of learning $f_{\bm{b}^*}$ corresponds to a Markov chain: 
        \begin{equation}
        \bm{b}^*\rightarrow \{\rho_{\bm{b}^*}(\bxi)\}_{i=1}^n \rightarrow \{\bm{o}(\bxi)\}_{i=1}^n \rightarrow \{y^{(i)}\}_{i=1}^n \rightarrow \bar{\bm{b}},
        \end{equation}
        where $\rho_{\bm{b}^*}(\bxi)=U(\bx)\rho_{\bm{b}^*}U(\bx)^{\dagger}$ refers to the evolved quantum states with the initial state $\rho_{\bm{b}^*}$ and the unitary $U(\bx)$ defined in Appendix~\ref{append:sec:learning_problem}, $\bm{o}(\bxi)=[\bm{o}_1(\bxi),\cdots,\bm{o}_m(\bxi)]\in \{-1,1\}^m$ are the measurement outcomes for the observable $O=X_1$, and $y^{(i)}=\sum_{k=1}^m \bm{o}_k(\bxi)/m$ refers to the statistical estimation of $f_{\bm{b}^*}(\bxi)$. In this regard, employing the data processing inequality yields
        \begin{align}
        	 I\left(\bm{b}^*; \bar{\bm{b}} \right)   \leq  ~& I\left(\bm{b}^*; \{y^{(i)}\}_{i=1}^n \right)
             \nonumber \\
        	\leq ~&  \sum_{i=1}^n I\left(\bm{b}^*; y^{(i)}, \bxi\right)
            \nonumber \\
        	 = ~& \sum_{i=1}^n I\left(\bm{b}^* ;\bxi\right) + I\left(\bm{b}^* ; y^{(i)} | \bxi\right) 
             \nonumber \\
        	 = ~& \sum_{i=1}^n \mathbb{E}_{\bxi \sim [-\pi, \pi]^d} I\left(\bm{b}^* ; y^{(i)} | \bxi\right),
             \label{append:eq:mi_decompose_1}
        \end{align}	
        where the second equality employs the subadditivity of mutual information and the independence between the random variable $\bm{b}^*$ and $\{\bxi\}_{i=1}^n$, the first equality employs the chain rule of mutual information, namely $I(X;Y_1,Y_2)=I(X;Y_1)+I(X;Y_2|Y_1)$ for any random variables $X,Y_1,Y_2$, the last equality follows that the random variables $\bm{b}^*$ and $\bxi$ are sampled independently and hence $I(\bm{b}^*;\bxi)=0$ as well as $\bm{b}^* |\bxi = \bm{b}^* $

        In this regard, this reduces the problem of upper-bounding the mutual information $I(\bm{b}^*; \bar{\bm{b}} )$ to upper-bounding $I(\bm{b}^* ; y^{(i)} | \bxi)$. Employing the definition of mutual information, we have
        	\begin{align}
        		I(\bm{b}^* ; y^{(i)} | \bxi) = ~& \DKL(\mathbb{P}_{\bm{b}^*, y^{(i)} | \bxi} \|\mathbb{P}_{\bm{b}^*}\mathbb{P}_{y^{(i)} | \bxi})
                \nonumber \\
        		= ~& \sum_{\bm{b}^* \in \{\bm{b}^{(j)}\}_{j\in[M]}} \int \mathrm{p}(\bm{b}^*, y^{(i)}| \bxi)\log\frac{ \mathrm{p}(\bm{b}^*, y^{(i)}| \bxi)}{ \mathrm{p}(\bm{b}^*) \mathrm{p}(y^{(i)}| \bxi)}d \bm{o}
                \nonumber \\
        		= ~& \sum_{\bm{b}^* \in \{\bm{b}^{(j)}\}_{j\in[M]}}  \mathrm{p}(\bm{b}^*) \int  \mathrm{p}( y^{(i)}|\bm{b}^*, \bxi)\log\frac{ \mathrm{p}(y^{(i)}|\bm{b}^*, \bxi)}{  \mathrm{p}(y^{(i)}| \bxi)}d \bm{o} 
                \nonumber \\
        		= ~& \frac{1}{M} \sum_{\bm{b}^* \in \{\bm{b}^{(j)}\}_{j\in[M]}}\DKL(\mathbb{P}_{y^{(i)}|\bm{b}^*, \bxi} \| \mathbb{P}_{y^{(i)}|\bxi}) \nonumber \\
        		\leq ~&  \frac{1}{M^2} \sum_{\bm{b}^*, \bm{b}' \in \{\bm{b}^{(j)}\}_{j\in[M]}}\DKL(\mathbb{P}_{y^{(i)}|\bm{b}^*, \bxi}\|  \mathbb{P}_{y^{(i)}|\bm{b}', \bxi}) 
                \nonumber \\
                = ~& \frac{1}{M^2} \sum_{\bm{b}^*, \bm{b}' \in \{\bm{b}^{(j)}\}_{j\in[M]}} \frac{\left(f_{\bm{b}^*}(\bxi)-f_{\bm{b}'}(\bxi)\right)^2}{2\sigma^2}
                \label{append:eq:mi_decompose_2}
        	\end{align}
        where $\DKL$ refer to the KL-divergence between probability distributions, the third equality exploits the definition of conditional probability distribution $\mathrm{p}(\bm{b}^*, y^{(i)}| \bxi)=\mathrm{p}(\bm{b}^*)\cdot \mathrm{p}(y^{(i)}|\bm{b}^*, \bxi)$, the fourth equality follows the fact that $\bm{b}^*$ is sampled according to the uniform distribution $\mathrm{p}(\bm{b}^*)=1/M$, and the first inequality uses the convexity property of the $-\log$ function in KL divergence, the last inequality employs the assumption that the statistical estimation $y^{(i)}=\sum_{k=1}^m\bm{o}_k(\bxi)/m$ approximately follows the normal distribution $\mathbb{N}(f_{\bm{b}^*}(\bx),\sigma^2)$ with the variance $\sigma^2$ assumed to be the same for varied $\bm{b}^*$. 

        In the following, we separately derive the upper bound of the average discrepancy between $f_{\bm{b}^*}(\bx)$ and $f_{\bm{b}'}(\bx)$ and the upper bound of $\sigma^2$. In particular, employing the explicit form of $f_{\bm{b}}(\bx)=2^{-\lfloor \mathcal{M}_{\lambda} \rfloor}\sum_{\bomega\in \{-1,1\}^d} \bm{b}_{\bomega} \Phi_{\bomega}(\bx)$ defined in Eq.~\eqref{append:eq:bit_target_function} and the analysis in Eq.~\eqref{append:eq:distance_d_H} yields
        \begin{equation}\label{append:eq:mi_decompose_mean}
            \mathbb{E}_{\bxi\sim[-\pi,\pi]^d} \left(f_{\bm{b}^*}(\bxi)-f_{\bm{b}'}(\bxi) \right)^2 = \frac{4\nu_f}{2^{\lfloor \mathcal{M}_{\lambda} \rfloor}}  \mathrm{d}_{\mathrm{H}}(\bm{b}^*,\bm{b}^{\prime}) =\frac{4\nu_f}{2^{\lfloor \mathcal{M}_{\lambda} \rfloor}}  \sum_{\bomega \in \{1,-1\}^{\lfloor \mathcal{M}_{\lambda} \rfloor}}  \mathbbm{1}\{\bm{b}_{\bomega}^{(i)} \ne \bm{b}_{\bomega}^{(j)}\} \le 4\nu_f,
        \end{equation}
        where the inequality follows that there are $2^{\lfloor \mathcal{M}_{\lambda} \rfloor}$ terms corresponding to $\bomega\in\{-1,1\}^{\lfloor \mathcal{M}_{\lambda} \rfloor}$ in the summation.

        As the estimation $y^{(i)}=\sum_{k=1}^m\bm{o}_k(\bxi)/m$ is the statistical average of $m$ measurement outcomes $\{\bm{o}_k(\bxi)\}_{k=1}^m$ with each $\frac{\bm{o}_k(\bxi)+1}{2}\in \{0,1\}$ follows the Bernoulli distribution $(\frac{1+f_{\bm{b}^*}(\bxi)}{2},\frac{1-f_{\bm{b}^*}(\bxi)}{2})$, we have
        \begin{equation}\label{append:eq:mi_decompose_var}
            \sigma^2 = \mathbb{E}_{\bxi \sim [-\pi,\pi]^d} \frac{(1+f_{\bm{b}}(\bxi)) \cdot (1-f_{\bm{b}}(\bxi))}{m}= \mathbb{E}_{\bxi \sim [-\pi,\pi]^d}\frac{1-f_{\bm{b}}(\bxi))^2}{m} = \frac{1-2^{-2\lfloor \mathcal{M}_{\lambda} \rfloor}}{m}\ge \frac{2}{3m},
        \end{equation}
        where the third equality follows the definition of $f_{\bm{b}}(\bx)$ in Eq.~\eqref{append:eq:bit_target_function} and the orthogonality of the trigonometric basis $\Phi_{\bomega}(\bx)$, namely $\mathbb{E}_{\bx\sim[-\pi,\pi]^d} \Phi_{\bomega}(\bx)\Phi_{\bomega}(\bx)=\delta_{\bomega\bomega'}\cdot 2^{-\lfloor \mathcal{M}_{\lambda} \rfloor}$ for $\bomega\in\{-1,1\}^{\lfloor \mathcal{M}_{\lambda} \rfloor}$.
        
        In conjunction with Eq.~(\ref{append:eq:mi_decompose_1}), Eq.~(\ref{append:eq:mi_decompose_2}), Eq.~(\ref{append:eq:mi_decompose_mean}), and Eq.~(\ref{append:eq:mi_decompose_var}), we have
        \begin{equation}\label{append:eq:mi_bound_1}
        	I(\bm{b}^*;\bar{\bm{b}})\leq  3nm\nu_f.
        \end{equation}
        This completes the proof.
    \end{proof}

\section{Quantum subroutine for important feature identification}
\label{append:sec:direct_mode_sampler}

In this section, we present the details about the quantum subroutine introduced in the main text for identifying important trigonometric features in the construction of the \DSE-guided classical surrogate, together with its complexity analysis.
Recall that the target function in $\mathcal{F}_{\mathcal{U}}$ for all $\mathcal{U}\subset \mathsf{Arc}(N,d)$ admits the trigonometric expansion
\begin{equation}
    f(\bx)
    =
    \Tr\!\left(\rho_0U(\bx)^\dagger O U(\bx)\right)
    =
    \sum_{\bomega\in\Lambda}
    \Tr(\rho_0Q_{\bomega})\Phi_{\bomega}(\bx),
    \qquad
    \Lambda=\{0,\pm1\}^d.
    \label{append:eq:direct_sampler_target_expansion}
\end{equation}
Constructing the \DSE-guided surrogate requires selecting an informative subset
$\Lambda_{\mathsf q}\subset\Lambda$.  Since $|\Omega|=3^d$, enumerating all modes is computationally prohibitive when $d$ is large.  A truncation rule based only on $\|\bomega\|_0$, as proposed in prior studies \cite{du2025efficient,liao2025demonstration}, also ignores the circuit- and observable-dependent coefficient $\Tr(\rho_0Q_{\bomega})$. We therefore seek a quantum subroutine whose measurement outcomes follow
\begin{equation}
    \mathrm p(\bomega)
    =
    \frac{
        2^{-\|\bomega\|_0}\Tr(\rho_0Q_{\bomega})^2
    }{\nu_f},
    \qquad
    \nu_f
    :=\mathbb E_{\bx\sim\Unif[-\pi,\pi]^d}f(\bx)^2=
    \sum_{\bomega\in\Lambda}
    2^{-\|\bomega\|_0}\Tr(\rho_0Q_{\bomega})^2.
    \label{append:eq:direct_sampler_target_distribution}
\end{equation}
Sampling from $\mathrm{p}(\bomega)$ favors precisely the modes $\bomega \in \{0,\pm 1\}$ that make large contributions to the target function $f(\bx)$. We achieve this aim by preparing a quantum state that encodes $\sqrt{\mathrm{p}(\bomega)}$ as the state amplitude, namely
\begin{equation}\label{append:eq:amplitude_encode_state}
    \ket{\psi_f}
    =
    \frac{1}{\sqrt{\nu_f}}
    \sum_{\bomega\in\Lambda}
    2^{-\|\bomega\|_0/2}
    \Tr(\rho_0Q_{\bomega})\ket{\bomega}
    =
    \sum_{\bomega\in\Lambda}
    s_{\bomega}\sqrt{\mathrm p(\bomega)}\ket{\bomega},
    \qquad
    s_{\bomega}:=\operatorname{sgn}\!\left[\Tr(\rho_0Q_{\bomega})\right].
\end{equation}
Here each entry $\bomega_i\in\{0,1,-1\}$ is encoded by a two-qubit computational-basis state.  The signs $s_{\bomega}$ do not affect the measurement probabilities, and Eq.~\eqref{append:eq:direct_sampler_target_distribution} ensures $\langle\psi_f|\psi_f\rangle=1$ whenever $\nu_f>0$. In this regard, the learner can conduct computational basis measurements on the prepared state $\ket{\psi_f}$ to collect $\bomega$ and construct the mode subset $\Lambda_{\mathsf{q}}$ for feature identification.

We now specify the access model required for the quantum subroutine.  The observable is nonzero and has a known $q$-sparse Pauli representation
$O=\sum_{l=1}^{q}O_l=\sum_{l=1}^{q}\bm{a}_l P_l$, where $O_l:=\bm{a}_l P_l$, $P_l\in\mathcal{P}_N$, and
\begin{equation}
    0<\|\bm{a}\|_1
    :=\sum_{l=1}^{q}|\bm{a}_l|
    =\sum_{l=1}^{q}\|O_l\|_\infty
    \le1,
    \label{append:eq:lcu_normalization}
\end{equation}
with $\bm{a}=(\bm{a}_1,\ldots,\bm{a}_q)$.  Since $\mathcal{P}_N=\{\mathbb I,X,Y,Z\}^{\otimes N}$ is the unnormalized Pauli basis, $\|P_l\|_\infty=1$ and Eq.~\eqref{append:eq:lcu_normalization} is identical to the observable normalization in Eq.~\eqref{eq:target_function_set}.  
The known sparse decomposition is an additional implementation assumption used only to block encode $O$.  Moreover, the quantum compiler is given the locations and target qubits of the $d$ tunable $\RZ$ gates and can execute the compiled circuits $U(\bx)$ and $U(\bx)^\dagger$.  This information is used only by the feature-identification subroutine during training, while the resulting surrogate remains a purely classical predictor at inference time.

The following lemma characterizes the time complexity of this quantum subroutine for preparing the quantum state $\ket{\psi_f}$, as well as the number of measurements required to construct the mode subset $\Lambda_{\mathsf{q}}\subset \{0,\pm 1\}^d$.

\begin{lemma}[Quantum subroutine for feature identification]
\label{append:lem:direct_mode_sampler}
Consider an arbitrary $N$-qubit circuit $U(\bx)\subset \mathsf{Arc}(N,d)$ containing $G$ gates in total, including $d$ tunable $\RZ$ gates, the input state $\rho_0=\ket{0}\bra{0}^{\otimes N}$, and the bounded observable $O=\sum_{l=1}^{q}O_l$ with $\sum_{l=1}^{q}\|O_l\|_{\infty}\le 1$.  Suppose $\nu_f>0$. Then there exists a quantum subroutine, with access to $U(\bx)$, $U(\bx)^{\dagger}$, the locations and target qubits of the tunable gates $\RZ$, that prepares
\begin{equation}
    \ket{\psi_f}
    =
    \frac{1}{\sqrt{\nu_f}}
    \sum_{\bomega\in\Omega}
    2^{-\|\bomega\|_0/2}
    \Tr(\rho_0Q_{\bomega})
    \ket{\bomega}.
    \label{append:eq:direct_sampler_main_state}
\end{equation}
A computational-basis measurement of its $2d$-qubit frequency register returns one exact sample from $\mathrm p(\bomega)$ in Eq.~\eqref{append:eq:direct_sampler_target_distribution}.

Preparing
After one-time classical preprocessing of complexity $\mathcal O(qN)$, one invocation of the expectation-state preparation unitary or its inverse uses
\begin{equation}
    C_{\rm eval}=\mathcal O(qN+G+d)
    \label{append:eq:direct_eval_cost}
\end{equation}
gates and $N+2d+\lceil\log_2q\rceil$ qubits.  Repeat-until-success postselection produces one sample with expected gate complexity
\begin{equation}
    \mathcal O\left(
        C_{\rm eval}\frac{\|\bm{a}\|_1^2}{\nu_f}
    \right).
    \label{append:eq:direct_repeat_cost}
\end{equation}
Amplitude amplification improves this to
\begin{equation}
    \mathcal O\left(
        C_{\rm eval}\frac{\|\bm{a}\|_1}{\sqrt{\nu_f}}
    \right)
    \label{append:eq:direct_amplified_cost}
\end{equation}
per sample.  Consequently, obtaining $m_{\mathsf f}$ independent frequency draws and retaining their distinct outcomes in $\Lambda_{\mathsf q}$ has expected gate complexity
\begin{equation}
    \mathcal O\left(
        (qN+G+d)m_{\mathsf f}
        \frac{\|\bm{a}\|_1}{\sqrt{\nu_f}}
    \right),
    \qquad
    |\Lambda_{\mathsf q}|\le m_{\mathsf f}.
    \label{append:eq:direct_total_cost}
\end{equation}
\end{lemma}

For clarity, we first outline the implementation of the quantum subroutine in Lemma~\ref{append:lem:direct_mode_sampler} and the circuit constructions and their detailed derivations are deferred to the subsequent subsections. As shown in Fig.~\ref{fig:quantum_sampler}, the subroutine uses three registers, denoted by $\mathsf{Q}, \mathsf{F}, \mathsf{A}$. The $N$-qubit system register $\mathsf{Q}$ is initialized in $\rho_0=\bra{0}\ket{0^N}$ and is used to execute the circuit $U(\bx)$. The $2d$-qubit register $\mathsf{F}$ serves as the frequency-output register that extracts the information of $\sqrt{\mathrm{p}(\bomega)}$ as the amplitudes of the computational basis state $\ket{\bomega}$, where the frequency labels in each coordinate $\bomega_i= 0,1,-1$ are encoded by the two-qubit states $\ket{00},\ket{01},\ket{10}$, respectively. The $a=\lceil \log q \rceil$-qubit ancillary register $\mathsf{A}$ is used to embed the observable $O$ into the quantum circuit through LCU block encoding.

\begin{figure*}[t]
 \centering
 \includegraphics[width=0.8\linewidth]{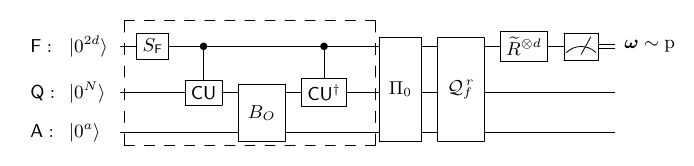}
 \caption{\textbf{Implementation of the quantum subroutine that samples important frequency modes.} The preparation $S_{\mathsf{F}}$ creates the uniform superposition over the three-angel grid. The dashed block coherently applies $U(\bx_{\bomega})$, the LCU block encoding $B_O$ of the observable, and $U(\bx_{\bomega})^{\dagger}$. The block $\Pi_0$ and $\mathcal{Q}_f^{r}$ denote the postselection of $\mathsf{QA}$ onto $\ket{0^N0^a}$ and its coherent amplitude-amplification implementation, respectively, producing $\ket{\chi_f}$ on $\mathsf{F}$. Finally, the local transformation $\widetilde{R}^{\otimes d}$ prepares $\ket{\psi_f}$, whose computational-basis measurement returns $\bomega\sim \mathrm{p}$. Bundled wires represent the registers rather than individual qubits.
 }
 \label{fig:quantum_sampler}
\end{figure*}

Starting from the initial state $\ket{0^{2d}}_{\mathsf{Q}}\ket{0^N}_{\mathsf{F}}\ket{0^{a}}_{\mathsf{A}}$, the quantum subroutine proceeds in three stages. First, the register $\mathsf{F}$ is prepared in the uniform superposition $\ket{\mathrm{unif}}_{\mathsf{F}}= \frac{1}{\sqrt{3^d}}\sum_{\bomega\in \{0,\pm 1\}^d}\ket{\bomega}$. Second, we apply a quantum circuit that transforms the state $\ket{\mathrm{unif}}{\mathsf{F}}\ket{0^N}_{\mathsf{Q}}\ket{0^{a}}_{\mathsf{A}}$ into
\begin{equation}\label{append:eq:uou_state}
    \ket{\Psi}=\frac{1}{\|\bm{a}\|_1\sqrt{3^d}}\sum_{\bomega}\ket{\bomega}_{\mathsf{F}}U(\alpha_{\bomega})^{\dagger}OU(\alpha_{\bomega})\ket{0^N}_{\mathsf{Q}}\ket{0^{a}}_{\mathsf{A}}.
\end{equation}
This state coherently encodes the evolved observables $U(\bm{\alpha}_{\bomega})^{\dagger}OU(\bm{\alpha}_{\bomega})$ for all $\bomega\in{0,\pm1}^d$, where $B_O$ is the LCU block encoding circuit of the observable $O$ (as detailed in SI.~\ref{append:subsec:direct_sampler_access}) and $\mathsf{CU}(\bx)=\sum_{\bomega} \ket{\bomega}\bra{\bomega}_{\mathsf{F}}\otimes U(\bm{\alpha}_{\bomega})_{\mathsf{Q}}$ is a controlled circuit defined over three-parameter grid points $\bm{\alpha}_{\bomega}\in \mathbb{R}^d$. Here, $\bm{\alpha}_{\bomega}$ are chosen such that $\sum_{\bomega} f(\bm{\alpha}_{\bomega})^2=\nu_f$ (as detailed in SI.~\ref{append:subsec:direct_sampler_access}). Consequently, postselecting the registers $\mathsf{Q}$ and $\mathsf{A}$ onto $\ket{0^N}_{\mathsf{Q}}\ket{0^{a}}_{\mathsf{A}}$ yields the state 
\begin{equation}\label{append:eq:f_amplitude_state}
    \ket{\chi_f}_{\mathsf{F}}=\frac{1}{\sqrt{3^d\nu_f}}\sum_{\bomega} f(\bm{\alpha}_{\bomega}) \ket{\bomega}_{\mathsf{F}}, 
\end{equation}
with success probability $\nu_f/\|\bm{a}\|_1^2$. This success probability could be improved to $\sqrt{\nu_f}/\|\bm{a}\|_1$ through amplitude amplification.
Finally, the state $\ket{\chi_f}_{\mathsf{F}}$ is transformed into the desired probability-amplitude-encoded state $\ket{\psi_f}_{\mathsf{F}}$ by applying a tensor product of $d$ two-qubit gates.

The subsequent sections are organized as follows.  First, SI.~\ref{append:subsec:direct_sampler_access} constructs the LCU block encoding circuit $B_O$ of the observable $O$ and elucidates how to set the grid rotation angles $\bm{\alpha}_{\bomega}$ in Eq.~\eqref{append:eq:uou_state}.  Second, SI.~\ref{append:subsec:coherent_expectation_encoding} implements the coherently controlled grid circuit $\mathsf{CU}$ and prepares the state $\ket{\chi_f}_{\mathsf{F}}$ whose amplitudes are the sampled expectation values $f(\bm{\alpha}_{\bomega})$ in Eq.~\eqref{append:eq:f_amplitude_state}. Third, SI.~\ref{append:subsec:grid_to_frequency} provide the tensor-provide transformation that evolves $\ket{\chi_f}_{\mathsf{F}}$ to the desired state $\ket{\psi_f}_{\mathsf{F}}$ in Eq.~\eqref{append:eq:direct_sampler_main_state}.  Finally, SI.~\ref{append:subsec:direct_sampler_complexity} provides the proof of Lemma~\ref{append:lem:direct_mode_sampler} by combining these components to derive the total gate and measurement complexities.

\subsection{Access primitives and the three-angle parameter grid}
\label{append:subsec:direct_sampler_access}

\subsubsection{Sparse-Pauli block encoding of the observable}

We first construct the observable block encoding required for coherent expectation-value evaluation.  The following lemma records both the construction and its gate complexity.

\begin{lemma}[LCU coefficient-state preparation and sparse-Pauli block encoding]
\label{append:lem:sparse_Pauli_LCU}
Let $O=\sum_{l=1}^{q}\bm{a}_{l} P_l$ be the nonzero Hermitian observable in Eq.~\eqref{append:eq:lcu_normalization}, specified by a classical list of $q$ distinct unnormalized $N$-qubit Pauli strings and their coefficients.  With
$a=\lceil\log_2q\rceil$ address qubits (and $a=0$ for $q=1$), there exists a coefficient-state preparation unitary satisfying
\begin{equation}
    \operatorname{PREP}_O\ket{0^a}
    =
    \ket{G_O}_{\mathsf A}
    :=
    \sum_{l=1}^{q}
    \sqrt{\frac{|\bm{a}_{l}|}{\|\bm{a}\|_1}}\ket{l}_{\mathsf A}.
    \label{append:eq:LCU_coefficient_state}
\end{equation}
After one-time classical preprocessing of complexity $\mathcal O(qN)$, $\operatorname{PREP}_O$ and $\operatorname{PREP}_O^\dagger$ can each be implemented using $\mathcal O(q)$ one- and two-qubit gates, while the corresponding Pauli-selection unitary can be implemented using $\mathcal O(qN)$ gates.  These operations construct a unitary $B_O$ satisfying
\begin{equation}
    (\mathbb I_{\mathsf Q}\otimes\bra{0^{a}}_{\mathsf A})
    B_O
    (\mathbb I_{\mathsf Q}\otimes\ket{0^{a}}_{\mathsf A})
    =\frac{O}{\|\bm{a}\|_1},
    \label{append:eq:observable_block_encoding}
\end{equation}
and both $B_O$ and $B_O^\dagger$ have gate complexity $\mathcal O(qN)$.  For $q=1$, no address qubit is required, and one may take $B_O=\operatorname{sgn}(\bm{a}_1)P_1$.  Here $\ket{G_O}$ is the coefficient state on the LCU address register.  The state-preparation and block-encoding constructions follow standard sparse-state-preparation and LCU/block-encoding techniques~\cite{gleinig2021efficient,gilyen2019quantum}.
\end{lemma}

\begin{proof}
Since $O$ is Hermitian, its Pauli coefficients are real.  Write $\bm{s}_l:=\operatorname{sgn}(\bm{a}_{l})$ and define
\begin{align}
    \operatorname{SELECT}_O
    &:=
    \sum_{l=1}^{q}
    \bm{s}_l P_l
    \otimes\ket{l}\bra{l}_{\mathsf A},
    \label{append:eq:LCU_SELECT}
\end{align}
with unused address states completed by any unitary action.  Define
\begin{equation}
    B_O
    :=
    (\mathbb I_{\mathsf Q}\otimes\operatorname{PREP}_O^\dagger)
    \operatorname{SELECT}_O
    (\mathbb I_{\mathsf Q}\otimes\operatorname{PREP}_O).
    \label{append:eq:LCU_block_encoding_circuit}
\end{equation}
Projecting the address register back onto $\ket{0^a}$ gives
\begin{align}
    (\mathbb I_{\mathsf Q}\otimes\bra{0^a})
    B_O
    (\mathbb I_{\mathsf Q}\otimes\ket{0^a})
    &=
    \sum_{l=1}^{q}
    \frac{|\bm{a}_{l}|}{\|\bm{a}\|_1}\bm{s}_l P_l
    =
    \frac{O}{\|\bm{a}\|_1},
    \label{append:eq:LCU_block_derivation}
\end{align}
which verifies Eq.~\eqref{append:eq:observable_block_encoding}.

It remains to count the implementation cost.  The $q$ amplitudes in Eq.~\eqref{append:eq:LCU_coefficient_state} can be compiled into a sparse state-preparation circuit using $\mathcal O(q)$ elementary gates~\cite{gleinig2021efficient}.  For each of the $N$ system qubits, $\operatorname{SELECT}_O$ is a uniformly controlled choice among $\{\mathbb I,X,Y,Z\}$ specified by the $q$ Pauli words; compiling these choices costs $\mathcal O(q)$ gates per system qubit and hence $\mathcal O(qN)$ gates in total.  Equation~\eqref{append:eq:LCU_block_encoding_circuit} calls $\operatorname{PREP}_O$, $\operatorname{SELECT}_O$, and $\operatorname{PREP}_O^\dagger$ once each, giving
\begin{equation}
    C_{B_O}
    =
    2C_{\operatorname{PREP}_O}
    +
    C_{\operatorname{SELECT}_O}
    =
    \mathcal O(qN).
    \label{append:eq:LCU_block_encoding_cost}
\end{equation}
The inverse circuit has the same asymptotic cost.  The classical descriptions contain $q$ Pauli words of length $N$, so reading them and compiling the controlled Pauli operations requires $\mathcal O(qN)$ one-time classical preprocessing.  This proves the lemma.
\end{proof}

\subsubsection{Three-angle grid and exact evaluation of the mean-square expectation}

For each coordinate, consider the three grid points
\begin{equation}
    \bm{\alpha}_0=0,\qquad
    \bm{\alpha}_1=\frac{2\pi}{3},\qquad
    \bm{\alpha}_{-1}=-\frac{2\pi}{3}.
    \label{append:eq:three_point_grid}
\end{equation}
Let $\phi_0(\alpha)=1$, $\phi_1(\alpha)=\cos\alpha$, and $\phi_{-1}(\alpha)=\sin\alpha$.  Their values at these three angles obey
\begin{equation}
    \frac{1}{3}\sum_{\bomega\in\{0,1,-1\}}
    \phi_r(\bm{\alpha}_{\bomega})\phi_s(\bm{\alpha}_{\bomega})
    =
    \delta_{r,s}\,2^{-\mathbbm 1\{r\ne0\}},
    \qquad r,s\in\{0,1,-1\}.
    \label{append:eq:three_point_orthogonality}
\end{equation}
Indeed, their sampled values are
\begin{equation}
    \bigl[\phi_r(\bm{\alpha}_{\omega})\bigr]_
    {\substack{r\in\{0,1,-1\}\\\omega\in\{0,1,-1\}}}
    =
    \begin{pmatrix}
        1&1&1\\
        1&-1/2&-1/2\\
        0&\sqrt3/2&-\sqrt3/2
    \end{pmatrix}.
    \label{append:eq:sampled_trig_table}
\end{equation}
The three row vectors in Eq.~\eqref{append:eq:sampled_trig_table} are mutually orthogonal and have squared Euclidean norms $3$, $3/2$, and $3/2$, respectively.  Dividing their inner products by $3$ gives Eq.~\eqref{append:eq:three_point_orthogonality}.
For $\bomega=(\bomega_1,\ldots,\bomega_d)\in\{0,1,-1\}^d$, define
$\bx_{\bomega}=\bm{\alpha}_{\bomega}:=(\bm{\alpha}_{\bomega_1},\cdots,\bm{\alpha}_{\bomega_d})$.  Applying the one-dimensional relation independently to all $d$ coordinates gives
\begin{equation}
    \frac{1}{3^d}\sum_{\bomega}
    \Phi_{\bm r}(\bx_{\bomega})
    \Phi_{\bm s}(\bx_{\bomega})
    =
    \delta_{\bm r,\bm s}2^{-\|\bm r\|_0},
    \qquad
    \bm r,\bm s\in\{0,1,-1\}^d.
    \label{append:eq:grid_tensor_orthogonality}
\end{equation}
To verify that this discrete average equals the corresponding continuous mean-square, write $\bm{c}_{\bm r}:=\Tr(\rho_0Q_{\bm r})$ and substitute
$f(\bx_{\bomega})=\sum_{\bm r}\bm{c}_{\bm r}\Phi_{\bm r}(\bx_{\bomega})$.  Then
\begin{align}
    \frac{1}{3^d}\sum_{\bomega}f(\bx_{\bomega})^2
    &=
    \frac{1}{3^d}
    \sum_{\bomega}
    \sum_{\bm r,\bm s}
    \bm{c}_{\bm r}\bm{c}_{\bm s}
    \Phi_{\bm r}(\bx_{\bomega})
    \Phi_{\bm s}(\bx_{\bomega})
    \nonumber\\
    &=
    \sum_{\bm r,\bm s}
    \bm{c}_{\bm r}\bm{c}_{\bm s}
    \delta_{\bm r,\bm s}
    2^{-\|\bm r\|_0}
    \nonumber\\
    &=
    \sum_{\bm r\in\Lambda}
    2^{-\|\bm r\|_0}\bm{c}_{\bm r}^2
    \nonumber\\
    &=
    \mathbb E_{\bx\sim\Unif[-\pi,\pi]^d}f(\bx)^2
    =\nu_f.
    \label{append:eq:grid_exact_nuf}
\end{align}
Thus, averaging the squared function values over these $3^d$ grid points gives exactly the same result as averaging $f(\bx)^2$ over the continuous uniform distribution on $[-\pi,\pi]^d$.

\subsection{Coherent grid evaluation and sampled expectation-value encoding}
\label{append:subsec:coherent_expectation_encoding}

The purpose of this subsection is to coherently encode the grid values
$\{f(\bx_{\bomega})\}_{\bomega\in\{0,1,-1\}^d}$ into a quantum state.  We assume compiled coherent access to the parameterized circuit: the locations and target qubits of its $d$ tunable rotations are known, so each parameter gate can be controlled by its associated two-qubit grid block, while every fixed gate is applied without this additional control.  Thus one invocation of the expectation-encoding circuit calls one coherently controlled $U(\bx_{\bomega})$, one coherently controlled $U(\bx_{\bomega})^\dagger$, and one observable block encoding $B_O$.  We first state the resulting preparation guarantee and then derive the circuit implementation.

\begin{lemma}[Coherent expectation-value encoding]
\label{append:lem:coherent_expectation_encoding}
Under the assumptions of Lemma~\ref{append:lem:direct_mode_sampler}, there exists a unitary $\mathcal W_f$ acting on registers $\mathsf F\mathsf Q\mathsf A$ such that projecting $\mathsf Q\mathsf A$ onto $\ket{0^N0^a}$ succeeds with probability
\begin{equation}
    p_{\rm succ}=\frac{\nu_f}{\|\bm{a}\|_1^2}
\end{equation}
and, conditioned on success, prepares
\begin{equation}
    \ket{\chi_f}_{\mathsf F}
    =
    \frac{1}{\sqrt{3^d\nu_f}}
    \sum_{\bomega\in\{0,1,-1\}^d}
    f(\bx_{\bomega})\ket{\bomega}_{\mathsf F}.
    \label{append:eq:coherent_expectation_lemma_state}
\end{equation}
Amplitude amplification raises the success probability to a constant using
$\mathcal O(\|\bm{a}\|_1/\sqrt{\nu_f})$ applications of $\mathcal W_f$ and
$\mathcal W_f^\dagger$, without changing the normalized conditional state
$\ket{\chi_f}$.
\end{lemma}

\begin{proof}[Proof of Lemma~\ref{append:lem:coherent_expectation_encoding}]
We first prepare the uniform grid-label state, then compile the coherently controlled circuit $\mathsf{CU}$.  We next combine this circuit with the observable block encoding to place the sampled expectation values in the amplitudes of $\mathsf{F}$.  Finally, we calculate the postselection probability and show how amplitude amplification reduces the preparation cost.

\medskip
\noindent\underline{Register encoding and uniform grid-state preparation.}
Let $\mathsf F=\mathsf F^{(1)}\cdots\mathsf F^{(d)}$ be a $2d$-qubit register, where each block $\mathsf F^{(j)}$ consists of two qubits and uses the encoding
\begin{equation}
    \ket{0}_{\mathsf F^{(j)}}:=\ket{00},\qquad
    \ket{1}_{\mathsf F^{(j)}}:=\ket{01},\qquad
    \ket{-1}_{\mathsf F^{(j)}}:=\ket{10}.
    \label{append:eq:two_qubit_grid_encoding}
\end{equation}
The fourth basis state $\ket{11}$ is unused, so the entire grid is stored in exactly $2d$ qubits.  Let $\mathsf Q$ and $\mathsf A$ denote the $N$-qubit system register and the block-encoding ancilla, respectively.  Prepare
\begin{equation}
    \ket{\mathrm{unif}}_{\mathsf F}
    =
    \frac{1}{\sqrt{3^d}}
    \sum_{\bomega\in\{0,1,-1\}^d}
    \ket{\bomega_1}\cdots\ket{\bomega_d}_{\mathsf F}
    =
    \frac{1}{\sqrt{3^d}}
    \sum_{\bomega\in\Lambda}\ket{\bomega}_{\mathsf F},
    \qquad
    \ket{0^N}_{\mathsf Q}\ket{0^a}_{\mathsf A}.
\end{equation}
Equivalently, let $S_{\mathsf F}=\bigotimes_{j=1}^d S_{\mathsf F^{(j)}}$, where one explicit choice of the local two-qubit unitary is
\begin{equation}
    S_{\mathsf F^{(j)}}
    =
    \begin{pmatrix}
        1/\sqrt3&\sqrt{2/3}&0&0\\
        1/\sqrt3&-1/\sqrt6&1/\sqrt2&0\\
        1/\sqrt3&-1/\sqrt6&-1/\sqrt2&0\\
        0&0&0&1
    \end{pmatrix}.
    \label{append:eq:uniform_grid_preparation_gate}
\end{equation}
Its columns are orthonormal, and its first column gives
$S_{\mathsf F^{(j)}}\ket{00}=(\ket{00}+\ket{01}+\ket{10})/\sqrt3$.
Then $\ket{\mathrm{unif}}_{\mathsf F}=S_{\mathsf F}\ket{0^{2d}}$, and $S_{\mathsf F}$ uses $\mathcal O(d)$ gates.

\medskip
\noindent\underline{Coherently controlled grid circuit.}
Conditioned locally on the two-qubit codeword $\ket{\bomega_j}_{\mathsf F^{(j)}}$, replace the $j$-th tunable rotation by its value at $\bm{\alpha}_{\bomega_j}$.  On the valid code space this implements
\begin{equation}
    \mathsf{CU}
    =
    \sum_{\bomega}
    \ket{\bomega}\bra{\bomega}_{\mathsf F}
    \otimes U(\bx_{\bomega}),
    \qquad
    \ket{\bomega}
    :=
    \ket{\bomega_1}\cdots\ket{\bomega_d},
    \label{append:eq:controlled_grid_circuit}
\end{equation}
with an arbitrary unitary extension on codewords containing $\ket{11}$.  More concretely, for an $\RZ$ parameter gate, the corresponding controlled grid gate is
\begin{align}
    \mathcal R_j
    =&\;
    \ket{00}\bra{00}_{\mathsf F^{(j)}}\otimes\mathbb I
    +
    \ket{01}\bra{01}_{\mathsf F^{(j)}}\otimes
    \RZ(2\pi/3)\nonumber\\
    &+
    \ket{10}\bra{10}_{\mathsf F^{(j)}}\otimes
    \RZ(-2\pi/3)
    +
    \ket{11}\bra{11}_{\mathsf F^{(j)}}\otimes\mathbb I.
    \label{append:eq:controlled_three_angle_rotation}
\end{align}
The four orthogonal control projectors in Eq.~\eqref{append:eq:controlled_three_angle_rotation} sum to the identity, and each target block is unitary; hence $\mathcal R_j$ is unitary.  We next give an explicit elementary-gate implementation.  Denote the two qubits of $\mathsf F^{(j)}$ by $(a_j,b_j)$, with $a_j$ the first bit in Eq.~\eqref{append:eq:two_qubit_grid_encoding}, and let $t_j$ be the system qubit on which the original $\RZ(\bx_j)$ acts.  Setting $\varphi:=2\pi/3$, the rotation angle selected by a computational-basis codeword is
\begin{equation}
    \theta(a_j,b_j)=\varphi(b_j-a_j).
    \label{append:eq:grid_angle_boolean}
\end{equation}
Indeed, the four bit strings $(a_j,b_j)=(0,0),(0,1),(1,0),(1,1)$ give the angles $0,\varphi,-\varphi,0$, respectively.  Consequently, Eq.~\eqref{append:eq:controlled_three_angle_rotation} is implemented exactly by two singly controlled rotations,
\begin{equation}
    \mathcal R_j
    =
    \operatorname{CRZ}_{b_j\rightarrow t_j}(\varphi)\,
    \operatorname{CRZ}_{a_j\rightarrow t_j}(-\varphi).
    \label{append:eq:controlled_grid_two_CRZ}
\end{equation}
The two gates commute because they are rotations about the same target axis.  In particular, on the unused codeword $\ket{11}$ their angles cancel, so the chosen unitary extension acts as the identity.

\medskip
\noindent\underline{Elementary-gate realization of the controlled rotations.}
For completeness, under the convention $\RZ(\theta)=e^{-i\theta Z/2}$, each controlled rotation in Eq.~\eqref{append:eq:controlled_grid_two_CRZ} has the exact decomposition
\begin{align}
    \operatorname{CRZ}_{c\rightarrow t}(\theta)
    =&\;
    \CNOT_{c,t}
    \bigl(\mathbb I_c\otimes\RZ_t(-\theta/2)\bigr)
    \CNOT_{c,t}
    \bigl(\mathbb I_c\otimes\RZ_t(\theta/2)\bigr),
    \label{append:eq:CRZ_elementary_decomposition}
\end{align}
where the rightmost gate acts first.  If the control is $\ket{0}$, the two target rotations cancel.  If the control is $\ket{1}$, using
$X\RZ(\gamma)X=\RZ(-\gamma)$ gives
\begin{equation}
    X\RZ(-\theta/2)X\RZ(\theta/2)
    =
    \RZ(\theta),
    \label{append:eq:CRZ_decomposition_check}
\end{equation}
which verifies Eq.~\eqref{append:eq:CRZ_elementary_decomposition}.  Hence each $\mathcal R_j$ uses four CNOT gates and four single-qubit $\RZ$ gates.

\medskip
\noindent\underline{Assembly of the full controlled circuit.}
It remains to verify that these local replacements implement the full controlled circuit in Eq.~\eqref{append:eq:controlled_grid_circuit}.  Write the original parameterized circuit as
\begin{equation}
    U(\bx)
    =
    V_d\RZ_{t_d}(\bx_d)V_{d-1}\cdots
    V_1\RZ_{t_1}(\bx_1)V_0,
    \label{append:eq:grid_original_gate_sequence}
\end{equation}
where all $V_j$ are independent of $\bx$.  Replace every $\RZ_{t_j}(\bx_j)$ by $\mathcal R_j$ and apply every $V_j$ without control.  The resulting circuit is
\begin{equation}
    \widehat{\mathsf CU}
    =
    (\mathbb I_{\mathsf F}\otimes V_d)\mathcal R_d
    (\mathbb I_{\mathsf F}\otimes V_{d-1})\cdots
    (\mathbb I_{\mathsf F}\otimes V_1)\mathcal R_1
    (\mathbb I_{\mathsf F}\otimes V_0).
    \label{append:eq:grid_compiled_sequence}
\end{equation}
Expanding each $\mathcal R_j$ in its orthogonal control projectors and restricting to the valid code space gives
\begin{align}
    \widehat{\mathsf CU}
    &=
    \sum_{\bomega_1,\ldots,\bomega_d\in\{0,1,-1\}}
    \left(
        \bigotimes_{j=1}^{d}
        \ket{\bomega_j}\bra{\bomega_j}_{\mathsf F^{(j)}}
    \right)
    \otimes
    \left[
        V_d\RZ_{t_d}(\alpha_{\bomega_d})V_{d-1}\cdots
        V_1\RZ_{t_1}(\alpha_{\bomega_1})V_0
    \right]\nonumber\\
    &=
    \sum_{\bomega}
    \ket{\bomega}\bra{\bomega}_{\mathsf F}
    \otimes U(\bx_{\bomega})
    =
    \mathsf{CU}.
    \label{append:eq:grid_compiled_expansion}
\end{align}
Thus the circuit processes all $3^d$ grid branches coherently without enumerating them.  Each parameter gate incurs only the constant overhead derived above, so all $d$ controlled grid rotations contribute $\mathcal O(d)$ elementary gates.

\medskip
\noindent\underline{Expectation-value amplitude encoding.}
The expectation-encoding circuit is explicitly
\begin{equation}
    \mathcal A_f
    :=
    \mathsf{CU}^{\dagger}
    (\mathbb I_{\mathsf F}\otimes B_O)
    \mathsf{CU}.
    \label{append:eq:expectation_encoding_circuit}
\end{equation}
Including the preparation of the uniform grid state, define the full state-preparation unitary
\begin{equation}
    \mathcal W_f
    :=
    \mathcal A_f
    (S_{\mathsf F}\otimes\mathbb I_{\mathsf Q\mathsf A}).
    \label{append:eq:full_expectation_state_preparation}
\end{equation}
It therefore uses the quantum circuits $U(\bx_{\bomega})$, $B_O$, and $U(\bx_{\bomega})^\dagger$ in this order.  Eqs.~\eqref{append:eq:observable_block_encoding} and the definition
$f(\bx)=\bra{0^N}U(\bx)^\dagger O U(\bx)\ket{0^N}$ imply the amplitude-encoding identity below.  We give its derivation explicitly.  The state after applying the controlled forward circuit $\mathsf{CU}$ is
\begin{equation}
    \ket{\Psi_1}
    =
    \frac{1}{\sqrt{3^d}}
    \sum_{\bomega}
    \ket{\bomega}_{\mathsf F}
    U(\bx_{\bomega})\ket{0^N}_{\mathsf Q}
    \ket{0^a}_{\mathsf A}.
    \label{append:eq:direct_state_after_forward}
\end{equation}
For every system state $\ket{\varphi}$, the block-encoding relation in Eq.~\eqref{append:eq:observable_block_encoding} is equivalently written as
\begin{equation}
    B_O\ket{\varphi}_{\mathsf Q}\ket{0^a}_{\mathsf A}
    =
    \frac{O}{\|\bm{a}\|_1}\ket{\varphi}_{\mathsf Q}
    \ket{0^a}_{\mathsf A}
    +
    \ket{\Gamma_{\varphi}^{\perp}}_{\mathsf Q\mathsf A},
    \qquad
    (\mathbb I_{\mathsf Q}\otimes\bra{0^a}_{\mathsf A})
    \ket{\Gamma_{\varphi}^{\perp}}=0.
    \label{append:eq:block_encoding_action}
\end{equation}
Applying Eq.~\eqref{append:eq:block_encoding_action} branchwise to $\ket{\Psi_1}$ and then applying $\mathsf{CU}^\dagger$ yields
\begin{align}
    \ket{\Psi_2}
    &=
    \frac{1}{\|\bm{a}\|_1\sqrt{3^d}}
    \sum_{\bomega}
    \ket{\bomega}_{\mathsf F}
    U(\bx_{\bomega})^\dagger
    O
    U(\bx_{\bomega})\ket{0^N}_{\mathsf Q}
    \ket{0^a}_{\mathsf A}
    +
    \ket{\Gamma^\perp},
    \label{append:eq:direct_state_after_inverse}
\end{align}
where $(\mathbb I_{\mathsf F\mathsf Q}\otimes\bra{0^a}_{\mathsf A})\ket{\Gamma^\perp}=0$.  Projecting both $\mathsf Q$ and $\mathsf A$ onto their all-zero states therefore gives
\begin{align}
    &\bigl(\mathbb I_{\mathsf F}\otimes
    \bra{0^N}_{\mathsf Q}\bra{0^a}_{\mathsf A}\bigr)
    \mathcal A_f
    \ket{\mathrm{unif}}_{\mathsf F}
    \ket{0^N}_{\mathsf Q}\ket{0^a}_{\mathsf A}
    \nonumber\\
    &\hspace{20mm}
    =
    \frac{1}{\|\bm{a}\|_1\sqrt{3^d}}
    \sum_{\bomega}f(\bx_{\bomega})
    \ket{\bomega}_{\mathsf F}.
    \label{append:eq:expectation_amplitude_state}
\end{align}

\medskip
\noindent\underline{Postselection probability and normalized sampled-value state.}
Hence, projecting $\mathsf Q\mathsf A$ onto
$\ket{0^N}\ket{0^a}$ succeeds with probability
\begin{equation}
    p_{\rm succ}
    =
    \frac{1}{\|\bm{a}\|_1^2 3^d}
    \sum_{\bomega}f(\bx_{\bomega})^2
    =
    \frac{\nu_f}{\|\bm{a}\|_1^2},
    \label{append:eq:direct_sampler_success_probability}
\end{equation}
where the first equality is the squared norm of the state in Eq.~\eqref{append:eq:expectation_amplitude_state}, and the last equality follows from Eq.~\eqref{append:eq:grid_exact_nuf}.  Moreover, $|f(\bx)|\le\|O\|_\infty\le\|\bm{a}\|_1$ implies $0\le\nu_f/\|\bm{a}\|_1^2\le1$, confirming that Eq.~\eqref{append:eq:direct_sampler_success_probability} is a valid probability.  Dividing the successful branch in Eq.~\eqref{append:eq:expectation_amplitude_state} by $\sqrt{p_{\rm succ}}=\sqrt{\nu_f}/\|\bm{a}\|_1$ gives the normalized $2d$-qubit state
\begin{equation}
    \ket{\chi_f}_{\mathsf F}
    =
    \frac{1}{\sqrt{3^d\nu_f}}
    \sum_{\bomega}f(\bx_{\bomega})
    \ket{\bomega}_{\mathsf F}.
    \label{append:eq:sampled_function_state}
\end{equation}

\medskip
\noindent\underline{Amplitude amplification.}
For later use in amplitude amplification, introduce the successful-subspace projector
\begin{equation}
    \Pi_{\rm good}
    :=
    \mathbb I_{\mathsf F}\otimes
    \ket{0^N0^a}\bra{0^N0^a}_{\mathsf Q\mathsf A}.
    \label{append:eq:direct_good_projector}
\end{equation}
The full preparation in Eq.~\eqref{append:eq:full_expectation_state_preparation} then has the orthogonal decomposition
\begin{equation}
    \mathcal W_f\ket{0^{2d+N+a}}
    =
    \sqrt{\frac{\nu_f}{\|\bm{a}\|_1^2}}\,
    \ket{\chi_f}_{\mathsf F}\ket{0^N0^a}_{\mathsf Q\mathsf A}
    +
    \sqrt{1-\frac{\nu_f}{\|\bm{a}\|_1^2}}\,
    \ket{\Psi_f^\perp},
    \label{append:eq:good_bad_decomposition}
\end{equation}
where $\Pi_{\rm good}\ket{\Psi_f^\perp}=0$.  Amplitude amplification rotates only between the normalized good and bad components in Eq.~\eqref{append:eq:good_bad_decomposition}.  It therefore increases the probability of the good outcome without changing its conditional frequency-register state $\ket{\chi_f}$.
To make this statement explicit, define
\begin{equation}
    \ket{G_f}
    :=
    \ket{\chi_f}_{\mathsf F}\ket{0^N0^a}_{\mathsf Q\mathsf A},
    \qquad
    \ket{B_f}:=\ket{\Psi_f^\perp},
    \qquad
    \sin\vartheta_f:=\frac{\sqrt{\nu_f}}{\|\bm{a}\|_1}.
    \label{append:eq:amplitude_amplification_angle}
\end{equation}
Equation~\eqref{append:eq:good_bad_decomposition} becomes
\begin{equation}
    \mathcal W_f\ket{0^{2d+N+a}}
    =
    \sin\vartheta_f\ket{G_f}
    +
    \cos\vartheta_f\ket{B_f}.
    \label{append:eq:amplitude_amplification_plane}
\end{equation}
Let
\begin{equation}
    S_0:=\mathbb I-2\ket{0^{2d+N+a}}\bra{0^{2d+N+a}},
    \qquad
    S_{\rm good}:=\mathbb I-2\Pi_{\rm good},
\end{equation}
and define the Grover iterate
\begin{equation}
    \mathcal Q_f
    :=
    -\mathcal W_fS_0\mathcal W_f^\dagger S_{\rm good}.
    \label{append:eq:direct_grover_iterate}
\end{equation}
The two reflections in Eq.~\eqref{append:eq:direct_grover_iterate} preserve the plane spanned by $\{\ket{G_f},\ket{B_f}\}$, and their product performs a rotation by $2\vartheta_f$.  Consequently,
\begin{equation}
    \mathcal Q_f^{\,r}\mathcal W_f\ket{0^{2d+N+a}}
    =
    \sin\bigl((2r+1)\vartheta_f\bigr)\ket{G_f}
    +
    \cos\bigl((2r+1)\vartheta_f\bigr)\ket{B_f}.
    \label{append:eq:amplitude_amplification_rotation}
\end{equation}
Thus $r=\mathcal O(\vartheta_f^{-1})=\mathcal O(\|\bm{a}\|_1/\sqrt{\nu_f})$ applications suffice to obtain a constant success probability.  Crucially, the normalized vector inside the good component remains exactly $\ket{G_f}$ for every $r$; amplitude amplification changes only its scalar coefficient.  Hence it does not alter the conditional state $\ket{\chi_f}$ or the final distribution $\mathrm p(\bomega)$.
If $\nu_f$ and hence $\vartheta_f$ are not known beforehand, one need not choose a precisely tuned value of $r$.  The standard randomized amplitude-amplification schedule chooses $r$ uniformly from
$\{0,\ldots,M-1\}$ and increases $M$ geometrically until the good outcome is observed~\cite{brassard2002quantum}.  Once
$M=\Theta(1/\sin\vartheta_f)$, averaging Eq.~\eqref{append:eq:amplitude_amplification_rotation} over $r$ gives a success probability bounded below by a universal constant.  The total number of calls accumulated over all preceding geometric stages is still
\begin{equation}
    \mathcal O\left(\frac{1}{\sin\vartheta_f}\right)
    =
    \mathcal O\left(\frac{\|\bm{a}\|_1}{\sqrt{\nu_f}}\right).
    \label{append:eq:unknown_success_AA_cost}
\end{equation}
Every successful trial has conditional state $\ket{G_f}$, so this unknown-success-probability version also produces an exact frequency sample.
\end{proof}

\subsection{Grid-to-frequency transformation and sampling correctness}
\label{append:subsec:grid_to_frequency}

We next transform the sampled expectation-value state into the desired mode-amplitude state.  The following lemma isolates this deterministic step.

\begin{lemma}[Grid-to-frequency transformation]
\label{append:lem:grid_to_frequency}
Given the state $\ket{\chi_f}$ in Eq.~\eqref{append:eq:coherent_expectation_lemma_state}, a tensor product of $d$ constant-size two-qubit unitaries prepares the frequency state in Eq.~\eqref{append:eq:direct_sampler_main_state}.  The transformation uses $\mathcal O(d)$ gates, and a computational-basis measurement of its output returns $\bomega$ according to $\mathrm p(\bomega)$.
\end{lemma}

\begin{proof}[Proof of Lemma~\ref{append:lem:grid_to_frequency}]
On each two-qubit block $\mathsf F^{(j)}$, introduce the following $4\times4$ unitary:
\begin{equation}
    \widetilde R=
    \begin{pmatrix}
        1/\sqrt{3} & 1/\sqrt{3} & 1/\sqrt{3} & 0\\
        \sqrt{2/3} & -1/\sqrt{6} & -1/\sqrt{6} & 0\\
        0 & 1/\sqrt{2} & -1/\sqrt{2} & 0\\
        0 & 0 & 0 & 1
    \end{pmatrix},
    \label{append:eq:local_trig_transform}
\end{equation}
where the computational-basis order is
$\{\ket{0},\ket{1},\ket{-1},\ket{11}\}
=\{\ket{00},\ket{01},\ket{10},\ket{11}\}$.
The upper-left $3\times3$ block mixes the amplitudes on the three valid codewords, while the fourth row and column leave the unused state $\ket{11}$ invariant.  Direct multiplication gives
\begin{equation}
    \widetilde R^\dagger\widetilde R
    =
    \begin{pmatrix}
        1&0&0&0\\
        0&1&0&0\\
        0&0&1&0\\
        0&0&0&1
    \end{pmatrix},
    \label{append:eq:local_trig_unitarity}
\end{equation}
so $\widetilde R$ is a valid two-qubit quantum gate.
Comparing Eqs.~\eqref{append:eq:uniform_grid_preparation_gate} and~\eqref{append:eq:local_trig_transform} also shows that
\begin{equation}
    \widetilde R=S_{\mathsf F^{(j)}}^\dagger.
    \label{append:eq:grid_preparation_transform_relation}
\end{equation}
Thus the same constant-size two-qubit gate synthesis can be used in reverse for uniform-grid preparation and in the forward direction for coefficient extraction.

\medskip
\noindent\underline{One-dimensional transformation.}
To see explicitly how $\widetilde R$ extracts the trigonometric coefficients, first take
\begin{equation}
    f(x)=c_0+c_1\cos x+c_{-1}\sin x.
\end{equation}
At the three grid points,
\begin{align}
    f(\alpha_0)&=c_0+c_1,\nonumber\\
    f(\alpha_1)&=c_0-\frac{c_1}{2}+\frac{\sqrt3\,c_{-1}}{2},\nonumber\\
    f(\alpha_{-1})&=c_0-\frac{c_1}{2}-\frac{\sqrt3\,c_{-1}}{2}.
    \label{append:eq:one_dim_sampled_values}
\end{align}
The corresponding normalized sampled-value state is
\begin{equation}
    \ket{\chi_f}
    =
    \frac{1}{\sqrt{3\nu_f}}
    \left[
        f(\alpha_0)\ket{0}
        +f(\alpha_1)\ket{1}
        +f(\alpha_{-1})\ket{-1}
    \right].
    \label{append:eq:one_dim_sampled_state}
\end{equation}
The output amplitudes on the three valid codewords are obtained by multiplying the rows of $\widetilde R$ by the sampled-value vector.  For the constant component,
\begin{align}
    \bra{0}\widetilde R\ket{\chi_f}
    &=
    \frac{f(\alpha_0)+f(\alpha_1)+f(\alpha_{-1})}
    {3\sqrt{\nu_f}}
    =
    \frac{c_0}{\sqrt{\nu_f}}.
    \label{append:eq:one_dim_constant_output}
\end{align}
For the cosine component, using
$f(\alpha_1)+f(\alpha_{-1})=2c_0-c_1$, we obtain
\begin{align}
    \bra{1}\widetilde R\ket{\chi_f}
    &=
    \frac{1}{\sqrt{3\nu_f}}
    \left[
        \sqrt{\frac23}f(\alpha_0)
        -\frac{f(\alpha_1)+f(\alpha_{-1})}{\sqrt{6}}
    \right]\nonumber\\
    &=
    \frac{1}{\sqrt{3\nu_f}}
    \left[
        \sqrt{\frac23}(c_0+c_1)
        -\frac{2c_0-c_1}{\sqrt{6}}
    \right]\nonumber\\
    &=
    \frac{c_1}{\sqrt{2\nu_f}}.
    \label{append:eq:one_dim_cosine_output}
\end{align}
For the sine component, using
$f(\alpha_1)-f(\alpha_{-1})=\sqrt{3}c_{-1}$, we obtain
\begin{align}
    \bra{-1}\widetilde R\ket{\chi_f}
    &=
    \frac{f(\alpha_1)-f(\alpha_{-1})}
    {\sqrt{6\nu_f}}
    =
    \frac{c_{-1}}{\sqrt{2\nu_f}}.
    \label{append:eq:one_dim_sine_output}
\end{align}
Consequently,
\begin{equation}
    \widetilde R\ket{\chi_f}
    =
    \frac{1}{\sqrt{\nu_f}}
    \left(
        c_0\ket{0}
        +
        \frac{c_1}{\sqrt2}\ket{1}
        +\frac{c_{-1}}{\sqrt2}\ket{-1}
    \right),
    \label{append:eq:one_dim_frequency_state}
\end{equation}
where the two-qubit encoding of $\ket{0}$, $\ket{1}$, and $\ket{-1}$ is given in Eq.~\eqref{append:eq:two_qubit_grid_encoding}.

\medskip
\noindent\underline{Tensor-product transformation.}
For a compact extension of the above calculation to $d$ dimensions, the valid-subspace matrix elements of $\widetilde R$ can be written as
\begin{equation}
    \bra{r}\widetilde R\ket{\bomega}
    =
    \frac{2^{\mathbbm1\{r\ne0\}/2}}{\sqrt{3}}
    \phi_r(\bm{\alpha}_{\bomega}),
    \qquad
    r,\omega\in\{0,1,-1\}.
    \label{append:eq:R_trig_matrix_element}
\end{equation}
Combining Eq.~\eqref{append:eq:R_trig_matrix_element} with Eq.~\eqref{append:eq:three_point_orthogonality} gives the local contraction identity
\begin{align}
    \sum_{\bomega\in\{0,1,-1\}}
    \bra{r}\widetilde R\ket{\bomega}
    \phi_s(\alpha_{\bomega})
    &=
    \frac{2^{\mathbbm1\{r\ne0\}/2}}{\sqrt3}
    \sum_{\bomega\in\{0,1,-1\}}
    \phi_r(\bm{\alpha}_{\bomega})\phi_s(\bm{\alpha}_{\bomega})
    \nonumber\\
    &=
    \sqrt3\,
    \delta_{r,s}
    2^{-\mathbbm1\{r\ne0\}/2}.
    \label{append:eq:local_trig_contraction}
\end{align}
Now expand $f(\bx)=\sum_{\bomega''}\bm{c}_{\bomega''}\Phi_{\bomega''}(\bx)$ with $\bm{c}_{\bomega''}=\Tr(\rho_0Q_{\bomega''})$.  In the following calculation, $\bomega'$ denotes the input-grid label, while $\bomega$ denotes the output frequency label.  The amplitude of $\bomega$ after applying $\widetilde R^{\otimes d}$ to Eq.~\eqref{append:eq:sampled_function_state} is
\begin{align}
    A_{\bomega}
    &:=
    \bra{\bomega}\widetilde R^{\otimes d}\ket{\chi_f}
    \nonumber\\
    &=
    \frac{1}{\sqrt{3^d\nu_f}}
    \sum_{\bomega'}
    f(\bx_{\bomega'})
    \prod_{j=1}^{d}
    \bra{\bomega_j}\widetilde R\ket{\bomega_j'}
    \nonumber\\
    &=
    \frac{1}{\sqrt{3^d\nu_f}}
    \sum_{\bomega''}\bm{c}_{\bomega''}
    \prod_{j=1}^{d}
    \left[
        \sum_{\bomega_j'\in\{0,1,-1\}}
        \bra{\bomega_j}\widetilde R\ket{\bomega_j'}
        \phi_{\bomega_j''}(\bm{\alpha}_{\bomega_j'})
    \right]\nonumber\\
    &=
    \frac{1}{\sqrt{3^d\nu_f}}
    \sum_{\bomega''}\bm{c}_{\bomega''}
    \prod_{j=1}^{d}
    \left[
        \sqrt3\,
        \delta_{\bomega_j,\bomega_j''}
        2^{-\mathbbm1\{\bomega_j\ne0\}/2}
    \right]\nonumber\\
    &=
    \frac{2^{-\|\bomega\|_0/2}\bm{c}_{\bomega}}
    {\sqrt{\nu_f}}.
    \label{append:eq:d_dim_frequency_amplitude}
\end{align}
In the last equality, the Kronecker deltas force $\bomega''=\bomega$, the $d$ factors of $\sqrt3$ cancel the prefactor $\sqrt{3^d}$, and
$\sum_{j=1}^d\mathbbm1\{\bomega_j\ne0\}=\|\bomega\|_0$.  Therefore,
\begin{equation}
    \widetilde R^{\otimes d}\ket{\chi_f}_{\mathsf F}
    =
    \frac{1}{\sqrt{\nu_f}}
    \sum_{\bomega\in\Lambda}
    2^{-\|\bomega\|_0/2}
    \Tr(\rho_0Q_{\bomega})
    \ket{\bomega}_{\mathsf F}
    =:\ket{\psi_f}_{\mathsf F}.
    \label{append:eq:direct_frequency_state}
\end{equation}
The normalization follows directly from
\begin{equation}
    \langle\psi_f|\psi_f\rangle
    =
    \frac{1}{\nu_f}
    \sum_{\bomega}
    2^{-\|\bomega\|_0}\bm{c}_{\bomega}^2
    =1.
    \label{append:eq:direct_frequency_normalization}
\end{equation}
Measuring the same $2d$ qubits in the computational basis then returns
\begin{equation}
    \mathrm p(\bomega)
    =
    |\langle\bomega|\psi_f\rangle|^2
    =
    \frac{2^{-\|\bomega\|_0}\Tr(\rho_0Q_{\bomega})^2}{\nu_f},
    \label{append:eq:direct_sampler_distribution}
\end{equation}
which is exactly the mode distribution used to define \DSE\ and to construct $\Lambda_{\mathsf q}$.
\end{proof}

\subsection{Computational complexity of the quantum subroutine---Proof of Lemma~\ref{append:lem:direct_mode_sampler}}
\label{append:subsec:direct_sampler_complexity}

\begin{proof}[Proof of Lemma~\ref{append:lem:direct_mode_sampler}]
We combine Lemmas~\ref{append:lem:sparse_Pauli_LCU}, \ref{append:lem:coherent_expectation_encoding}, and~\ref{append:lem:grid_to_frequency}.  Lemma~\ref{append:lem:coherent_expectation_encoding} shows that $\mathcal W_f$ prepares the sampled expectation-value state $\ket{\chi_f}$ in its good subspace with probability $\nu_f/\|\bm a\|_1^2$.  Lemma~\ref{append:lem:grid_to_frequency} then maps this conditional state deterministically to $\ket{\psi_f}$, whose computational-basis measurement is exactly distributed according to $\mathrm p$.  Reinitializing all registers makes repeated outcomes independent, and retaining their distinct values gives $|\Lambda_{\mathsf q}|\le m_f$.  It remains to derive the resource bounds.

Let $C_S$, $C_{\rm grid}$, $C_O$, and $C_R$ denote the costs of $S_{\mathsf F}$, one invocation of $\mathsf{CU}$, the block encoding $B_O$, and $\widetilde R^{\otimes d}$, respectively.  Lemma~\ref{append:lem:sparse_Pauli_LCU} supplies the cost of $B_O$, while the remaining constructions give
\begin{equation}
    C_S=\mathcal O(d),\qquad
    C_{\rm grid}=\mathcal O(G+d),\qquad
    C_O=\mathcal O(qN),\qquad
    C_R=\mathcal O(d).
    \label{append:eq:direct_cost_components}
\end{equation}
Because $\mathcal W_f$ contains $S_{\mathsf F}$, one forward grid circuit, $B_O$, and one inverse grid circuit,
\begin{align}
    C_{\mathcal W_f}
    &=
    C_S+2C_{\rm grid}+C_O\nonumber\\
    &=
    \mathcal O(qN+G+d).
    \label{append:eq:direct_preparation_cost_derivation}
\end{align}
The reflections about the all-zero state and the good subspace use $\mathcal O(N+d+a)$ gates and are dominated by Eq.~\eqref{append:eq:direct_preparation_cost_derivation} for $q\ge1$ and the standard sparse LCU block encoding.  This proves Eq.~\eqref{append:eq:direct_eval_cost}.

Without amplitude amplification, the number of independent executions until the first success is a geometric random variable with mean
\begin{equation}
    \mathbb E[T_{\rm post}]
    =
    \frac{1}{p_{\rm succ}}
    =
    \frac{\|\bm{a}\|_1^2}{\nu_f}.
    \label{append:eq:direct_geometric_mean}
\end{equation}
Multiplying Eq.~\eqref{append:eq:direct_geometric_mean} by $C_{\mathcal W_f}$ and adding the lower-order transformation cost $C_R$ proves Eq.~\eqref{append:eq:direct_repeat_cost}.

With amplitude amplification, Eq.~\eqref{append:eq:amplitude_amplification_rotation} shows that the number of uses of $\mathcal W_f$ and $\mathcal W_f^\dagger$ is
\begin{equation}
    T_{\rm AA}
    =
    \mathcal O\left(\frac{1}{\sqrt{p_{\rm succ}}}\right)
    =
    \mathcal O\left(\frac{\|\bm{a}\|_1}{\sqrt{\nu_f}}\right).
    \label{append:eq:direct_AA_iterations}
\end{equation}
The normalized good state is preserved throughout this rotation, so the final application of $\widetilde R^{\otimes d}$ and the computational-basis measurement still return an exact draw from $\mathrm p(\bomega)$.  Combining Eqs.~\eqref{append:eq:direct_preparation_cost_derivation} and~\eqref{append:eq:direct_AA_iterations} proves Eq.~\eqref{append:eq:direct_amplified_cost}.  Finally, reinitializing all registers makes successive outputs independent, and hence
\begin{align}
    C_{\rm total}
    &=
    m_{f}
    \left[
        \mathcal O\left(
            C_{\mathcal W_f}\frac{\|\bm{a}\|_1}{\sqrt{\nu_f}}
        \right)
        +C_R
    \right]\nonumber\\
    &=
    \mathcal{O}\left(
        (qN+G+d)m_{f}
        \frac{\|\bm{a}\|_1}{\sqrt{\nu_f}}
    \right),
    \label{append:eq:direct_total_cost_derivation}
\end{align}
where $\|\bm{a}\|_1/\sqrt{\nu_f}\ge1$ follows from $\nu_f\le\|\bm{a}\|_1^2$.  This proves Eq.~\eqref{append:eq:direct_total_cost} and completes the proof.
\end{proof}

\noindent
Taken together, this section establishes only the feature-identification interface: it produces independent draws from $\mathrm p(\bomega)$ and hence a random set $\Lambda_{\mathsf q}$ of distinct observed modes.  The number of draws $m_f$ required to control the omitted probability mass, and the resulting truncation and prediction errors of the classical surrogate, are analyzed separately in SI.~\ref{append:sec:surrogate_prediction}.

\section{Provably efficient \DSE-guided classical surrogate---Proof of Theorem~\ref{mt:thm:prediction_error_bound}}\label{append:sec:surrogate_prediction}
This section provides the proof of Theorem~\ref{mt:thm:prediction_error_bound}, which analyzes how the prediction error of the proposed \DSE-guided classical surrogate $h_{\mathsf{q}}$ depends on the dynamical stabilizer entropy (\DSE) $\mathcal{M}^{(1)}[O(\bx;U)]$ of the  quantum circuit ensemble $\{U(\bx)\}_{\bx \in [-\pi,\pi]^d}$, the number of training examples $n$, and the size of quantum system $N$. Recall that in the main text, the proposed prediction model is
\begin{equation}\label{eq:learning_surrogate_cs}    
    h_{\mathsf{q}}(\bx)=\frac{1}{n}\sum_{i=1}^{n} \kappa_{\mathsf{q}}(\bx,\bxi) y^{(i)}~\mbox{with}~\kappa_{\mathsf{q}}(\bx,\bxi) =\sum_{\bomega\in \Lambda_{\mathsf{q}}} 2^{\|\bomega\|_0} \Phi_{\bomega}(\bx)\Phi_{\bomega}(\bxi)\in\mathbb{R},
\end{equation}
where $y^{(i)}$ is the statistical estimation of $f(\bxi)$ from $m$ measurements, and $\Lambda_{\mathsf{q}} \subset \{0,1,-1\}^d$ is the subset of frequencies identified by the quantum subroutine described in Lemma~\ref{append:lem:direct_mode_sampler}. In particular, Let $m_f$ be the number of measurements on the amplitude-encoding states $\ket{\psi_{f}}_{\mathsf{F}}=\sum_{\bomega} \bm{s}_{\bomega}\sqrt{\mathrm{p}(\bomega)} \ket{\bomega}_{\mathsf{F}}$ defined in Eq.~\eqref{append:eq:amplitude_encode_state} to collect $\bomega\in \Lambda_{\mathsf{q}}$. In this regard, we have $|\Lambda_{\mathsf{q}}|\le m_{f}$.

Given the identified feature map $\{\Phi_{\bomega}(\bx)\}_{\bomega\in \Lambda_{\mathsf{q}}}$, we intend to prove that the average discrepancy between $h_{\mathsf{q}}(\bx)$ and the ground truth $f_O(\bx)=\Tr[\rho O(\bx;U)]$ with $O(\bx;U)=U(\bx)^{\dagger}OU(\bx)$ being the evolved operator when $\bx$ is uniformly and randomly sampled from $[-\pi,\pi]^d$, i.e., $\mathbb{E}_{\bx\sim [-\pi,\pi]^d} |h_{\mathcal{T}}(\bx)-f_O(\bx)|^2$, where the observable $O=\sum_{l=1}^q \bm{a}_l P_{l}$ has bounded norm $\sum_{l=1}^q |\bm{a}_l|\le 1$.

\begin{theorem}\label{thm:prediction_error_bound}
    Consider the function class $\mathcal{F}=\{f_O(\bx)=\Tr(\rho_0 U(\bx)^{\dagger}OU(\bx))|U\in \mathrm{Arc}(N,d),\bx\in [\-\pi,\pi]^d\}$ for the $N$-qubit circuits $U(\bx)$, an initial state $\rho_0=\ket{0}\bra{0}^N$, and a $q$-sparse Pauli observable $O=\sum_{l=1}^q \bm{a}_l P_l$ with bounded norm $\|O\|_{\infty}\le \|\bm{a}\|_1=\sum_{l=1}^q |\bm{a}_l|\le 1$. Suppose $U(\bx)$ consists of $G$ gates in total, including $d$ tunable $\RZ$ gates, with all remaining gates drawn from $\{\rm H,P,CNOT,T\}$. Let $\mathrm{p}(\bomega)=\tilde{\mathrm{p}}(\bomega)/\nu_f$ be the induced probability density function over the frequency $\bomega\in \Lambda$, where $\tilde{\mathrm{p}}(\bomega)=\mathbb{E}_{\bx\sim [-\pi,\pi]^d} \Tr(\rho_0 Q_{\bomega})^2\Phi_{\bomega}(\bx)^2$ and $\nu_f=\sum_{\bomega\in \Lambda}\tilde{\mathrm{p}}(\bomega)$. Suppose $\nu_f>0$ and let $m_f$ be the number of independent sampled frequencies. Then the quantum subroutine outputs a set $\Lambda_{\mathsf{q}}\subset \{0,\pm 1\}^d$ of distinct modes with runtime complexity $\mathcal{O}((qN+G+d)m_f\|\bm{a}\|_1/\sqrt{\nu_f})$, where
    \begin{equation}
        d_{\mathsf{q}}:=|\Lambda_{\mathsf{q}}| \le |m_f|= 2^{\frac{2\nu_f\cdot\mathcal{M}^{(1)}[O(\bx;U)]}{\epsilon}} \cdot \left( \frac{2\nu_f\cdot \mathcal{M}^{(1)}[O(\bx;U)]}{\epsilon} + \log\left(\frac{2}{\delta} \right) \right),
    \end{equation}
    such that the classical surrogate $h_{\mathsf{q}}$ constructed on $\{\Phi_{\bomega}(\bx)\}_{\bomega \in  \Lambda_{\mathsf{q}}}$ and trained on $\mathcal{T}=\{(\bxi,y^{(i)})\}_{i=1}^{n}$ with $n=4d_{\mathsf{q}}/{\epsilon\delta}$ yields with probability at least $1-\delta$,
    \begin{equation}
        \mathbb{E}_{\bx\sim [-\pi,\pi]^d}\left|f_O(\bx)-h_{\mathsf{q}}(\bx)\right|^2 \le \epsilon 
    \end{equation}
\end{theorem}

To reach Theorem~\ref{thm:prediction_error_bound}, we first use the orthogonality of the trigonometric basis to decompose the difference between the prediction of $h_{\mathsf{q}}(\bx)$ and ground truth $f(\bx)$ into the truncation error and the estimation error, i.e.,
\begin{align}\label{append:eqn:tri_decouple}
    \mathbb{E}_{\bx\sim [-\pi,\pi]^d} \left(f_O(\bx)-h_{\mathcal{T}}(\bx) \right)^2
    ~= ~\underbrace{\mathbb{E}_{\bx\sim [-\pi,\pi]^d} \left( f_O(\bx)-f_{\mathsf{q}}(\bx)\right)^2}_{T_1} + \underbrace{\mathbb{E}_{\bx\sim [-\pi,\pi]^d}\left(f_{\mathsf{q}}(\bx)-h_{\mathsf{q}}(\bx)\right)^2}_{T_2},
\end{align}
where $f_{\mathsf{q}}(\bx)=\sum_{\bomega \in {\Omega}_{\mathsf{q}}} \Phi_{\bomega}(\bx)\Tr(\rho_0 Q_{\bomega})$ denotes the truncated target function over the frequency set $\Lambda_{\mathsf{q}}$, and this equality follows
\begin{equation}
    \mathbb{E}_{\bx\sim [-\pi,\pi]^d}\left[\left(f_O(\bx)-f_{\mathsf{q}}(\bx)\right)\left(f_{\mathsf{q}}(\bx)-h_{\mathsf{q}}(\bx)\right)\right]=0.
\end{equation}

After decoupling, we then separately derive the upper bounds of these two terms, where the relevant results are
encapsulated in the following three lemmas whose proofs are given in the subsequent two sections.

\begin{lemma}[Truncation error of $T_1$]\label{append:lem:trunc_error_bound}
    Following notations in Theorem~\ref{thm:prediction_error_bound}. Let $\epsilon,\delta\in (0,1)$ be arbitrary positive real numbers. Let $\{\bomega^{(1)}, \cdots, \bomega^{(m_{f})}\}$ be $m_f$ independently sampled frequencies obtained by measuring the quantum state $\ket{\psi_f}$, and define $\Lambda_{\mathsf{q}}=\{\bomega^{(1)}, \cdots, \bomega^{m_f}\}_{\rm distinct}$ and $d_{\mathsf{q}}:=|\Lambda_{\mathsf{q}}|\le m_f$. If the number of sampled frequencies $m_{f}$ yields
    \begin{equation}
        m_{f} = 2^{\frac{2\nu_f\cdot\mathcal{M}^{(1)}[O(\bx;U)]}{\epsilon}} \cdot \left( \frac{2\nu_f\cdot \mathcal{M}^{(1)}[O(\bx;U)]}{\epsilon} + \log\left(\frac{2}{\delta} \right) \right),
    \end{equation}
    then we have $T_1 \le \epsilon/2$ with probability at least $1-\delta/2$.
\end{lemma}

\begin{lemma}  
    [Estimation error of $h_{\mathsf{q}}(\bx)$] 
    \label{append:lem:estimation_error_bound}
    Following notations in Theorem~\ref{thm:prediction_error_bound} and Lemma~\ref{append:lem:trunc_error_bound}. If the size of training dataset $\mathcal{T}=\{(\bxi,y^{(i)})\}_{i=1}^{n}$ satisfies
    \begin{equation}
        n =  \frac{4d_{\mathsf{q}}}{\epsilon\delta},
    \end{equation}
    then we have that with probability at least $1-\delta/2$, the estimation error yields $\mathbb{E}_{\bx} | h_{\mathsf{q}}(\bx)-f_{\mathsf{q}}(\bx|^2  \leq  \epsilon/2$
    where $f_{\mathsf{q}}(\bx)=\sum_{\bomega\in \Lambda_{\mathsf{q}}}\Phi_{\bomega}(\bx)\Tr(\rho_0 Q_{\bomega})$ denotes the truncated target function.
\end{lemma}

We are now ready to present the proof of Theorem~\ref{thm:prediction_error_bound}.

\begin{proof}
    [Proof of Theorem~\ref{thm:prediction_error_bound}] 
    We note that both the truncation error of and estimation error relies on the construction of the feature map $\{\Phi_{\bomega}(\bx)\}_{\bomega\in \Lambda_{\mathsf{q}}}$ as well as its dimension $d_{\mathsf{q}}=|\Lambda_{\mathsf{q}}|$. To identify the index set $\Lambda_{\mathsf{q}}$, we employs the quantum subroutine in Lemma~\ref{append:lem:direct_mode_sampler} to prepare the quantum state $$\ket{\psi_f}=\sum_{\bomega\in \Lambda} \bm{s}_{\bomega}\sqrt{\mathrm{p}(\bomega)}\ket{\bomega}$$ 
    with $\bm{s}_{\bomega}$ is a sign vector, 
    and make $m_f$ independent sampling by applying the computational-basis measurement $\{\ket{i}\bra{i}\}_{i=1}^{2^{2d}}$ to this state $\ket{\psi(\rho)}$, retaining only the distinct outcomes in $\Lambda_{\mathsf{q}}$. 
    
    Based on the feature map $\{\Phi_{\bomega}(\bx)\}_{\bomega\in \Lambda_{\mathsf{q}}}$, the difference between the prediction and ground truth can be obtained by integrating Lemma~\ref{append:lem:trunc_error_bound} and Lemma~\ref{append:lem:estimation_error_bound} into Eq.~\eqref{append:eqn:tri_decouple}. Mathematically, taking the sampling budget $m_f$ for construction $\Lambda_{\mathsf{q}}$ such that the dimension of $\Lambda_{\mathsf{q}}$ yields
    \begin{equation}
         d_{\mathsf{q}} =|\Lambda_{\mathsf{q}}| \le m_{f}=2^{2\nu_f\cdot\mathcal{M}^{(1)}[O(\bx;U)]/\epsilon} \cdot \left( \frac{2\nu_f\cdot \mathcal{M}^{(1)}[O(\bx;U)]}{\epsilon} + \log\left(\frac{2}{\delta} \right) \right),
    \end{equation}
    and the number of training examples satisfies $n = 4d_{\mathsf{q}}/{\epsilon \delta} $
    with probability at least $1-\delta$, we have
    \begin{equation}
        \mathbb{E}_{\bx\sim [-\pi,\pi]^d}  \left(h_{\mathsf{q}}(\bx)-f_O(\bx) \right)^2 \le \epsilon.
    \end{equation}
    Finally, Lemma~\ref{append:lem:direct_mode_sampler} shows that collecting the $m_f$ independent frequencies according to the quantum subroutine requires the time complexity of $\mathcal{O}((qN+G+d)m_f\|\bm{a}\|_1/\sqrt{\nu_f})$.  
    This completes the proof.
    
\end{proof}

\subsection{Truncation error bound for the \texorpdfstring{\DSE}{DSE}-guided classical surrogate \texorpdfstring{$h_{\mathsf{q}}$}{hq}---proof of Lemma~\ref{append:lem:trunc_error_bound}}
The proof of Lemma~\ref{append:lem:trunc_error_bound} employs the following lemma whose proof is deferred to Appendix~\ref{appendix:subsec:proof_lem9}.

\begin{lemma}\label{lem:cdf_sse_bound_1_order}
    Let $\mathrm{p}(\bomega)=2^{-\|\bomega\|_0}\Tr(\rho_0 Q_{\bomega})^2/\nu_f$ be a probability distribution induced by the circuit expectation $f(\bx)=\Tr(\rho_0 U(\bx)^{\dagger}OU(\bx))$ over the frequency $\bomega\in \{0,1,-1\}^d$. 
    Let $\mathcal{M}^{(1)}[O(\bx;U)]$ be the 1-order dynamical stabilizer entropy (\DSE) defined in Eq.~\eqref{append:eq:SSE}. Then, we have
    \begin{equation}
        \Pr_{\bomega \sim p}(\mathrm{p}(\bomega)\le \tau) \le \frac{\mathcal{M}^{(1)}[O(\bx;U)] }{\log(1/\tau)}.
    \end{equation}
\end{lemma}

We are now ready to present the proof of Lemma~\ref{append:lem:trunc_error_bound}.

\begin{proof}[Proof of Lemma~\ref{append:lem:trunc_error_bound}]
    We begin the proof by rewriting the problem of upper-bounding the truncation error term $T_1=\sum_{\Omega \backslash \Lambda_{\mathsf{q}}} 2^{-\|\bomega\|_0} \Tr(\rho_0 Q_{\bomega})^2$ as
    \begin{align}\label{appendix:eq:T1_p0_1}
        \frac{1}{\nu_f}\cdot T_1=& \frac{1}{\nu_f} \cdot \left( \nu_f- \sum_{\bomega \in \Lambda_{\mathsf{q}}}2^{-\|\bomega\|_0} \Tr(\rho_0 Q_{\bomega})^2 \right) 
        \nonumber \\
        = & 1-\sum_{\bomega \in \Lambda_{\mathsf{q}}}\frac{2^{-\|\bomega\|_0} \Tr(\rho_0 Q_{\bomega})^2}{\nu_f}
        \nonumber \\
        = & 1-\sum_{\bomega \in \Lambda_{\mathsf{q}}}\mathrm{p}(\bomega) 
    \end{align}
    where $\Lambda_{\mathsf{q}}=\{\bomega^{(1)}, \cdots, \bomega^{(m_f)}\}_{\rm distinct}$, and $\{\bomega^{(1)}, \cdots, \bomega^{(m_f)}\}$ refer to the independently sampled frequencies according to the probability distribution $\mathrm{p}(\bomega)$.

    For each $\bomega\in \Lambda=\{0,1,-1\}^d$, we define the count $s_{\bomega}=\sum_{i=1}^{m_f} \mathbbm{1}\{\bomega^{(i)}=\bomega\}$. Then we have $\bomega\in \Lambda_{\mathsf{q}}$ if and only if $s_{\bomega} \ge 1$ so that
    \begin{equation}\label{appendix:eq:T1_p0_2}
        1-\sum_{\bomega \in \Lambda_{\mathsf{q}}}\mathrm{p}(\bomega)  = \sum_{\bomega\in \Lambda} \mathrm{p}(\bomega) \mathbbm{1}\{s_{\bomega}=0\}.
    \end{equation}
    To derive the upper bound of this summation, we split the summation over $\bomega \in \Lambda$ into two parts, namely 
    \begin{equation}\label{appendix:eq:split_sum}
        \frac{1}{\nu_f}\cdot T_1=\sum_{\bomega\in \Lambda} \mathrm{p}(\bomega) \mathbbm{1}\{s_{\bomega}=0\}= \sum_{\bomega\in \Lambda_{\tau}} \mathrm{p}(\bomega) \mathbbm{1}\{s_{\bomega}=0\} +\sum_{\bomega\in \Lambda_{\tau}^{\rm c}} \mathrm{p}(\bomega) \mathbbm{1}\{s_{\bomega}=0\}
    \end{equation}
    where $\Lambda_\tau=\{\bomega\in \Lambda: \mathrm{p}(\bomega)\ge \tau\}$ denotes the subset of frequency whose probability is small than a specific value $\tau\in (0,1)$ and $ \Lambda_\tau^{\rm c}=\Omega\backslash \Lambda_\tau:= \{\bomega\in \Lambda: \mathrm{p}(\bomega)< \tau\}$ denotes the complement of $\Lambda_\tau$. In the following, we separately derive the upper bounds of these two summation terms over $\Lambda_\tau^{\rm c}$ and $\Lambda_\tau$.

    \smallskip

    \noindent \underline{\textit{Upper bound of $\sum_{\bomega\in \Lambda_{\tau}^{\rm c}} \mathrm{p}(\bomega) \mathbbm{1}\{s_{\bomega}=0\} $.}} For the summation term over the frequency set $\Lambda_{\tau}^{\rm c}$, employing Lemma~\ref{lem:cdf_sse_bound_1_order} immediately implies that 
    \begin{equation}\label{appendix:eq:sum_less_tau}
        \sum_{\bomega\in \Lambda_{\tau}^{\rm c}} \mathrm{p}(\bomega) \mathbbm{1}\{s_{\bomega}=0\} \le \sum_{\bomega\in \Lambda_{\tau}^{\rm c}} \mathrm{p}(\bomega) = \Pr_{\bomega \sim \mathrm{p}} \left(\mathrm{p}(\bomega) \le  \tau\right) \le \frac{\mathcal{M}^{(1)}[O(\bx;U)]}{\log(1/\tau)}.
    \end{equation}

    \smallskip

    \noindent \underline{\textit{Upper bound of $\sum_{\bomega\in \Lambda_{\tau}} \mathrm{p}(\bomega) \mathbbm{1}\{s_{\bomega}=0\}$ .}} For the summation term over the frequency set $\Lambda_{\tau}$, we will show that with high probability, 
    \begin{equation}
        \sum_{\bomega\in \Lambda_{\tau}} \mathrm{p}(\bomega) \mathbbm{1}\{s_{\bomega}=0\} =0.
    \end{equation}
    To achieve this, we first note that the event $\sum_{\bomega\in \Lambda_{\tau}} \mathrm{p}(\bomega) \mathbbm{1}\{s_{\bomega}=0\} >0$ means that there exists at least one frequency $\bomega \in \Lambda_{\tau}$ that was never sampled. Hence, by using the union bound, we have 
    \begin{align}\label{appendix:eq:prob_bound_1}
        & \Pr_{\bomega \sim p} \left(\sum_{\bomega\in \Lambda_{\tau}} \mathrm{p}(\bomega) \mathbbm{1}\{s_{\bomega}=0\} >0\right) 
        \nonumber \\
        \le & \Pr_{\bomega \sim p} \left(\bigcup_{\bomega\in \Lambda_{\tau}} \left[\mathrm{p}(\bomega) \mathbbm{1}\{s_{\bomega}=0\} >0 \right] \right) 
        \nonumber \\
        \le & \sum_{\bomega\in \Lambda_{\tau}}  \Pr_{\bomega \sim p} \left(\mathrm{p}(\bomega) \mathbbm{1}\{s_{\bomega}=0\} >0 \right) 
        \nonumber \\
        \le & \sum_{\bomega \in \Lambda_{\tau}} \Pr_{\bomega \sim p}(s_{\bomega}=0)
        \nonumber \\
        = & \sum_{\bomega \in \Lambda_{\tau}} \left(1-\mathrm{p}(\bomega)\right)^{m_f},
    \end{align}
    where the third inequality follows that the event $\mathrm{p}(\bomega) \mathbbm{1}\{s_{\bomega}=0\} >0 $ implies the event $\{s_{\bomega}=0\} $ and hence $\Pr(\mathrm{p}(\bomega) \mathbbm{1}\{s_{\bomega}=0\} >0) \le \Pr(s_{\bomega}=0)$, the last equality employs the observation that $s_{\bomega}=0$ means that none of the $m_f$ sampled frequency equals $\bomega$, which happens with probability $(1-\mathrm{p}(\bomega))^{m_f}$. Moreover, employing the definition of $\Lambda_{\tau}$ given in Eq.~\eqref{appendix:eq:split_sum}, we have $\mathrm{p}(\bomega) > \tau$ for $\bomega\in \Lambda_{\tau}$ and hence
    \begin{equation}\label{appendix:eq:prob_bound_2}
        \Pr_{\bomega \sim p} \left(\sum_{\bomega\in \Lambda_{\tau}} \mathrm{p}(\bomega) \mathbbm{1}\{s_{\bomega}=0\} >0\right) \le \sum_{\bomega \in \Lambda_{\tau}} \left(1-\mathrm{p}(\bomega)\right)^{m_f} \le \sum_{\bomega \in \Lambda_{\tau}}  e^{-m_f\mathrm{p}(\bomega)} \le  |\Lambda_{\tau}| e^{-m_f\tau},
    \end{equation}
    where the second inequality employs the fundamental inequality $(1-x) \le e^{-x}$. Then it remains to bound $|\Lambda_{\tau}|$. In particular, we have
    \begin{equation}\label{appendix:eq:size_Omega_complt}
        |\Lambda_{\tau}| < \sum_{\bomega\in \Lambda_{\tau}} \frac{\mathrm{p}(\bomega)}{\tau} \le \sum_{\bomega\in \Lambda} \frac{\mathrm{p}(\bomega)}{\tau} = \frac{1}{\tau},
    \end{equation}
    where the first inequality employs the definition of $\Lambda_{\tau}$ with $\mathrm{p}(\bomega)> \tau $, the second inequality follows $\Lambda_{\tau} \subset \Omega$, and the final equality uses the fact that $\sum_{\bomega}\mathrm{p}(\bomega)=1$. To the end,
    in conjunction with Eq.~\eqref{appendix:eq:prob_bound_1}, Eq.~\eqref{appendix:eq:prob_bound_2} and Eq.~\eqref{appendix:eq:size_Omega_complt}, we have 
    \begin{equation}\label{appendix:eq:sum_greater_tau}
        \Pr_{\bomega \sim p} \left(\sum_{\bomega\in \Lambda_{\tau}} \mathrm{p}(\bomega) \mathbbm{1}\{s_{\bomega}=0\} >0\right) \le \frac{e^{-m_f\tau}}{\tau}.
    \end{equation}

    Finally, combining Eq.~\eqref{appendix:eq:split_sum} with the results achieved in Eq.~\eqref{appendix:eq:sum_less_tau} and Eq.~\eqref{appendix:eq:sum_greater_tau}, with probability at least $1-e^{-m_f\tau}/\tau$, we have
    \begin{equation}
        \frac{1}{\nu_f} \cdot T_1 =  \left( 1-\sum_{\bomega\in \Lambda_{\mathsf{q}}} \mathrm{p}(\bomega) \right) = \sum_{\bomega\in\Lambda} \mathrm{p}(\bomega)\mathbbm{1}\{s_{\bomega}=0\} = \sum_{\bomega\in\Lambda_{\tau}^{\rm c}} \mathrm{p}(\bomega)\mathbbm{1}\{s_{\bomega}=0\} \le \frac{ \mathcal{M}^{(1)}[O(\bx;U)]}{\log(1/\tau)},
        \end{equation}
    where the first two equalities follow Eq.~\eqref{appendix:eq:prob_bound_1} and Eq.~\eqref{appendix:eq:prob_bound_2}.
    To ensure the error term $T_1$ is upper bounded by $\epsilon/2$ with probability $1-\delta/2$ for any $\epsilon,\delta>0$, by setting $\tau = 2^{-2\nu_f\cdot \mathcal{M}^{(1)}[O(\bx;U)]/\epsilon}$ and the sampling budget
    $$ m_f = \frac{1}{\tau} \cdot \log\left(\frac{2}{\delta \tau} \right) = 2^{2\nu_f \cdot \mathcal{M}^{(1)}[O(\bx;U)]/\epsilon} \cdot \log\left(\frac{2\cdot 2^{2\nu_f \cdot \mathcal{M}^{(1)}[O(\bx;U)]/\epsilon}}{\delta} \right),$$
    we have $e^{-m_f\tau}/\tau\le \delta/2$ and can obtain the final result
    \begin{equation}
        \Pr_{\bomega\sim p} \left(T_1 \le \frac{\epsilon}{2} \right) \ge 1-\frac{\delta}{2}.
    \end{equation}
    This completes the proof.
\end{proof}

\subsection{Estimation error of the \texorpdfstring{\DSE}{DSE}-guided classical surrogate \texorpdfstring{$h_{\mathsf{q}}$}{hq}---proof of Lemma~\ref{append:lem:estimation_error_bound}}
\label{append:subsec:proof-lemma2} 

The proof of the estimation error bound follows the conventions of Ref.~\cite{du2025efficient}, the core of the proof is to show that the \DSE-guided classical surrogate $h_{\mathsf{q}}(\bx)$ in Eq.~\eqref{eq:learning_surrogate_cs} is equal to the trigonometric expansion of the truncated target function $f_{\mathsf{q}}(\bx)=\sum_{\bomega\in \Lambda_{\mathsf{q}}} \Phi_{\bomega}(\bx)\Tr(\rho_0 Q_{\bomega})$ after averaging over the training data $\mathcal{T}=\{(\bxi, y^{(i)})\}_{i=1}^N$. This average includes the randomness of the sampled inputs $\bxi$ and statistical unbiased estimation of $y^{(i)}=\sum_{k=1}^m \bm{o}_k(\bxi)/m$ obtained from $m$ measurement outcomes $\bm{o}=(\bm{o}_1(\bxi), \cdots, \bm{o}_T(\bxi))$, namely $\mathbb{E}_{\bm{o}} y^{(i)} =  f_O(\bxi)=\Tr(\rho_0O(\bxi;U))$. We focus on the bounded observable $O=\sum_{l=1}^q\bm{a}_l P_l$ with $\|\bm{a}\|_1=\sum_{l=1}^q|\bm{a}_l|\le 1$ and $P_i$ being Pauli operators such that the single-shot measurement outcome $\bm{o}_1(\bxi)$ is the eigenvalue of the weighted Pauli observable $\bm{a}_lP_l$. 

\smallskip

\noindent \underline{Remark.} Instead of using the shadow-based estimation strategy in Ref.~\cite{du2025efficient} to construct the training data, which is particularly advantageous for estimating an exponential number of observables simultaneously, we focus on a fixed computationally efficient observable $O=\sum_{l=1}^q\bm{a}_lP_l$, where $P_l$ could be the global Pauli operators. In this case, the unbiased label $y^{(i)}$ can be estimated more efficiently using measurement strategies tailored to the observable, such as direct measurement, Pauli grouping, or derandomized shadow methods. Moreover, our results could be extended to more general bounded observables as long as the observable has a classically efficient representation and its block-encoding state could be efficiently implemented. This also includes the projector measurement. This setting allows classical surrogates to be applied to a broader class of observables beyond the locality restrictions inherent in standard shadow-based methods.
 
\begin{proof}[Proof of Lemma~\ref{append:lem:estimation_error_bound}]
	
    Most of this proof follows the conventions of Ref.~\cite[Lemma~F.2]{du2025efficient}. The first step is to prove the equivalence between the expectation of the classical surrogate $h_{\mathcal{T}}$ and the truncated target quantum state $f_{\mathsf{q}}(\bx)$, i.e., $\mathbb{E}_{\mathcal{T}}h_{\mathcal{T}}(\bx) = f_{\mathsf{q}}(\bx)$. Following the explicit form of $h_{\mathcal{T}}(\bx)$ in Eq.~(\ref{eq:learning_surrogate_cs}), we have  
    \allowdisplaybreaks
    \begin{subequations}
    	\begin{eqnarray}\label{eqn:exp-model-recover-true-state-trigeo}
    	\mathbb{E}_{\mathcal{T}}h_{\mathcal{T}}(\bx) = &&  \frac{1}{n}\sum_{i=1}^{n} \mathbb{E}_{\bxi \sim [-\pi, \pi]^d} \left[\kappa_{\mathsf{q}}\left(\bx, \bxi\right) \right]\mathbb{E}_{\bm{o}^{(i)}} y^{(i)}\label{eqn:exp-model-recover-true-state-trigeo-1} \\
    	= &&  \mathbb{E}_{\bx^{(1)}\sim [-\pi, \pi]^{d}} \left[\kappa_{\mathsf{q}}\left(\bx, \bx^{(1)}\right) \right]\Tr(\rho(\bx^{(1)})O) \label{eqn:exp-model-recover-true-state-trigeo-2} \\
        = && \sum_{\bomega\in \Lambda_{\mathsf{q}}} \Phi_{\bomega}(\bx)\Tr(\rho_0 Q_{\bomega}) \\
       = && f_{\mathsf{q}}(\bx),
    \end{eqnarray}
    \end{subequations}
    where $\bm{o}^{(i)}=(\bm{o}_1(\bxi), \cdots, \bm{o}_m(\bxi))$ denotes the randomized measurement outcome for the state $\rho(\bxi)$, the second equality employs the unbiasedness of the label $y^{(i)}$, the final equality follows the derivation of Ref.~\cite[Eq.~(34b)-(34i)]{du2025efficient}.

    The second step is to analyze the estimation error  $\mathbb{E}_{\bx\sim [-\pi ,\pi]^d}  | h_{\mathcal{T}}(\bx) - f_{\mathsf{q}}(\bx)) |^2 $, namely
    \allowdisplaybreaks
    \begin{align}
    	& \mathbb{E}_{\bx\sim [-\pi, \pi]^d} \left |h_{\mathcal{T}}(\bx) - f_{\mathsf{q}}(\bx))\right|^2  \\
    	= & ~\mathbb{E}_{\bx\sim [-\pi, \pi]^d} \left|\frac{1}{n}\sum_{i=1}^{n} \kappa_{\mathsf{q}}(\bx, \bxi) y^{(i)} -  \sum_{\bomega\in \Lambda_{\mathsf{q}}}\Phi_{\bomega}(\bx) \Tr(\rho_0 Q_{\bomega})\right|^2  \\
    	= & ~  \sum_{\bomega\in \Lambda_{\mathsf{q}}}  2^{-\|\bomega\|_0}    \left| \frac{1}{n}\sum_{i=1}^{n} 2^{\|\bomega\|_0}  \Phi_{\bomega}(\bxi) y^{(i)} -   \Tr(\rho_0 Q_{\bomega})   \right|^2  \\
        = & ~  \sum_{\bomega\in \Lambda_{\mathsf{q}}}      \left| \frac{1}{n}\sum_{i=1}^{n} 2^{\|\bomega\|_0/2}  \Phi_{\bomega}(\bxi) y^{(i)} -   2^{-\|\bomega\|_0/2} \Tr(\rho_0 Q_{\bomega})   \right|^2
        \label{subeqn:exp-kernel-Lowesa-2} 
    \end{align}
    where the second equality employs the evaluation of the orthogonality of basis functions $\Phi_{\bomega}(\bx)$, namely $\mathbb{E}_{\bx\sim [-\pi,\pi]}\Phi_{\bomega}(\bx) \Phi_{\bomega'}(\bx)=\delta_{\bomega \bomega'} 2^{-\|\bomega\|_0}$.
    For each $\bomega\in \Lambda_{\mathsf{q}}$ and each training examples $(\bxi, y^{(i)})$, we define the random variable conditioned on the sampled set $\Lambda_{\mathsf{q}}$ as $
    X_{i,\bomega}:=2^{\|\bomega\|_0/2} \Phi_{\bomega}(\bxi)y^{(i)}.$
    Employing the unbiasedness of $y^{(i)}=\mathbb{E}_{\bm{o}^{(i)}} f(\bxi)$, we have
    \begin{equation}
        \mathbb{E}_{\bxi,\bm{o}^{(i)}}[X_{i,\bomega}] = 2^{\|\bomega\|_0/2} \mathbb{E}_{\bxi,\bm{o}^{(i)}} \Phi_{\bomega}(\bxi)y^{(i)}= 2^{\|\bomega\|_0/2} 2^{-\|\bomega\|_0}\Tr(\rho_0 Q_{\bomega}) =  2^{-\|\bomega\|_0/2}\Tr(\rho_0 Q_{\bomega}).
    \end{equation}
    Moreover, using the orthogonality relation $ \mathbb{E}_{\bx\sim [-\pi,\pi]^d} \Phi_{\bomega}(\bx)^2=2^{-\|\bomega\|_0}$ and the bounded-label $|y^{(i)}|\le 1$, we obtain the variance of $X_{i,\bomega}$ as 
    \begin{equation}
        \mathrm{Var}(X_{i,\bomega}) \le \mathbb{E}[X_{i,\bomega}^2] =  2^{\|\bomega\|_0}\mathbb{E}_{\bxi,\bm{o}^{(i)}} \left[\Phi_{\bomega}(\bxi)^2(y^{(i)})^2 \right] \le 1.
    \end{equation}
    Since the training examples are independent, we can write the expectation of the estimation error over the training dataset as 
    \begin{align}
        \mathbb{E}_{\mathcal{T}} = & \mathbb{E}_{\mathcal{T}} \mathbb{E}_{\bx\sim [-\pi,\pi]^d } \left| h_{\mathsf{q},\mathcal{T}}(\bx)-f_{\mathsf{q}}(\bx)\right|^2
        \nonumber \\
        = & \mathbb{E}_{\mathcal{T}} \sum_{\bomega\in \Lambda_{\mathsf{q}}} \left| \frac{1}{n} \sum_{i=1}^nX_{i,\bomega} -\mathbb{E}_{\bxi,\bm{o}^{(i)}} X_{i,\bomega}  \right|^2
        \nonumber \\
        = &  \sum_{\bomega\in \Lambda_{\mathsf{q}}} \mathrm{Var}\left(\frac{1}{n}\sum_{i=1}^nX_{i,\bomega} \right) \le \frac{d_{\mathsf{q}}}{n}.
    \end{align}
    Applying Markov's inequality to the nonnegative random variable $T_2$ gives
    \begin{align}
        \Pr \left[T_2 \ge \frac{\epsilon}{2} \right] \le \frac{2\mathbb{E}_{\mathcal{T}} T_2}{\epsilon} \le \frac{2d_{\mathsf{q}}}{n\epsilon}.
    \end{align}
    In this regard, when the number of training examples yields
    \begin{equation}
        n = \frac{4d_{\mathsf{q}}}{\epsilon\delta}
    \end{equation}
    for a constant $\delta>0$, then we have $T_2 \le \epsilon/2$ with probability at least $1 - \delta/2$. This completes the proof.
\end{proof}

\subsection{Proof of Lemma~\ref{lem:cdf_sse_bound_1_order}} 
\label{appendix:subsec:proof_lem9}

\begin{proof}
    [Proof of Lemma~\ref{lem:cdf_sse_bound_1_order}] The proof uses Markov's inequality, $\Pr(X\ge a) \le \mathbb{E}[X]/a$, which holds for any non-negative random variable $X$ and $a>0$. First, let $X:=-\log(\mathrm{p}(\bomega))$, where $\bomega \in \{0,1,-1\}^{d}$ follows the distribution $\mathrm{p}(\bomega)$. Then we have
    \begin{equation}
        \mathbb{E}[X]=\sum_{\bomega} \mathrm{p}(\bomega) \log\left(\frac{1}{\mathrm{p}(\bomega)} \right) = \mathcal{S}^{(1)}[O(\bx;U)].
    \end{equation}
    Hence, using Markov's inequality, we find that for any $a>0$,
    \begin{equation}
        \Pr_{\bomega \sim p}(-\log(\mathrm{p}(\bomega))\ge a) \le \frac{\mathcal{S}^{(1)}[O(\bx;U)]}{a}  ~~\Leftrightarrow~~\Pr_{\bomega \sim p}(\mathrm{p}(\bomega)\le 2^{-a}) \le \frac{\mathcal{S}^{(1)}[O(\bx;U)]}{a},
    \end{equation}
    or equivalently, letting $\tau=2^{-a}$,
    \begin{equation}
        \sum_{\bomega\in \Lambda\backslash \Lambda_{\tau}}\mathrm{p}(\bomega)=\Pr_{\bomega \sim p}(\mathrm{p}(\bomega)\le \tau) \le \frac{\mathcal{S}^{(1)}[O(\bx;U)] }{\log(1/\tau)} \le \frac{\mathcal{M}^{(1)}[O(\bx;U)] }{\log(1/\tau)},
    \end{equation}
    where the second inequality follows the definition of $\calM^{(1)}[O(\bx;U)]$.
\end{proof}

\section{Computational hardness for purely classical algorithms to emulate observable expectations}\label{append:sec:classical-hard}

In this section, we complement the efficient surrogate-emulation result with a complexity-theoretic limitation for purely classical algorithms. Recall that the learnability guarantee established in Theorem~\ref{mt:thm:prediction_error_bound} applies to a classical predictor trained with quantum-generated data. By contrast, a purely classical simulator receives only the classical description of the circuit and the observable, and has no access to measurement data from the quantum device. We prove that, under the standard complexity-theoretic conjecture $\mathsf{BQP}\not\subseteq\mathsf{BPP}$, there exists a circuit family that is efficiently learnable by the $\DSE$-guided surrogate, yet cannot be emulated in polynomial time by any purely classical algorithm. In particular, for a fixed sufficiently small constant target error $\epsilon>0$, we will construct $N$-qubit quantum circuits $U_{\ell}(\bx)$ containing $d=\mathcal{O}(\min\{1/\epsilon,N\})$ tunable rotation gates, which nevertheless yield a logarithmically scaling \DSE\ for a specific observable $O$, namely $\mathcal{M}^{(1)}[O(\bx;U_{\ell})]=\mathcal{O}(\log N)$.

The organization of this section is as follows. In SI.~\ref{append:subsec:construction}, we present the explicit construction of the circuit family and establish the separation, i.e., the constructed family is efficiently learnable by the $\DSE$-guided surrogate but is intractable for any purely classical algorithm unless $\mathsf{BQP}\subseteq\mathsf{BPP}$. In SI.~\ref{append:subsec:proof_lemmas}, we present the proofs of the two technical lemmas used in the above analysis.

\subsection{Construction of the circuit family and the separation}\label{append:subsec:construction}
The construction of the quantum circuits is based on the computational hardness of the $\mathsf{BQP}$ complexity class, whose hardness is embedded into the parameter-independent blocks of the circuit, while the tunable rotation gates are arranged to generate a low-\DSE trigonometric structure. In this regard, we first recall the $\mathsf{BQP}$ complexity class.

\smallskip

\noindent\textbf{$\mathsf{BQP}$ complexity class}. $\mathsf{BQP}$ (short for bounded-error quantum polynomial time) is the complexity class of promise problems that can be solved by polynomial-time quantum computations with a small probability of error~\cite{watrous2008quantum}. Concretely, for a language $L\in\mathsf{BQP}$ and an input $\ell$, there exists a uniform family of polynomial-size quantum circuits $\{W_\ell\}$ acting on $m=m(\ell)=\mathrm{poly}(|\ell|)$ qubits initialized in $\ket{0^m}$, such that measuring the first output qubit in the computational basis yields an acceptance probability $p_\ell$ satisfying $p_\ell\ge 2/3$ if $\ell\in L$ and $p_\ell\le 1/3$ if $\ell\notin L$. It is widely believed that $\mathsf{BQP}\not\subseteq\mathsf{BPP}$; otherwise, problems such as integer factoring would admit randomized polynomial-time classical algorithms.

\smallskip
\noindent\textbf{Explicit construction}. Fix a language $L\in\mathsf{BQP}$. For an arbitrary input $\ell$, standard error amplification yields a polynomial-size circuit $W_\ell$ on $m=\mathrm{poly}(|\ell|)$ qubits, compiled over the universal gate set $\{H,S,\CNOT,T\}$, whose first output qubit obeys
\begin{equation}\label{append:eq:dpos_mu}
\mu_\ell := \bra{0^m} W_\ell^{\dagger} Z_1 W_\ell\ket{0^m} = 1-2p_\ell,
\qquad
\ell\in L \Rightarrow \mu_\ell\le -\tfrac13,
\quad
\ell\notin L \Rightarrow \mu_\ell\ge \tfrac13 .
\end{equation}
In other words, deciding the sign of $\mu_\ell$ is equivalent to deciding the membership of $\ell$ in $L$. Let $N=m+d$, where the $d$ additional qubits, indexed by $m+1,\dots,m+d$, serve as \emph{probe} qubits. We then define the $N$-qubit circuit family
\begin{equation}\label{append:eq:dpos_U}
\mathcal{U}_L=\left\{U_{\ell}(\bx) = W_\ell \otimes \prod_{j=1}^{d}\RZ_{m+j}(\bx_j)\,H_{m+j}~\Big|~ \ell \in \{0,1\}^{*},~ \bx\in[-\pi,\pi]^d \right\},
\end{equation}
together with the fixed initial state $\rho_0=\ket{0}\bra{0}^{\otimes N}$ and the fixed observable
\begin{equation}\label{append:eq:dpos_O}
O = \frac1d\sum_{j=1}^d Z_1\,X_{m+j},
\qquad \sum_{l}\|O_l\|_{\infty}=1 .
\end{equation}
Each $U_{\ell}(\bx)$ is an instance of the $\RZ+\CI+T$ circuit of Eq.~\eqref{eq:target_function_set} with $d$ tunable rotations, where the parameter-independent gates, namely $W_\ell$ and the Hadamards, form the fixed blocks. In this way, the entire $\mathsf{BQP}$ hardness is placed into the fixed block $W_\ell$, whereas the $d$ tunable probes generate the trigonometric structure of the target function. We emphasize that both $\rho_0$ and $O$ are prescribed and independent of $\ell$, so that the unknown instance is carried solely by the circuit.

The following lemma shows that the mean-value function of the constructed family admits a product form, and that the resulting $\DSE$ for the evolved observable $O(\bx;U_{\ell})=U_{\ell}(\bx)^{\dagger}OU_{\ell}(\bx)$ remains logarithmic in $N$.

\begin{lemma}[Product form and $\DSE$ bound]\label{append:lem:dpos_sse}
For the circuit family $\mathcal{U}_L$ and the observable $O$ defined in Eqs.~\eqref{append:eq:dpos_mu}--\eqref{append:eq:dpos_O}, the mean-value function yields
\begin{equation}
f_{\ell}(\bx)=\Tr\big(\rho_0\, U_{\ell}(\bx)^{\dagger} O\, U_{\ell}(\bx)\big) = \mu_\ell\, h(\bx),
\quad \text{with}\quad h(\bx)=\frac{1}{d}\sum_{j=1}^d\cos(\bx_j).
\end{equation}
Consequently, we have $\nu_f=\mathbb{E}_{\bx}f_{\ell}(\bx)^2=\mu_\ell^2/(2d)$ and $\mathcal{M}^{(1)}[O(\bx;U_{\ell})] \le \log(2d)=\mathcal{O}(\log N)$.
\end{lemma}

Building on the product form $f_{\ell}(\bx)=\mu_\ell h(\bx)$, the next lemma provides a reduction showing that any accurate hypothesis $\hat{f}$ of the target function $f_{\ell}$ can be converted into a decision procedure for the sign of $\mu_\ell$. The intuition is that, since the $\bx$-dependence of $f_{\ell}$ is entirely carried by the \emph{known} function $h(\bx)$, correlating a hypothesis $\hat f$ against $h$ isolates the unknown constant $\mu_\ell$.

\begin{lemma}[Reduction to the decision problem of $\mu_{\ell}$]\label{append:lem:dpos_reduction}
Let $\hat f$ satisfy $\mathbb{E}_{\bx}|\hat{f}(\bx)-f_{\ell}(\bx)|^2\le\epsilon$ with $\epsilon\le\nu_f/16$. Then the estimator
\begin{equation}
\hat\mu:=2d\,\mathbb{E}_{\bx}\big[\hat f(\bx)\,h(\bx)\big]
\quad\text{obeys}\quad
|\hat\mu-\mu_\ell|\le \tfrac14|\mu_\ell| ,
\end{equation}
so that $\mathrm{sign}(\hat\mu)=\mathrm{sign}(\mu_\ell)$ decides whether $\ell\in L$. Moreover, $\hat\mu$ can be computed to accuracy $\tfrac18|\mu_\ell|$ from $M=\mathcal{O}(d^2\log(1/\delta))$ evaluations of $\hat f$ and $h$, with probability at least $1-\delta$.
\end{lemma}

Supported by Lemmas~\ref{append:lem:dpos_sse} and \ref{append:lem:dpos_reduction}, whose proofs are deferred to SI.~\ref{append:subsec:proof_lemmas}, we are now ready to present the formal statement of the separation, in which the constructed family is simultaneously efficiently learnable by the $\DSE$-guided surrogate and hard for any purely classical algorithm.

\begin{theorem-non}[Formal statement of Theorem~\ref{coro:hardness_classical_simulat}]\label{append:thm:dpos_separation}
Let $L\in\mathsf{BQP}$ be an arbitrary language and fix a universal constant $0< \epsilon\le 1/288$. Consider the $N$-qubit circuit family $\mathcal{U}_L$ in Eq.~\eqref{append:eq:dpos_U} with $d\ge 1$ satisfying $d\epsilon \le 1/288$, the fixed initial state $\rho_0$, and the fixed observable $O$ in Eq.~\eqref{append:eq:dpos_O}. The induced target function class is
\begin{equation}
    \mathcal{F}_{L}=\left\{f_{\ell}(\bx)=\Tr\big(\rho_0\, U_{\ell}(\bx)^{\dagger} O\, U_{\ell}(\bx)\big) = \mu_\ell\, h(\bx)~\Big|~\ell \in \{0,1\}^{*},~ \bx\in [-\pi,\pi]^d\right\}.
\end{equation}
Then $\mathcal{F}_{L}$ satisfies the following two properties.
\begin{itemize}
\item[\textnormal{(i)}] \textnormal{(Classically learnable.)} For every input $\ell$, the $\DSE$-guided surrogate $h_{\mathsf q}$ achieves $\mathbb{E}_{\bx}|h_{\mathsf q}(\bx)-f_{\ell}(\bx)|^2\le\epsilon$ using a feature dimension $d_{\mathsf q}\le m_f=\mathrm{poly}(N)$ and $n=\mathrm{poly}(N)$ quantum-labeled samples.
\item[\textnormal{(ii)}] \textnormal{(Classically non-simulable.)} Suppose there exists a randomized polynomial-time classical algorithm which, without access to quantum-generated data, outputs for every input $\ell$ a hypothesis $\hat f$ satisfying $\mathbb{E}_{\bx}|\hat f(\bx)-f_{\ell}(\bx)|^2\le\epsilon$. Then $\mathsf{BQP}\subseteq\mathsf{BPP}$.
\end{itemize}
\end{theorem-non}

\noindent\underline{Remark}. The quantifier in property (ii) is over \emph{all} inputs $\ell$ rather than a single one, since deciding the membership $\ell\in L$ requires the classical algorithm to succeed on both the yes-instances and the no-instances of the promise problem. Accordingly, the index set of $\mathcal{U}_L$ and $\mathcal{F}_L$ ranges over all bit strings $\ell\in\{0,1\}^{*}$, and the language $L$ enters only through the sign of $\mu_\ell$ in Eq.~\eqref{append:eq:dpos_mu}.

\begin{proof}[Proof of Theorem~\ref{coro:hardness_classical_simulat}]
This proof is composed of two parts, which separately establish the classical learnability and the classical hardness of the constructed family.

\smallskip
\noindent\underline{Classical learnability}. This part follows by directly invoking Theorem~\ref{mt:thm:prediction_error_bound}, whose hypotheses are all met by the constructed target function class $\mathcal{F}_L$. Specifically, the initial state is the fixed product state $\rho_0=\ket{0}\bra{0}^{\otimes N}$, the circuit $U_{\ell}(\bx)$ has polynomial depth, and the observable $O$ is a sum of weight-two Pauli operators admitting direct eigenvalue measurement. Supported by Lemma~\ref{append:lem:dpos_sse}, we have $\mathcal{M}^{(1)}[O(\bx;U_{\ell})]\le\log(2d)$. Since $|\mu_{\ell}|\ge 1/3$ and $d\epsilon \le 1/288$, we have
\begin{equation}
    \epsilon \le \frac{1}{288d}\le \frac{\mu_{\ell}^2}{32d}=\frac{\nu_f}{16}.
\end{equation}
For the observable in Eq.~\eqref{append:eq:dpos_O}, we have $q = d, \bm{a}_j = 1/d$, and hence $\|\bm{a}\|_1 =\sum_{j=1}^d|\bm{a}_j| = 1$. Moreover, we have
\begin{equation}
    \frac{\|\bm{a}\|_1^2}{\nu_f} = \frac{2d}{\mu_{\ell}^2}\le 18d =\mathcal{O}(1/\epsilon)=\mathcal{O}(\poly(N)),
\end{equation}
where the inequality uses $|\mu_{\ell}|\ge 1/3$ and the scaling follows from $d\epsilon \le 1/288$. Moreover, both $\epsilon$ and $d=\mathcal{O}(1/\epsilon)$ are fixed independently of $N$. Hence, at the target accuracy $\epsilon$, the exponent governing the model dimension obeys
\begin{equation}
    \frac{2\nu_f \mathcal{M}^{(1)}[O(\bx;U_{\ell})]}{\epsilon}\le \frac{\log(2d)}{d\mu_\ell^2}=\mathcal{O}(1).
\end{equation}
Consequently, the draw budget satisfies $m_f=\tilde{\mathcal{O}}(2^{2\nu_f \mathcal{M}^{(1)}[O(\bx;U_{\ell})])/\epsilon})=\poly(N)$ for the inverse-polynomial error $\epsilon$. Since $d_{\mathsf{q}}\le m_f$, the feature dimension is also polynomial in $N$. Moreover, choosing $n = 4d_{\mathsf q}/(\epsilon\delta)$ gives $n = \poly(N, 1/\delta)$. Taken together, the \DSE-guided surrogate $h_{\mathsf{q}}$ emulates this family with polynomial sample complexity and runtime.

\smallskip
\noindent\underline{Hardness for purely classical algorithms}. Suppose, for the sake of the argument, that a randomized polynomial-time classical algorithm without access to quantum-generated data outputs, for every input $\ell$, a hypothesis $\hat f$ satisfying $\mathbb{E}_{\bx}|\hat f(\bx)-f_{\ell}(\bx)|^2\le\epsilon$. As shown above, $d\epsilon\le 1/288$ and $|\mu_{\ell}|\ge 1/3$ imply $\epsilon\le\nu_f/16$. Employing Lemma~\ref{append:lem:dpos_reduction}, the estimator $\hat\mu$ can be computed from $\hat f$ in randomized polynomial time and obeys $\mathrm{sign}(\hat\mu)=\mathrm{sign}(\mu_\ell)$. By Eq.~\eqref{append:eq:dpos_mu}, the sign of $\mu_\ell$ decides whether $\ell\in L$, so the above procedure decides the language $L$ in randomized polynomial time. Since $L\in\mathsf{BQP}$ was arbitrary, this yields a randomized polynomial-time classical decision procedure for every language in $\mathsf{BQP}$, and hence $\mathsf{BQP}\subseteq\mathsf{BPP}$. This completes the proof.
\end{proof}

\subsection{Proof of Lemmas~\ref{append:lem:dpos_sse} and \ref{append:lem:dpos_reduction}}\label{append:subsec:proof_lemmas}

Here we present the proofs of the two technical lemmas used in the separation.

\begin{proof}[Proof of Lemma~\ref{append:lem:dpos_sse}]
This proof consists of two parts, which separately derive the product form of $f_{\ell}(\bx)$ given the evolved observable $O(\bx,U_{\ell})=U_{\ell}^{\dagger}(\bx)OU_{\ell}(\bx)$ defined in Eq.~\eqref{append:eq:dpos_U}-\eqref{append:eq:dpos_O} and the upper bound of the related $\DSE$.

\smallskip
\noindent\underline{Explicit form of $f_{\ell}(\bx)$}. By construction, the circuit takes the tensor-product form $U_{\ell}(\bx)=W_\ell\otimes V_B(\bx)$ with $V_B(\bx)=\prod_{j=1}^{d}\RZ_{m+j}(\bx_j)H_{m+j}$, and the observable takes the tensor-product form $O=Z_1\otimes\big(\tfrac1d\sum_{j=1}^{d} X_{m+j}\big)$. Since conjugation preserves tensor products, the evolved observable yields
\begin{equation}\label{append:eq:evolved_tensor}
U_{\ell}(\bx)^{\dagger} O\, U_{\ell}(\bx) = \big(W_\ell^{\dagger}Z_1 W_\ell\big)\otimes \frac1d\sum_{j=1}^d V_B(\bx)^{\dagger} X_{m+j} V_B(\bx).
\end{equation}
Because the initial state $\rho_0=\ket{0}\bra{0}^{\otimes N}$ is also a product state across the same bipartition, the expectation value of the tensor-product operator in Eq.~\eqref{append:eq:evolved_tensor} equals the product of the two registers' expectation values. The first factor is exactly the $\mathsf{BQP}$ value, i.e., $\bra{0^m}W_\ell^{\dagger}Z_1W_\ell\ket{0^m}=\mu_\ell$, which is independent of $\bx$. For the second factor, a single-qubit calculation gives
\begin{equation}
    V_B(\bx)^{\dagger}X_{m+j}V_B(\bx)=H\big(\cos(\bx_j)\,X-\sin(\bx_j)\,Y\big)H=\cos(\bx_j)\,Z_{m+j}-\sin(\bx_j)\,Y_{m+j},
\end{equation}
so that using $\bra{0}Z\ket{0}=1$ and $\bra{0}Y\ket{0}=0$, its expectation value is $\tfrac1d\sum_{j=1}^{d}\cos(\bx_j)=h(\bx)$, which depends only on the input $\bx$. Multiplying the two factors yields $f_{\ell}(\bx)=\mu_\ell\,h(\bx)$. Employing the orthogonality relation $\mathbb{E}_{\bx}\cos^2(\bx_j)=\tfrac12$ then gives $\nu_f=\mu_\ell^2\,\mathbb{E}_{\bx}h(\bx)^2=\mu_\ell^2/(2d)$.
 
\smallskip

\noindent \underline{$\DSE$ bound for $O(\bx,U_\ell)$}. As shown in Eq.~\eqref{append:eq:evolved_tensor}, the evolved observable is a sum of single-probe terms, so its trigonometric expansion contains only frequencies of Hamming weight $\|\bomega\|_0\le 1$. Explicitly, its frequency support consists of those modes $\bomega\in\{0,\pm 1\}^d$ that are nonzero in exactly one entry, i.e., $\bomega_j\in\{+1,-1\}$ for a single index $j\in[d]$ and $\bomega_{j'}=0$ for all $j'\ne j$. The number of such modes is $2d$, since there are $d$ choices of the index $j$ and two choices of the sign of $\bomega_j$. Recall that in the definition of $\mathcal{M}^{(\alpha)}$ in Eq.~\eqref{eq:SSE}, a fixed unitary $V$ maps each operator coefficient as $Q_{\bomega}\mapsto V^{\dagger}Q_{\bomega}V$ but, being independent of $\bx$, leaves the frequency support unchanged. Hence for \emph{every} fixed $V$, the induced distribution $\mathrm{p}(\bomega)$ is supported on at most $2d$ modes, which gives $\mathcal{S}^{(1)}[V^{\dagger}O(\bx;U_{\ell})V]\le\log(2d)$ and therefore
\begin{equation}
    \mathcal{M}^{(1)}[O(\bx;U_{\ell})]=\max_{V}\mathcal{S}^{(1)}\big[V^{\dagger}O(\bx;U_{\ell})V\big]\le\log(2d).
\end{equation}
This completes the proof.
\end{proof}

\begin{proof}[Proof of Lemma~\ref{append:lem:dpos_reduction}]
Since $\mathbb{E}_{\bx}h(\bx)^2=1/(2d)$, the estimator can be rewritten as $\hat\mu=\mathbb{E}_{\bx}[\hat{f}(\bx) h(\bx)]/\mathbb{E}_{\bx}[h(\bx)^2]$. Employing the product form $f_{\ell}(\bx)=\mu_\ell h(\bx)$ from Lemma~\ref{append:lem:dpos_sse} and the Cauchy--Schwarz inequality, we have
\begin{align}
|\hat\mu-\mu_\ell|
= & \frac{\big|\mathbb{E}_{\bx}\big[(\hat{f}(\bx)-\mu_\ell h(\bx))\,h(\bx)\big]\big|}{\mathbb{E}_{\bx}[h(\bx)^2]}
\nonumber \\ 
\le & \frac{\sqrt{\mathbb{E}_{\bx}(\hat{f}(\bx)-f_{\ell}(\bx))^2}\,\sqrt{\mathbb{E}_{\bx}[h(\bx)^2]}}{\mathbb{E}_{\bx}[h(\bx)^2]}
\nonumber \\
= & \sqrt{\frac{\epsilon}{\mathbb{E}_{\bx}[h(\bx)^2]}}
=\sqrt{2d\,\epsilon}
\le\sqrt{2d\cdot\tfrac{\nu_f}{16}}
=\frac{|\mu_\ell|}{4},
\end{align}
where the second inequality uses the assumption $\mathbb{E}_{\bx}|\hat{f}(\bx)-f_{\ell}(\bx)|^2\le\epsilon$, and the last equality follows from $\nu_f=\mu_\ell^2/(2d)$. Hence $\hat\mu$ and $\mu_\ell$ share the same sign, and $|\hat\mu|\ge\tfrac34|\mu_\ell|\ge\tfrac14$ with the last inequality using $|\mu_\ell|\ge 1/3$ from Eq.~\eqref{append:eq:dpos_mu}.

It remains to bound the cost of estimating $\hat\mu$. Clipping $\hat{f}(\bx)$ to the interval $[-1,1]$ never increases the prediction error and ensures that each summand obeys $\hat f(\bx)h(\bx)\in[-1,1]$. Employing Hoeffding's inequality, a Monte Carlo average over $M=\mathcal{O}(d^2\log(1/\delta))$ i.i.d.\ inputs estimates $\mathbb{E}_{\bx}[\hat{f}(\bx) h(\bx)]$ to an additive accuracy $|\mu_\ell|/(16d)$, and hence estimates $\hat\mu$ to an accuracy $|\mu_\ell|/8$, with probability at least $1-\delta$. This completes the proof.
\end{proof}

\section{More numerical simulations}
\label{append:sec:more-numerical-simulations}

This section complements the numerical results presented in the main text with additional implementation details and supporting validation experiments. We first describe the baseline methods considered in the main text and introduce the evaluation metric used throughout our numerical studies in SI.~\ref{append:sec:numerical-methods}. We then investigate the predictive performance of the \DSE-guided classical surrogate under finite-shot constraints in SI.~\ref{append:sec:finite-shot-feature-selection}. Next, we provide additional numerical evidence for the scalability of the \DSE-guided surrogate in SI.~\ref{append:sec:controlled-dse}. Finally, in SI.~\ref{append:sec:application_1_dse}, we demonstrate that the \DSE-guided surrogate improves upon existing classical surrogates with provable efficiency guarantees for pretraining variational quantum algorithms.

\subsection{Implementation details of the numerical methods and evaluation metrics}
\label{append:sec:numerical-methods}
In this section, we elucidate the implementation of the numerical methods and the evaluation metrics used in our experiments. We first describe the learning-based surrogates and classical simulation methods considered in the benchmarks, and then specify the metric used to evaluate their performance.

\medskip
\noindent \textbf{Classical surrogates and simulators.}
We consider four numerical methods: the \DSE-guided surrogate $h_{\mathsf q}$, the Hamming-weight surrogate $h_{\mathsf c}$, the matrix-product-state (MPS) simulator, and the Pauli-path simulator (PPS), where the latter two simulators are representative for simulating low-entanglement and low-magic circuits, respectively. Their implementation details are described below.

\smallskip
\noindent \underline{\textit{\DSE-guided surrogate $h_{\mathsf{q}}$.}} Recall that constructing the \DSE-guided classical surrogate involves three steps: (i) data collection, (ii) feature selection, and (iii) model construction. In step (i), the classical learner collects a training dataset $\mathcal{T}=\{(\bxi,y^{(i)})\}_{i=1}^n$ containing $n$ examples, where each input $\bxi$ is independently sampled from the uniform distribution over $[-\pi,\pi]^d$ and $y^{(i)}$ is a statistical estimation of $f(\bxi)=\Tr(\rho_0 U(\bxi)^{\dagger}OU(\bxi))$ obtained from $m_y$ measurements. In step (ii), the classical learner selects a subset of features $\{\Phi_{\bomega}(\bx)\}_{\Lambda_\mathsf{q}}$ from the trigonometric expansion $f(\bx)=\sum_{\bomega\in\{0,\pm 1\}^d} \Phi_{\bomega}(\bx)\Tr(\rho_0 Q_{\bomega})$ defined in Eq.~\eqref{append:eq:tri_exp_T}, where the selected feature set is indexed by $\Lambda_{\mathsf{q}}\subset \{0,\pm 1\}^d$. To determine $\Lambda_{\mathsf q}$, we use the quantum subroutine introduced in SI.~\ref{append:sec:direct_mode_sampler} to prepare the $2d$-qubit state $\ket{\psi_f}=\sum_{\bomega\in \{0,\pm 1\}^d}\sqrt{\mathrm{p}(\bomega)}\ket{\bomega}$. The state is measured $m_f$ times in the computational basis $\ket{\bomega}\bra{\bomega}$. The learner then selects the $D$ most frequently observed modes $\{\tilde{\bomega}^{(i)}\}_{i=1}^D$ and defines $\Lambda_{\mathsf{q}}=\{\tilde{\bomega}^{(i)}\}_{i=1}^D$. In the ideal limit $m_f\rightarrow\infty$, the selected modes are ordered according to $\mathrm{p}(\tilde{\bomega}^{(1)})\ge \cdots \ge \mathrm{p}(\tilde{\bomega}^{(D)})$.

In step (iii), the learner uses the training dataset $\mathcal{T}$ and the selected feature map $\bm{\Phi}_{\Lambda{\mathsf q}}$ to construct the linear surrogate $h_{\mathsf{q}}(\bx,\hat{\bm{\mathrm{w}}})=\braket{\bm{\Phi}_{\Lambda_{\mathsf{q}}}(\bx),\hat{\bm{\mathrm{w}}}}$. The coefficient vector $\hat{\bm{\mathrm w}}$ is obtained by solving the ridge-regression problem
\begin{equation}
    \hat{\bm{\mathrm{w}}}=\arg \min_{\bm{\mathrm{w}}} \sum_{i=1}^n \left(h_{\mathsf{q}}(\bx,\bm{\mathrm{w}}) - y^{(i)}\right)^2 + \alpha \|\bm{\mathrm{w}}\|,
\end{equation}
where $\alpha$ is the regularization parameter. We set $\alpha=1$ throughout the numerical experiments.

To distinguish between ideal and finite-shot implementations, we use $h_{\mathsf q}$ to denote the \DSE-guided surrogate constructed in the ideal limit $m_y=m_f=\infty$, and $\bar h_{\mathsf q}$ to denote the surrogate constructed using finite values of $m_y$ and $m_f$. In the numerical comparisons presented below, $h_{\mathsf q}$ is primarily used to characterize the predictive capability of the \DSE-guided surrogate under a prescribed inference-time computational budget. Unless explicitly stated otherwise, the one-time cost of constructing $h_{\mathsf q}$ during the training stage is not included in the inference-cost comparison.

\smallskip
\noindent \underline{\textit{Hamming-weight surrogate $h_{\mathsf{c}}$.}} The Hamming-weight surrogate $h_{\mathsf c}$ is a representative classical surrogate proposed in previous studies and equipped with provable efficiency guarantees~\cite{du2025efficient,liao2025demonstration}. Its construction follows the same three-step procedure as that of $h_{\mathsf q}$, but employs a different feature-selection rule in stage (ii). In particular, the Hamming-weight surrogate has the linear form $h_{\mathsf{c}}(\bx,\hat{\bm{\mathrm{w}}})=\braket{\bm{\Phi}_{\Lambda_{\mathsf{c}}}(\bx),\hat{\bm{\mathrm{w}}}}$, where $\hat{\bm{\mathrm{w}}}$ is obtained using the same ridge-regression procedure described above. The feature vector $\bm{\Phi}_{\Lambda_{\mathsf{c}}}(\bx)$ is constructed by prioritizing modes $\bomega$ with small Hamming weight $\|\bomega\|_0$. For an integer threshold $\Gamma$, define $
\Lambda(\Gamma)=\{\bomega\in{0,\pm1}^{d}:
|\bomega|_0\leq\Gamma\}$. For a prescribed feature budget (D), there may not exist an integer $\Gamma$ satisfying $|\Lambda(\Gamma)|=D$. We therefore choose $\Gamma'$ such that $|\Lambda(\Gamma')|\le D \le |\Lambda(\Gamma'+1)|$. To construct $\Lambda_{\mathsf{c}}$ of cardinality $D$, we include all modes in $\Lambda(\Gamma')$ and then uniformly sample $D-|\Lambda(\Gamma')|$ additional modes from the set $\{\bomega:\|\bomega\|_0 = \Gamma'+1\}$. Notably, this feature selection procedure is purely classical and ignores the effect of the structure-dependent quantity $\Tr(\rho_0Q_{\bomega})$ in the trigonometric expansion of $f(\bx)$.

\smallskip
\noindent \underline{\textit{Matrix-product-state simulator.}} 
Matrix-product-state methods are tensor-network-based approaches that are particularly effective for simulating quantum circuits whose intermediate states exhibit limited entanglement. To compare $h_{\mathsf q}$ with a representative tensor-network simulator, we implement an MPS baseline that independently contracts each parameterized circuit $U(\bx)$ using the PennyLane tensor-network backend~\cite{bergholm2018pennylane}. 

The maximum MPS bond dimension $\chi$ is set equal to the feature dimension $D$ used by the classical surrogates. This choice provides a controlled comparison between the inference budgets of the learning-based and simulation-based methods. Evaluating a surrogate with $D$ selected trigonometric features requires at most $\mathcal{O}(dD)$ arithmetic operations per input when the features are evaluated directly. By contrast, simulating an $N$-qubit circuit of depth $L$ using an MPS generally requires a computational cost that grows at least linearly with $NL$ and polynomially with $\chi$. Consequently, setting $\chi=D$ ensures that the MPS baseline is not assigned a smaller nominal truncation budget than the surrogate while still reflecting the substantially different computational structures of the two methods.

\smallskip
\noindent \underline{\textit{Pauli-path simulator.}} 
As a complementary operator-based simulation method, PPS estimates expectation values $f(\bx)$ by propagating the measured Pauli observable backward through the circuit in the Heisenberg picture~\cite{rudolph2026pauli}. Clifford gates map each Pauli string to a single Pauli string, whereas non-Clifford and parameterized rotation gates can map one Pauli string to a linear combination of multiple Pauli strings. Repeated propagation through such gates therefore produces a branching collection of weighted Pauli paths whose size can grow exponentially with the number of non-Clifford gates.

Different PPS variants control this growth using different truncation strategies. In our implementation, we employ a top-$K$ truncation rule. After each circuit-layer update, the simulator retains at most the $K$ Pauli paths with the largest coefficient magnitudes and discards the remaining paths. This truncation keeps the memory and runtime costs explicitly controlled and enables a direct comparison with classical surrogates under a fixed representation budget. The hyperparameter $K$ specifies the maximum number of Pauli paths retained during propagation. Throughout our numerical comparisons, we set $K=D$ to compare PPS and the predictive surrogates under the same representation budget and to demonstrate the superior predictive performance of $h_{\mathsf{q}}$ over PPS.

\medskip
\noindent \textbf{Evaluation metric.}
To evaluate the prediction performance of classical surrogates and simulators under a common criterion, we employ the coefficient of determination ($R^2$) on the test dataset $\{\bx^{(i)},y^{(i)}\}_{i=1}^{n_{\mathrm{test}}}$
\begin{equation}
R^2=1-\frac{\sum_{i=1}^{n_{\mathrm{test}}}
[y_i-h(\bx^{(i)})]^2}
{\sum_{i=1}^{n_{\mathrm{test}}}(y_i-\overline y)^2},
\qquad
\overline y=\frac{1}{n_{\mathrm{test}}}
\sum_{i=1}^{n_{\mathrm{test}}}y_i,
\label{append:eq:r2}
\end{equation}
where $y^{(i)}=\Tr[\rho_0U(\bx^{(i)})^{\dagger}OU(\bx^{(i)})]$ is the exact expectation value and $h(\bx^{(i)})$ denotes the surrogate prediction or simulator estimate. An ideal predictor achieves $R^2=1$, while $R^2=0$ corresponds to the constant predictor $\overline y$ (the mean of the test targets). Negative values of $R^2$ indicate a larger squared prediction error than this constant baseline. Since all negative values represent performance worse than the baseline, we display them as zero in the main-text figures for visualization purposes.

To provide a more nuanced assessment across the full range of prediction accuracy, we also introduce another bounded evaluation metric derived from $R^2$, which is given by
\begin{equation}
\widetilde{R}^2=\frac{1}{2-R^2}.
\label{append:eq:bounded-nmse-score}
\end{equation}
Since $\widetilde{R}^2$ is a monotonic transformation of $R^2$, it preserves the ranking of different methods while mapping the entire range $R^2\in(-\infty,1]$ onto the normalized interval $\widetilde{R}^2\in(0,1]$. Consequently, even arbitrarily negative values of $R^2$ are assigned finite scores within this bounded range, making the metric more suitable for comparing methods with poor predictive performance. Under this definition, exact prediction yields $\widetilde{R}^2=1$, the mean predictor gives $\widetilde{R}^2=1/2$, and values below $1/2$ indicate performance worse than the mean predictor.

\subsection{Numerical results for \DSE-guided classical surrogates under finite-shot constraints}
\label{append:sec:finite-shot-feature-selection}

The numerical results in the main text are obtained in an idealized setting
that neglects the statistical errors arising from a finite number of quantum
measurements. In this section, we complement the random-circuit experiments conducted in Fig.~\ref{fig:exp_results_1}a by
incorporating two distinct sources of measurement error: the finite-shot error in
the training labels and the finite-sampling error in identifying the
\DSE-guided features with the quantum subroutine. Treating these two resources
separately allows us to distinguish imperfect feature identification from
label noise. Specifically, we examine how prediction accuracy varies with the
numbers of measurements used for feature identification and label collection,
compare practical \DSE-guided feature selection with the ideal \DSE-guided,
Hamming-weight, and random Fourier strategies, and determine how a fixed total
measurement budget should be divided between feature identification and
training-label collection.

\smallskip
\noindent \textbf{Experimental setting.}
We consider the random-circuit setting used in Fig.~\ref{fig:exp_results_1}a. In particular, the system size is fixed at $N=6$ with circuit depth $L=66$. The number of rotation gates and $T$ gates contained in the circuit $U(\bx)$ are set as $d=6$ and $G_t=6$, respectively, while CNOT gates are inserted independently with probability $\gamma=0.1$. For surrogate construction, the training set contains $n=100$ examples and the feature dimension is fixed at $D=8$. All surrogates are evaluated on an independent test set of $n_{\mathrm{test}}=200$ examples. We generate three random quantum circuits and perform $20$ independent measurement repetitions for each circuit instance.

To investigate the effect of measurement noise in the training labels, we vary the number of measurements used to estimate each label $y^{(i)}$ according to
\begin{equation}
 m_{\mathrm y}\in\{1,2,4,8,16,32,64,128,256,512,1024,\infty\},
 \label{append:eq:label-shot-grid}
\end{equation}
where $m_{\mathrm y}=\infty$ denotes exact expectation values. For the \DSE-guided surrogate $\bar h_{\mathsf q}$, we also vary the number of measurements used in the quantum subroutine for feature identification with
\begin{equation}
 m_{\mathrm f}\in\{2^4,2^5,\ldots,2^{18}\}.
 \label{append:eq:feature-shot-grid}
\end{equation}
Notably, the feature-identification budget $m_{\mathrm f}$ is required only for constructing $\bar h_{\mathsf q}$. In contrast, the ideal \DSE-guided surrogate $h_{\mathsf q}$ assumes the exact feature-identification, corresponding to the scenario of $m_{\mathrm f}=\infty$, while the Hamming-weight surrogate $h_{\mathsf c}$ and the random Fourier surrogate $h_{\mathsf r}$ do not require the quantum feature-identification procedure.

For the fixed-budget comparison, we define the total measurement budget as
\begin{equation}
 m_{\mathrm t}=m_{\mathrm f}+n\cdot m_{\mathrm y},
 \label{append:eq:total-shot-budget}
\end{equation}
and consider
\[
m_{\mathrm t}\in
\{10^3,3\times10^3,10^4,3\times10^4,
10^5,3\times10^5,10^6\}.
\]
For $\bar h_{\mathsf q}$, fractions
$m_{\mathrm f}/m_{\mathrm t}\in\{0.05,0.10,0.25,0.50\}$
are allocated to feature identification, with the remaining budget distributed uniformly across the $n$ training labels. The Hamming-weight and random Fourier surrogates devote the entire measurement budget to label estimation.

\begin{figure*}[t]
 \centering
 \includegraphics[width=\textwidth]{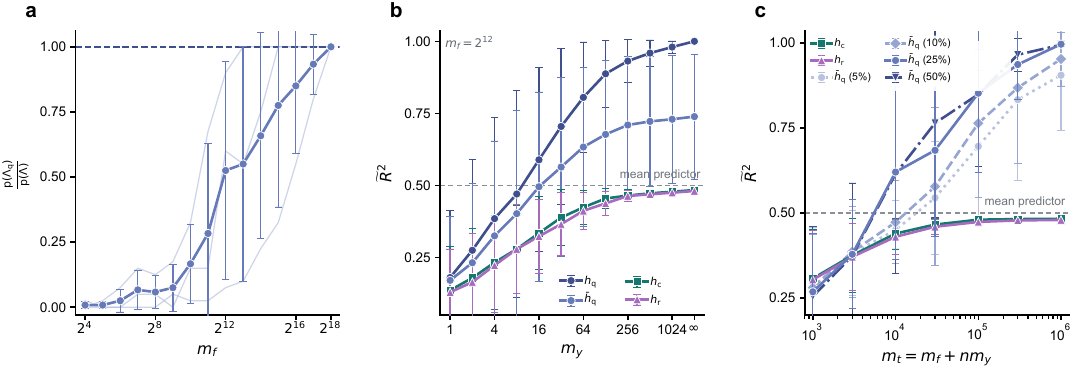}
 \caption{\textbf{Prediction performance of \DSE-guided classical surrogate under finite-shot constraints.}
 \textbf{a.} Probability mass captured by the quantum subroutine using varying number of measurements $m_f=\{2^4,2^8,2^{12},2^{16},2^{18}\}$. Light lines in show individual circuit-seed averages.
 \textbf{b.} Prediction accuracy of classical surrogates with varying number of measurements used for estimating per training label $m_{y}\in \{1,4,16,64,256,1024,\infty\}$, where $m_{y}=\infty$ refers to using exact training labels. The surrogate $\bar{h}_{\mathsf{q}}$ is constructed with setting $m_{\mathrm f}=2^{12}$.
 \textbf{c.} Prediction accuracy under a fixed total measurement budget $m_t=m_{f}+nm_{y}$. Percentages specify the fraction of $m_t$ allocated to feature identification for $\bar{h}_{\mathsf q}$. The dashed horizontal line in \textbf{b,c} denotes the train-set mean predictor, $A=1/2$. Markers and error bars report the mean and standard deviation over three random quantum circuits after averaging $20$ measurement repetitions within each circuit.}
 \label{fig:supp-finite-shot}
\end{figure*}

\smallskip
\noindent \textbf{Numerical results.}
Fig.~\ref{fig:supp-finite-shot}a first isolates the finite-sampling error arising from feature identification and evaluates the ability of the quantum subroutine to recover important trigonometric features. We quantify the quality of the identified feature set using the captured probability mass
\begin{equation}
    \frac{\mathrm{p}(\Lambda_{\mathsf{q}})}{\mathrm{p}(\Lambda)}=\frac{\sum_{\bomega\in \Lambda_{\mathsf{q}}}\mathrm{p}(\bomega)}{\sum_{\bomega\in \{0,\pm 1\}} \mathrm{p}(\bomega)},
\end{equation}
where $\mathrm{p}(\bomega)= \tilde{\mathrm{p}}(\bomega)/\sum_{\bomega\in \{0,\pm 1\}^d}\tilde{\mathrm{p}}(\bomega)$ with $\tilde{\mathrm{p}}(\bomega)=\mathbb{E}_{\bx} \Phi_{\bomega}(\bx)^2\Tr(\rho_0Q_{\bomega})^2$ is the probability distribution induced by the trigonometric expansion of $f(\bx)$ in Eq.~\eqref{append:eq:tr_equal_fbx}. Here, $p(\Lambda)=1$ and $\Lambda_{\mathsf{q}}$ denotes the set of distinct modes sampled by employing the quantum subroutine.
The results show that increasing the number of feature-identification measurements $m_f$ substantially improves the recovery of important modes. In particular, the captured probability mass increases from $0.167$ at $m_f=2^{10}$ to $0.93$ at $m_f=2^{17}$. It reaches unity at $m_f=2^{18}$, indicating that all modes with nonzero probability, $\{\bomega: \mathsf{p}(\bomega) >0\}$, have been successfully identified. These results demonstrate that the quantum subroutine can reliably recover the important trigonometric features as the measurement budget increases.

Fig.~\ref{fig:supp-finite-shot}b fixes the feature-identification budget at $m_f=2^{12}$ and varies the number of measurements $m_y$ used to estimate each training label. The predictive performance of all methods improves with $m_y$, showing that label-estimation noise dominates in the low-shot regime. Nevertheless, the \DSE-guided surrogates consistently outperform the Hamming-weight and random-feature baselines once $m_y$ is sufficiently large. In particular, with exact training labels, $h_{\mathsf q}$ and $\bar h_{\mathsf q}$ achieve $\widetilde{R}^{2}=1$ and $0.74$, respectively, whereas $h_{\mathsf c}$ and $h_{\mathsf r}$ remain below $0.5$. Thus, although $m_f=2^{12}$ does not recover all important modes, the selected features still retain a clear predictive advantage, while the gap between $h_{\mathsf q}$ and $\bar{h}_{\mathsf q}$ reflects the cost of imperfect feature identification.

Finally, Fig.~\ref{fig:supp-finite-shot}c examines how a fixed total measurement budget $m_t$ should be allocated between feature identification and label estimation. At small $m_t$, the additional cost of feature identification can outweigh its benefit, causing the conventional baselines to perform slightly better. As $m_t$ increases, however, the practical \DSE-guided surrogate exhibits a clear advantage, with allocating $25\%$--$50\%$ of the budget to feature identification providing the best overall performance. For instance, at $m_t=10^{4}$, these two allocation strategies achieve $\widetilde{R}^{2}=0.62$ and $0.61$, respectively, compared with $0.44$ for $h_{\mathsf c}$ and $0.43$ for $h_{\mathsf r}$. These results show that, once the total measurement budget is sufficient to support both feature and label estimation, the \DSE-guided surrogate substantially outperforms the Hamming-weight and random-feature baselines.

\subsection{More numerical results for the scalability of the \DSE-guided classical surrogates}
\label{append:sec:controlled-dse}

This subsection complements the scalability analysis in Fig.~\ref{fig:scalability} by further separating the effect of the system sizes and the number of rotation gates, the latter of which controls the \DSE of the considered circuit. Using the same structured cluster-circuit family, we perform two controlled comparisons. First, we vary $N$ while keeping \DSE fixed to determine whether increasing the system size alone reduces the prediction accuracy. Second, we vary \DSE at fixed $N$ to examine how the prediction accuracy changes as the mode distribution becomes more complex. Together, these experiments test whether the performance of $h_{\mathsf q}$ is governed primarily by \DSE rather than by the number of qubits. Throughout this analysis, we use the tractable quantity $\mathcal{S}^{(1)}$ as a proxy for \DSE, denoted by $\mathcal{M}^{(1)}$.

\smallskip
\noindent \textbf{Experimental setting.}
We consider the structured cluster-circuit family introduced in Fig.~\ref{fig:scalability}. For the cluster observable with window size $k=5$, the nonzero modes $\{\bomega:\mathrm{p}(\bomega)>0\}$
have equal probabilities $\mathrm{p}(\bomega)$ and The number of such modes is
\begin{equation}
 M_{\mathrm{nz}}=(d-k+1)2^k,\qquad
 \mathcal{S}^{(1)}=\log_2 M_{\mathrm{nz}}.
 \label{append:eq:cluster-dse}
\end{equation}
In the fixed-\DSE experiment, we set $d=12$, for which
$M_{\mathrm{nz}}=256$ and $\mathcal{S}^{(1)}=8$, and vary the number of qubits as
$N\in\{20,30,40,50,60,80\}$. Because increasing $N$ only adds qubits that do not alter the target function or its mode distribution, both $f(\bx)$ and $\mathrm{p}(\bomega)$ remain unchanged. In the fixed-$N$ experiment, we
set $N=40$ and vary the number of rotation gates as $d\in\{6,9,12,15,18,21\}$, corresponding to
$\mathcal{S}^{(1)} \in\{6.00,7.32,8.00,8.46,8.81,9.09\}$.

For both experiments, we consider three combinations of training-set size $n$ and feature dimension $D$, i.e., $(n,D)=(100,50),(300,200)$, and $(500,300)$. When $D>M_{\mathrm{nz}}$, all nonzero modes are retained, and the effective feature dimension is therefore $M_{\mathrm{nz}}$. Each configuration is averaged over $20$ independently sampled training datasets and evaluated on $1000$ test inputs with exact labels.

\smallskip
\noindent \textbf{Numerical results.}
As shown in Fig.~\ref{append:fig:controlled-dse}, the two controlled experiments clearly separate the dependence of the prediction accuracy on \DSE from its dependence on the number of qubits. At fixed $\mathcal{S}^{(1)}=8$, the performance of $h_{\mathsf q}$ exhibits no systematic variation with (N). In particular, for all three choices of $(n,D)$, the fitted slopes of the values of $R^2$ with respect to $N$ are close to zero. In contrast, at fixed $N=40$, the accuracy decreases monotonically or near-monotonically as the \DSE proxy $\mathcal{S}^{(1)}$ increases. 
For example, at $(n,D)=(300,200)$, the prediction accuracy $R^2$ decreases from $1$
at $\mathcal{S}^{(1)}=6$ to $0.32$ at $\mathcal{S}^{(1)}=8$ and becomes zero for the three largest values of $\mathcal{S}^{(1)}$. At $(n,D)=(500,300)$, the model remains exact through
$\mathcal{S}^{(1)}=8$, where the requested $D=300$ already contains all $256$ nonzero modes, and then degrades as the mode support $\{\bomega:\mathrm{p}(\bomega)>0\}$ grows beyond the feature budget $D$.
These controlled comparisons show that enlarging $N$ does not reduce accuracy when the mode distribution $\mathrm{p}(\bomega)$ is fixed, whereas increasing \DSE at fixed $N$ increases the feature and data requirements.

\begin{figure}[t]
\centering
\includegraphics[width=0.96\linewidth]{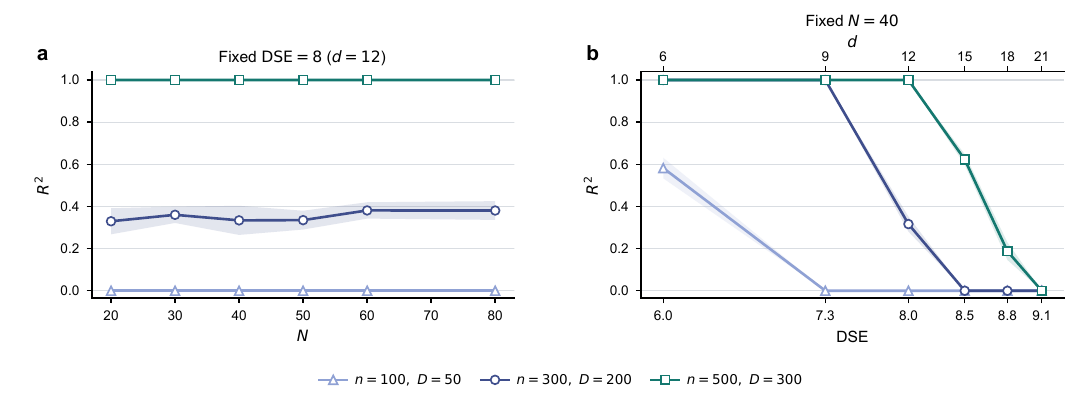}
\caption{\textbf{Controlled validation of \DSE as the dominant scale of the
structured surrogate benchmark.} \textbf{a}, Prediction accuracy at fixed
$d=12$ and $\mathcal{S}^{(1)}=8$ as the number of qubits $N$ is varied, where $\mathcal{S}^{(1)}$ is the proxy of \DSE. The target function
and its mode distribution remain unchanged. \textbf{b}, Prediction accuracy at
fixed $N=40$ as $d$, and hence \DSE, is increased. The three curves correspond
to $(n,D)=(100,50),(300,200)$, and $(500,300)$. Points and shaded bands denote
the mean and 95\% confidence interval over $20$ independently sampled datasets. Each model is evaluated on $1000$ exact test inputs. }
\label{append:fig:controlled-dse}
\end{figure}

\subsection{Application of \DSE-guided surrogates: pre-training variational quantum eigensolvers (VQEs)}
\label{append:sec:application_1_dse}
While the classical surrogates have wide applications as studied in previous studies \cite{du2025efficient, liao2025demonstration}, here we focus on a major application of \DSE-guided surrogates for enhancing variational quantum algorithms by substantially reducing the quantum resource demands. In the following, we first introduce how to use the classical surrogate to pre-train variational quantum eigensolvers, and then introduce the numerical settings and results.

\medskip
\noindent\textbf{Implementations of pre-training VQEs with \DSE-guided classical surrogate.} We first recap the mechanism of VQEs. Given an initial state $\rho_0$ and a target Hamiltonian $\mathsf{H}$, the objective of VQEs is to find the ground state of $\mathsf{H}$ by optimizing the following expectation function
\begin{equation}
    \min_{\bx} f(\bx, \mathsf{H}):=\min_{\bx}\Tr(\rho_0 U(\bx)^{\dagger} \mathsf{H}U(\bx)),
\end{equation}
where $U(\bx)$ is a variation quantum circuits. To achieve this, a common strategy is to use a gradient descent optimizer to iteratively update the tunable parameters $\bx$ to minimize $f(\bx, \mathsf{H})$. Specifically, at the $t$-th iteration, the updating 
rule yields
\begin{equation}
    \bx^{(t+1)}=\bx^{(t)}-\eta \nabla_{\bx} f(\bx^{(t)}, \mathsf{H}),
\end{equation}
where $\eta$ is the learning rate. This process continues until the loss function converges to a predefined condition. However, implementing VQEs on quantum processors incurs significant measurement overhead, primarily due to the need for estimating expectation values and gradients during the optimization process. As indicated in previous studies \cite{du2025efficient,liao2025demonstration}, a well-optimized classical surrogate $h_{\mathsf{q}}$ can alleviate the expensive measurement overhead via pre-training. Instead of directly optimizing $f(\bx^{(t)}, \mathsf{H})$, the target parameters can be obtained by minimizing $h_{\mathsf{q}}(\bx,\mathsf{H})$, i.e.,
\begin{equation}
    \min_{\bx} h_{\mathsf{q}}(\bx,\mathsf{H})
\end{equation}
This process is entirely classical and does not involve any access to quantum processors. Given access to the optimized $\hat{\bx}$, one can further minimize $f(\hat{\bx},\mathsf{H})$ on quantum devices for additional fine-tuning.

\medskip
\noindent\textbf{Experimental settings.}
Here, we consider the 1D open-boundary transverse field Ising models (TFIMs), with the formal expression as
\begin{equation}
 \mathsf{H} = -\sum_{i=1}^{N-1}Z_iZ_{i+1}
 -\sum_{i=1}^{N}X_i.
\end{equation}
The initial state is $\rho_0=\ket{0}\bra{0}^{\otimes N}$ and the adopted ansatz in VQE employs the Hardware efficient ansatz of the form
\begin{equation}
    U(\bx)=\prod_{\ell=1}^L \left[ \prod_{i=1}^{N-1}\CNOT_{i,i+1} \prod_{i=1}^N \mathrm{RY}(\bx_{\ell,i})\right]
\end{equation}
where the input dimension is $d=LN$. 

The TFIM considered here is set to $N=6$, $J=g=1$. The circuit depth of $U(\bx)$ is set as $L=2$. For the construction of $h_{\mathsf{q}}$, the training data size and the feature dimension is set as 
\begin{equation}
 n\in\{50,100,200,400,1600\},
 \qquad
 D\in\{8,16,32,128,512\}.
\end{equation}
The total budget of measurement for constructing surrogates is set as 
\begin{equation}
 m_t\in\{2^8,2^{10},2^{12},2^{14},2^{16}\}.
\end{equation}

For the optimization of VQEs, all methods start from the same initial parameter point, and we use Adam \cite{kingma2017adam} as the optimizer with learning rate $0.08$. To examine the ability of the optimized surrogate $h$ in estimating the ground state energy, we use the normalized deviation as a metric to quantify its performance
\begin{equation}
 \mathfrak{R}(\bx) =
 \frac{\left|f(\bx)-E_0\right|}
 {E_{\max}-E_0},
 \label{append:eq:normalized-vqe-error}
\end{equation}
where $E_0$ and $E_{\max}$ are the minimum and maximum eigenvalues of $\mathsf{H}$.

\medskip
\noindent\textbf{Numerical results.}
Figure~\ref{fig:append:vqe-pretraining}a presents the prediction performance $\widetilde{R}^2=(2-R^2)^{-1}$ of the ideal \DSE-guided surrogate $h_{\mathsf{q}}$ and Hamming-weight surrogate $h_{\mathsf{c}}$ across different feature dimensions $D$ and training-set sizes $n$. The \DSE-guided surrogate achieves strong predictive performance with only moderate resources, namely $h_{\mathsf q}$ reaches $\widetilde{R}^{2}=0.87$ using $D=32$ and $n=50$. By contrast, even with $D=512$ and $n=1600$, $h_{\mathsf c}$ remains below $\widetilde{R}^{2}=0.5$, indicating performance worse than the mean predictor. These results demonstrate the superior performance of $h_{\mathsf q}$ over $h_{\mathsf c}$ and highlight the importance of identifying the most relevant features when constructing predictive surrogates.

Figure~\ref{fig:append:vqe-pretraining}b further examines whether this advantage persists when feature identification and training-label estimation share a fixed total measurement budget. With $25\%$ of the budget allocated to feature identification, the performance of the finite-shot surrogate $\bar h_{\mathsf q}$ improves from $\widetilde{R}^{2}=0.41$ at $m_t=2^{8}$ to $0.93$ at $m_t=2^{16}$. At the largest budget $m_t=2^{18}$, $h_{\mathsf c}$ reaches only $0.30$ with $n=50$ and $0.46$ with $n=200$. Thus, although increasing the number of training samples improves $h_{\mathsf c}$, it does not recover the advantage obtained by allocating part of the measurement budget to \DSE-guided feature identification.

Finally, we examine whether the prediction advantage of the \DSE-guided surrogate translates into improved VQE optimization. As shown in Fig.~\ref{fig:append:vqe-pretraining}(c), increasing $m_f$ from $2^8$ to $2^{12}$ improves the test accuracy of $\bar h_{\mathsf q}$ from $0.561$ to $0.898$. Correspondingly, the normalized energy error after 80 classical pre-training iterations decreases from $\mathfrak{R}=0.279$ to $0.044$. By contrast, pre-training with $h_{\mathsf c}$ yields $\mathfrak{R}=0.506$, providing a substantially less effective initialization. We then fine-tune the pre-trained parameters by optimizing the circuit objective $f$ for 20 additional iterations. The final errors are $0.070$, $0.036$, and $0.026$ for $m_f=2^8$, $2^{10}$, and $2^{12}$, respectively, compared with $0.171$ for $h_{\mathsf c}$. In particular, the $m_f=2^{12}$ protocol approaches the accuracy of directly optimizing $f$ for 100 iterations, which yields $\mathfrak{R}=0.024$.

We further compare the quantum measurement costs of these two procedures. The ansatz contains $d=12$ independent parameters, so the parameter-shift rule requires $2d=24$ energy evaluations per gradient step. With $m_y=64$ measurements per energy evaluation, direct optimization for 100 iterations requires $100(2d)m_y=153{,}600$ circuit executions. The \DSE-guided protocol instead requires $m_f+nm_y=10{,}496$ executions to construct the surrogate at $m_f=2^{12}$ and $n=100$, followed by $20(2d)m_y=30{,}720$ executions for circuit fine-tuning. Its total measurement cost is therefore $41{,}216$, while the preceding 80 surrogate-optimization iterations are performed entirely classically. Under this parameter-shift accounting, \DSE-guided pre-training reduces the total quantum measurement cost by $73.2\%$, while attaining nearly the same ground-state accuracy as direct VQE optimization. Together, these results show that the improved landscape approximation enabled by the \DSE-guided surrogate provides a high-quality initialization and substantially reduces the quantum measurement cost required for VQE optimization.

\begin{figure}[t]
\centering
\includegraphics[width=0.96\linewidth]{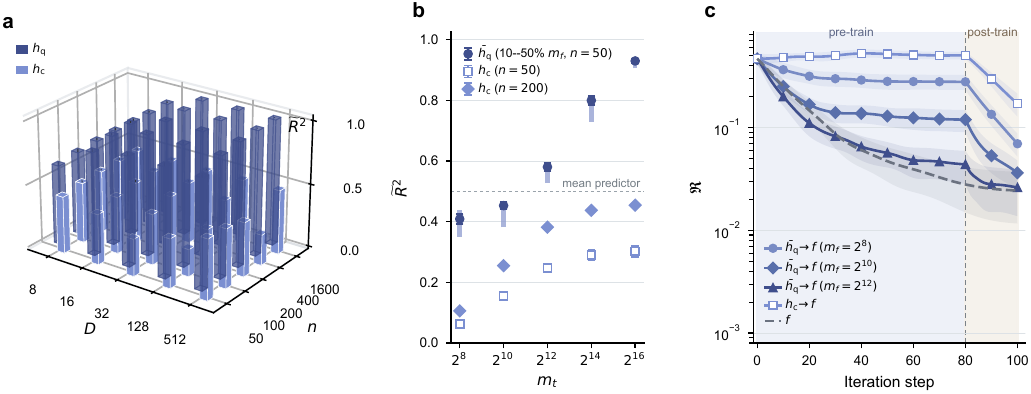}
\caption{\textbf{\DSE-guided surrogate pre-training for a variational quantum
eigensolver.}
\textbf{a.} Prediction accuracy $\widetilde{R}^{2}$ of the ideal \DSE-guided surrogate $h_{\mathsf q}$ and the Hamming-weight surrogate $h_{\mathsf c}$ as functions of the requested feature dimension $D\in\{8,16,32,128,512\}$ and training size $n\in\{50,100,200,400,1600\}$. Training labels are collected using $m_y=2^6$ measurements, and $h_{\mathsf q}$ uses the ideal feature ordering corresponding to $m_f=\infty$. \textbf{b.} Prediction accuracy under the fixed total measurement budgets $m_t\in\{2^8,2^{10},2^{12},2^{14},2^{16}\}$ at $D=32$. For
$\bar h_{\mathsf q}$, the circles show the 25\% allocation to feature identification and the light vertical ranges span allocations from 10\% to 50\%, with $n=50$. The Hamming-weight results use $n=50$ or $n=200$ and allocate the entire budget to training-label measurements. The dashed horizontal line denotes the test-set mean predictor.
\textbf{c.} Normalized deviation $\mathfrak{R}$ during VQE optimization with $n=100$, $m_y=2^6$, and $D=32$. The practical \DSE-guided surrogates use $m_f=\{2^8, 2^{10}, 2^{12}\}$. The first 80 iterations optimize the classical surrogate (pre-train), after which the true circuit objective is optimized for 20 iterations (post-train). The direct baseline optimizes the circuit objective for all 100 iterations from the same initial parameters. Bars and points in \textbf{a,b} show means over 20 seeds after averaging five measurement repetitions within each seed; error bars denote 95\% confidence intervals. Lines in \textbf{c} show means over 20 seeds, and shaded regions denote 95\% confidence intervals.}
\label{fig:append:vqe-pretraining}
\label{append:fig:tfim}
\end{figure}  
 
\end{document}